\documentclass[a4paper,11pt]{article}
\usepackage{a4wide}
\usepackage{microtype}
\usepackage{amsmath}  
\usepackage{amsthm}  
\usepackage{amssymb}  
\usepackage{mathtools}
\usepackage{enumitem}
\usepackage[numbered]{bookmark}
\usepackage{hyperref}
\hypersetup{colorlinks=true,citecolor=blue,linkcolor=blue,urlcolor=blue}
\usepackage[normalem]{ulem}
\usepackage{xspace}
\usepackage{xcolor}
\usepackage{multicol}

\title{Pliability and Approximating Max-CSPs\thanks{An extended abstract of this work appeared in the Proceedings of SODA'21~\cite{rwz21:soda}. Work done while Miguel Romero and Marcin Wrochna were at the University of Oxford. Stanislav \v{Z}ivn\'y was supported by a Royal Society University Research Fellowship. This project has received funding from the European Research Council (ERC) under the European Union's Horizon 2020 research and innovation programme (grant agreement No 714532). The paper reflects only the authors' views and not the views of the ERC or the European Commission. The European Union is not liable for any use that may be made of the information contained therein. This work was also supported by UKRI EP/X024431/1. For the purpose of Open Access, the authors have applied a CC BY public copyright licence to any Author Accepted Manuscript version arising from this submission. All data is provided in full in the results section of this paper.}}
 
\author{
 	Miguel Romero\\
  Pontificia Universidad Cat\'{o}lica de Chile, Chile\\
  \texttt{mgromero@uc.cl}
 	\and
 	Marcin Wrochna\\
 	University of Warsaw, Poland\\
 	\texttt{m.wrochna@mimuw.edu.pl}
 	\and
 	Stanislav \v{Z}ivn\'{y}\\
 	University of Oxford, UK\\
 	\texttt{standa.zivny@cs.ox.ac.uk}
}
 
\date{\today}

\newtheorem{theorem}{Theorem}
\newtheorem*{theorem*}{Theorem}
\numberwithin{theorem}{section} 
\newtheorem{lemma}[theorem]{Lemma}
\newtheorem*{lemma*}{Lemma}
\newtheorem{corollary}[theorem]{Corollary}
\newtheorem*{corollary*}{Corollary}
\newtheorem{proposition}[theorem]{Proposition}
\newtheorem*{proposition*}{Proposition}

\newtheorem{observation}[theorem]{Observation}
\newtheorem*{observation*}{Observation}
\theoremstyle{definition}
\newtheorem{remark}[theorem]{Remark}
\newtheorem{example}[theorem]{Example}
\newtheorem{definition}[theorem]{Definition}
\newtheorem{conjecture}[theorem]{Conjecture}
\newtheorem{question}[theorem]{Question}

\newcommand{\eps}{\varepsilon}
\newcommand{\tuple}[1]{\mathbf{#1}}
\DeclareMathOperator{\val}{value}
\newcommand{\tup}[1]{\text{tup}(#1)}
\DeclareMathOperator{\ar}{ar}
\newcommand{\defeq}{\vcentcolon=}
\newcommand{\Oh}{\ensuremath{\mathcal{O}}}
\def\qplus{\mathbb{Q}_{\geq 0}}

\newcommand{\toset}[1]{\text{Set}(#1)}

\newcommand\overcasts{\mathbin{\succeq}}
\def\structA{\mathbb{A}}
\def\structB{\mathbb{B}}
\def\structC{\mathbb{C}}
\def\A{\mathcal{A}}
\def\B{\mathcal{B}}
\def\Gg{\mathcal{G}}
\newcommand{\opt}[1]{\mathrm{opt}(#1)}

\newcommand{\optfrac}[2]{\mathrm{opt}_{#1}(#2)}
\newcommand{\AOver}[2]{\A_{#1}^{#2}}

\newcommand{\ZZ}{\mathbb{Z}}
\newcommand{\NN}{\mathbb{N}}
\newcommand{\QQ}{\mathbb{Q}}
\newcommand{\RR}{\mathbb{R}}
\def\Qn{\QQ_{\geq 0}}
\DeclareMathOperator*{\EX}{\mathbb{E}}
\def\Pp{\mathcal{P}}

\newcommand{\dopt}{\mathrm{d}_{\mathrm{opt}}}
\DeclareMathOperator{\dedit}{d_{1}}
\DeclareMathOperator{\Img}{\mathbb{I}m}
\DeclareMathOperator{\tw}{tw}
\DeclareMathOperator{\pw}{pw}
\DeclareMathOperator{\td}{td}
\DeclareMathOperator{\mad}{mad}
\newcommand{\Hadwiger}{\text{Hadwiger}}
\DeclareMathOperator{\indeg}{in-deg}
\DeclareMathOperator{\param}{p}
\DeclareMathOperator{\id}{id}
\DeclareMathOperator{\cc}{cc}
\DeclareMathOperator{\size}{size}

\DeclareMathOperator{\Gaifman}{G}

\newcommand\vvec[2]{v^{(#1)}\if\relax\detokenize{#2}\relax\else_{#2}\fi}
\newcommand\wvec[2]{w^{(#1)}\if\relax\detokenize{#2}\relax\else_{#2}\fi}

\newcommand\CSP{Max-CSP\xspace}
\newcommand\CSPs{Max-CSPs\xspace}
\newcommand\HOM{\ensuremath{\mbox{Max-Hom}}\xspace}
\newcommand\bHOM[2]{\HOM(#1,#2)}

\newcommand\aCSP[1]{\ensuremath{\mbox{Max-}#1\mbox{-CSP}}}
\newcommand\aCSPs[1]{\ensuremath{\mbox{Max-}#1\mbox{-CSPs}}\xspace}
\newcommand\bCSP{\aCSP{2}}
\newcommand\bCSPs{\aCSPs{2}}
\newcommand\rCSP{\aCSP{r}}
\newcommand\rCSPs{\aCSPs{r}}

\newcommand\gCSPs{graph \CSPs}

\begin{document}
\maketitle

\begin{abstract}
We identify a sufficient condition, \emph{treewidth-pliability}, that gives a polynomial-time algorithm for an arbitrarily good approximation of the optimal value in a large class of Max-2-CSPs parameterised by the class of allowed constraint graphs (with arbitrary constraints on an unbounded alphabet). Our result applies more generally to the maximum homomorphism problem between two rational-valued structures.

The condition unifies the two main approaches for designing a polynomial-time approximation scheme. One is Baker's layering technique, which applies to sparse graphs such as planar or excluded-minor graphs. The other is based on Szemer\'edi's regularity lemma and applies to dense graphs. We extend the applicability of both techniques to new classes of Max-CSPs. On the other hand, we prove that the condition cannot be used to find solutions (as opposed to approximating the optimal value) in general.
  
Treewidth-pliability turns out to be a robust notion that can be defined in several equivalent ways, including characterisations via size, treedepth, or the Hadwiger number. We show connections to the notions of fractional-treewidth-fragility from structural graph theory, hyperfiniteness from the area of property testing, and regularity partitions from the theory of dense graph limits. These may be of independent interest. In particular we show that a monotone class of graphs is hyperfinite if and only if it is fractionally-treewidth-fragile and has bounded degree.

\end{abstract}

\section{Introduction}

The problem of finding a \emph{maximum cut} in a graph (Max-Cut) is one of the
most studied problems from Karp's original list of 21 NP-complete
problems~\cite{Karp72:reducibility}. While Max-Cut is NP-hard to solve
optimally, 
there is a trivial $0.5$-approximation algorithm~\cite{Sahni76:jacm} and
the celebrated $0.878$-approximation algorithm of Goemans and
Williamson~\cite{Goemans95:jacm}.
Papadimitriou and Yannakakis established that Max-Cut is
Max-SNP-hard~\cite{Papadimitriou91:jcss}. By the work of Arora, Lund, Motwani,
Sudan, and Szegedy~\cite{Arora1998:jacm-proof} this implies that, unless P=NP,
there is no polynomial-time approximation scheme (PTAS) for Max-Cut in general
graphs. However, non-trivial results exist for important special cases. On the
one hand, Max-Cut is solvable exactly in planar graphs, as shown by
Hadlock~\cite{Hadlock75:sicomp}, and more generally, Max-Cut admits a PTAS on
graph classes excluding a fixed minor, as shown by Demaine, Hajiaghayi, and
Kawarabayashi~\cite{Demaine05:focs}. On the other hand, Arora, Karger, and
Karpinski showed a PTAS for Max-Cut in dense graphs~\cite{Arora99:jcss},
where a graph class is dense if every graph in it contains at least a constant fraction of all
possible edges. 

Max-Cut is an example of \emph{maximum constraint satisfaction problem}
(\CSP), although a very special one (the alphabet size is $2$, in particular constant, and every constraint uses the same symmetric predicate ``$x \neq y$'' of arity 2).
Another well-known example is Max-$r$-SAT, with alphabet~size~2 and $r$-ary clauses.
Motivated by results on planar, excluded-minor, and dense graph classes,
our goal in this paper is to understand the following question: 

\begin{quote}
	\centering\emph{What structure allows for the existence of a PTAS for \CSPs?}
\end{quote}

We adopt a permissive definition of PTAS here: given a \CSP instance and an
arbitrarily small $\varepsilon>0$, the goal is to find a $(1-\varepsilon)$
multiplicative approximation of the value of an optimal solution to the instance
(but, unlike in most papers, we do not require that the algorithm should find a
solution achieving the bound).

We focus on two computational problems.
First, we study the general $\bCSP(\Gg)$ problem parameterised  by the class
of underlying constraint graphs (a.k.a. \emph{primal} or \emph{Gaifman} graphs).
The input is a graph $G \in \Gg$, an alphabet $\Sigma_v$ for each vertex,
and a valued constraint $f_{uv} \colon \Sigma_u \times \Sigma_v \to \QQ_{\geq 0}$
for each edge $uv$.
The goal is to find an assignment $h(v) \in \Sigma_v$ maximising $\sum_{uv} f_{uv}(h(u),h(v))$.
Similarly, in $\rCSP(\Gg)$ a constraint may appear on any $r$-clique in $G$.
The constraints are arbitrary (non-negative) and the alphabets are not fixed, making the problem very expressive.%
\footnote{One could attempt to generalise counting problems by maximising $\prod_{uv} f_{uv}(h(u),h(v))$ instead, or equivalently its logarithm $\sum_{uv} \log f_{uv}(h(u),h(v))$. However, the requirement $f_{uv} \geq 0$ and the approximation ratio change.
This changes the complexity: for example, approximating the number of 3-colourings requires deciding whether there is at least one in polynomial time, which is NP-hard already in 4-regular planar graphs~\cite{Dailey80}.}

Second, we consider a more general framework called the \emph{maximum
	homomorphism problem} (\HOM) of computing the maximum value of any map between
two given $\qplus$-valued structures $\structA$ and $\structB$; the value
will be denoted by $\opt{\structA,\structB}$ (see Section~\ref{sec:prelims} for precise definitions).
Intuitively, the left-hand-side structure
describes the (weighted) scopes of the constraints and the right-hand-side
structure describes the different types of constraints.
Following Grohe's notation~\cite{Grohe07:jacm}, for a class of structures
$\A$ we denote by $\HOM(\A,-)$ the restriction of \HOM to instances
$(\structA,\structB)$ with $\structA\in\A$ and $\structB$ arbitrary.
This framework captures the $\rCSP(\Gg)$ problem as a particular case: it is equivalent to $\HOM(\AOver{\Gg}{(r)},-)$,
where by $\AOver{\Gg}{(r)}$ we denote the class of all valued structures with an underlying graph in $\Gg$ and arity $r$.
Another example is the case of \emph{graph} \CSP, by which we mean a \bCSP{} that uses the same symmetric
predicate in all constraints (as in Max-Cut or Max-$q$-Cut); this case is equivalent to $\HOM(\A,-)$ where the structures in $\A$ are graphs.

The question of what structure allows to solve \CSPs \emph{exactly} in polynomial time is well understood.
A standard dynamic approach works for $\rCSP(\Gg)$ when $\Gg$ is a class of graphs of bounded treewidth.
Grohe, Schwentick, and Segoufin~\cite{GroheSS01} in fact proved the converse: if $\Gg$ has unbounded treewidth
then $\rCSP(\Gg)$, in fact already deciding the existence of a solution satisfying all constraints,
cannot be solved in polynomial time (assuming FPT$\neq$W[1]).
Grohe's theorem~\cite{Grohe07:jacm} then extended it to the more general framework:
for a class of relational (or $\{0,1\}$-valued) structures $\A$ of bounded arity,
the decision problem Hom$(\A, -)$ can be solved in polynomial time if and only if
the \emph{cores} of structures in $\A$ have bounded treewidth.
(The core is the smallest homomorphically equivalent substructure;
for example, bipartite graphs all have the single edge graph $K_2$ as a core,
so Hom$(\A, -)$ is easy when $\A$ is a class of bipartite graphs).
This was recently extended further to exact optimisation with valued structures  by
Carbonnel, Romero, and \v{Z}ivn\'y~\cite{crz22:sicomp}.

\rCSPs do not admit a PTAS in general, since already Max-Cut does not.
On the other hand, the techniques that give PTASes for Max-Cut on sparse and dense graphs
apply more generally (in fact to a variety of problems beyond Max-CSPs).
Our main contribution is a unifying condition, \emph{treewidth-pliability}, 
that captures
all known PTASes for $\rCSP(\Gg)$ and $\HOM(\A,-)$ problems.

We call a class of structures $\A$ $\tw$-pliable if it is uniformly close to structures of bounded treewidth.
More formally, for any $\eps>0$ there is a
$k=k(\eps)$ such that every structure in $\A$ has an $\eps$-close structure with treewidth at most $k$.
Here we consider two structures $\structA $ and $\structB$ to be $\eps$-close if $\opt{\structA,\structC}$ is $\eps$-close to
$\opt{\structB,\structC}$ for all $\structC$ (details in Section~\ref{sec:Dist}; this notion of distance, which we also characterise combinatorially, may be of independent interest).
While the structure of bounded treewidth is not known and cannot be efficiently computed,
we show that the Sherali-Adams LP relaxation gives a PTAS for $\HOM(\A,-)$. 

\begin{theorem}\label{thm:main1}
	If $\A$ is a $\tw$-pliable class of structures of bounded arity, then $\HOM(\A,-)$ admits a PTAS.
\end{theorem}

We emphasise the generality of Theorem~\ref{thm:main1}.\footnote{However, the
generality comes at a cost, as detailed in Section~\ref{subsec:robust}: while an
approximate optimum can be found, an approximate solution cannot be constructed
unless P=NP.}
Firstly, the computational problem (\HOM) captures many fundamental problems,
including graph homomorphisms~\cite{Hell:graphs}, Max-Cut, Max-DiCut, Max-SAT,
\CSPs, and query related problems
coming from database theory~\cite{Grohe07:jacm}.
Secondly, the
notion of pliability captures many previously discovered cases of structures
that admit a PTAS. In particular, we now discuss how Theorem~\ref{thm:main1}
extends the applicability of the two main approaches for obtaining PTASes. 

\subsection{Sparse structures: Baker's technique and fragility}
Perhaps the best known technique for solving problems on planar graphs is Lipton and Tarjan's planar separator theorem~\cite{LiptonTarjan} and the divide~\&~conquer approach it enables~\cite{Lipton80:sicomp}.
It can be used to give a PTAS for \CSPs with fixed alphabet size 
on planar graphs (this extends to excluded-minor graphs~\cite{AlonST90} and more~\cite{DvorakN16}) of bounded degree.

This approach was superseded by Baker's technique~\cite{Baker94:jacm}, which provides better running times and is easily applied to general $\rCSPs$ on arbitrary planar graphs (see e.g.~\cite{Khanna96:stoc}).
The idea is very elegant: we partition a planar graph into Breadth-First-Search layers, remove every $\ell$-th layer, and show that the remaining components of $\ell-1$ consecutive layers have bounded treewidth (and so can be solved exactly).
By trying different starting layers we can ensure that the removed layers intersect an unknown optimal solution at most $\Oh(\frac{1}{\ell})$ times, giving a $1\pm\Oh(\frac{1}{\ell})$ approximation.

From planar graphs this was extended to graphs of bounded genus by Eppstein~\cite{Eppstein0d:algorithmica-diameter} and later to all graph
classes excluding a fixed minor by Grohe~\cite{Grohe03} and Demaine et al.~\cite{Demaine05:focs}.
The structural property needed for this approach, originally proved for excluded-minor graphs by DeVos, Ding, Oporowski, Sanders, Reed, Seymour, and Vertigan~\cite{DeVos+04}, is $\tw$-\emph{fragility}: they can be partitioned into any constant number of parts such that removing any one part leaves a graph of bounded treewidth.
As shown by Hunt, Marathe, and Stearns~\cite{Marathe97,Marathe98} (see also~\cite{Hunt98}) as well as Grigoriev and Bodlaender~\cite{Grigoriev2007}, the same property applies to some geometrically-defined graph classes that do not exclude any minor. One example is intersection graphs of unit disks whose centers are at least some constant apart (capturing some applications of the closely related shifting technique of Hochbaum and Maass~\cite{HochbaumM85} for geometric packing and covering problems).
Another example is 1-planar graphs, or more generally graphs drawn on a fixed surface with a bounded number of intersections per edge.

An important generalisation, \emph{fractional-$\tw$-fragility}, was introduced by Dvo\v{r}\'ak~\cite{Dvorak16}:
it suffices that the parts whose removal results in a graph of bounded treewidth are nearly-disjoint (Definition~\ref{def:fragile}).
This applies to $d$-dimensional variants of the geometric classes mentioned above (for any constant~$d$), in particular to $d$-dimensional grids, which are not $\tw$-fragile~\cite{BergerDN18}; this also includes classes of polynomial growth~\cite{KrauthgamerL07,Dvorak2021}.
Another large family of fractionally-$\tw$-fragile classes are classes of bounded degree with strongly sublinear separators~\cite{Dvorak16} (equivalently, bounded degree and polynomial expansion~\cite{DvorakN16}).
For such concrete examples of fragile classes, known proofs show that the nearly-disjoint parts can be computed efficiently.
A PTAS can then easily be designed from the definition~\cite{Dvorak16}.

We show that the assumption about efficient construction is not needed.
We do this by proving that if $\Gg$ is \emph{any} fractionally-$\tw$-fragile class of graphs
(intuitively, any class where a Baker-like technique is known to work),
then the class $\AOver{\Gg}{(r)}$ of all possible structures of bounded arity
$r$ and with Gaifman graph in $\Gg$ is $\tw$-pliable.

\begin{theorem}\label{thm:fragileIsNice}
	Let $\Gg$ be a fractionally-$\tw$-fragile class of graphs.
Then $\AOver{\Gg}{(r)}$  is $\tw$-pliable for every~$r$.
	Consequently, $\rCSP(\Gg)$ admits a PTAS.
\end{theorem}

This captures all graph classes $\Gg$ where a PTAS for $\rCSP(\Gg)$ is known.

\subsection{Dense structures: the regularity lemma}
It is perhaps more surprising that dense structures admit a PTAS.
Here a class is \emph{dense} if a constant factor of all possible constraints is present in every structure in the class, e.g. graphs with $\Omega(n^2)$ edges.
Arora, Karger, and Karpinski~\cite{Arora99:jcss} showed that \rCSPs admit a PTAS in the dense regime if the
alphabet size is constant (in fact Boolean); de la Vega~\cite{FernandezdelaVega96:rsa} independently gave a PTAS for dense Max-Cut. 
Frieze and Kannan~\cite{Frieze96:focs} proved that these results are essentially possible because of Szemer\'edi's regularity lemma~\cite{Szemeredi78:reglar}: intuitively, every graph can be approximated to within an additive $\pm\eps n^2$ error by a random graph (with a constant number of parts, depending on $\eps$ only, so that the edges between two parts form a uniformly random graph of some density).
For dense graphs, the additive error translates to a relative error, giving a PTAS.
They also showed a variant of the regularity lemma that is still applicable to \rCSPs with constant alphabet size, yet avoids its infamous tower-type dependency on $\eps$.

Goldreich, Goldwasser, and Ron~\cite{GoldreichGR98} connected these results to the area of \emph{property testing}, spawning an entirely new direction of research.
They gave constant-time algorithms estimating the optimum value of some \gCSPs.
In fact, Alon, de la Vega, Kannan, and Karpinski~\cite{Alon03:jcss} (see also Andersson and Engebretsen~\cite{AnderssonE02}) showed that \rCSPs with a fixed alphabet can be approximated with accuracy $\pm\eps n^r$ by sampling a constant number of vertices (polynomial in $\frac{1}{\eps}$) and finding the optimum on the resulting (constant-size) induced substructure. 

None of these results apply to any $\rCSP(\Gg)$ and $\HOM(\A,-)$ problem, that is, to unbounded alphabets.
We give the first such example: undirected graphs with $\Omega(n^2)$ edges.
\begin{theorem}\label{thm:denseIsNice}
	Let $c>0$ and let $\A$ be a class of graphs with at least $cn^2$ edges.
	Then $\A$ is $\tw$-pliable.
	Consequently, $\HOM(\A,-)$ admits a PTAS.
\end{theorem}
\noindent
(Note here the graphs in $\A$ are input structures, not just Gaifman graphs of input structures).
We also show that this cannot be extended to general CSPs:
already for the class of tournaments---that is, orientations of complete graphs---
a PTAS is impossible, (cf. Corollary~\ref{cor:tournamentsAreHard} in
Section~\ref{sec:hardness}) and indeed, this class is not $\tw$-pliable (cf.
Remark~\ref{rem:tournaments}).

\subsection{Robustness of pliability}\label{subsec:robust} 

The notion of treewidth-pliability not only unifies the different existing
algorithmic techniques but it is also quite robust: treewidth-pliability captures a
valued analogue of ``homomorphic equivalence'' (e.g. bipartite graphs, or
3-colourable graphs where each edge is contained in exactly one triangle, cf.
Examples~\ref{ex:bipartite} and~\ref{ex:triangle} in Section~\ref{subsec:pliable}) as
well as small edits: if $\A$ is a pliable class of graphs, say, then the class
of graphs obtained by adding or removing $o(m)$ edges from $m$-edge graphs in
$\A$ is again pliable (Lemma~\ref{lem:closeToPliable} in Section~\ref{subsec:pliable}).
However, this generality comes at a price. First, we show that even for fixed alphabet size,
although the approximate optimum value can be found, an approximate solution cannot be constructed
(unless $\text{P}=\text{NP}$, cf. Example~\ref{ex:noFind} in Section~\ref{subsec:pliable}).
Second, unlike in some of the previous results for more restricted classes, our result does not give an EPTAS (i.e., with the degree of the
polynomial time bound independent of $\eps$) for fixed alphabet size (cf. Question~\ref{q:noEPTAS}).
Finally, the use of strong versions of
the regularity lemma yields tower-type dependencies on the approximation ratio $\eps$ in the dense case.

\medskip 

In the definition of treewidth-pliability we approximate structures by
comparing their $\opt$ values and we ask them to be close to structures where
the problem can be solved exactly. This is a non-constructive and very general 
definition. In fact, it is not inconceivable that this captures all tractable cases, i.e., that $\HOM(\A,-)$ has a PTAS if and only if $\A$ is $\tw$-pliable.
Nevertheless, we show a variety of equivalent combinatorial definitions,
which allow us to place a fairly tight bound on what pliability is, structurally.

For classes of the form $\AOver{\Gg}{(r)}$, that is, if we only restrict the underlying Gaifman graphs, we show that pliability collapses to fractional fragility.
In this sense we understand the ``sparse'' setting exactly.

\begin{lemma}\label{lem:fragileIffNice}
	Let $\Gg$ be a class of graphs. The following are equivalent, for any $r \geq 2$:
	\begin{itemize}
		\item $\Gg$ is fractionally-$\tw$-fragile;
		\item $\AOver{\Gg}{(r)}$ is $\tw$-pliable.
	\end{itemize}
\end{lemma}

In general, we can replace treewidth with other parameters of the Gaifman graph:
size (number of vertices), treedepth, denoted by $\td$, Hadwiger number (maximum clique minor size), or maximum connected component size, which we denote by $\cc$.

\begin{theorem}\label{thm:niceWrt}
	Let $\A$ be any class of structures. The following are equivalent:
	\begin{itemize}\vspace*{-1ex}
		\begin{multicols}{3}
		\item $\A$ is $\td$-pliable;	
		\item $\A$ is $\tw$-pliable;
		\item $\A$ is $\Hadwiger$-pliable.
		\end{multicols}
	\end{itemize}\vspace*{-2ex}
	If structures in $\A$ have bounded signatures,
	then the following are equivalent to the above as well:\looseness=-1
	\begin{itemize}\vspace*{-1ex}
		\begin{multicols}{3}
		\item $\A$ is $\size$-pliable;		
		\item $\A$ is $\cc$-pliable.
		\end{multicols}
	\end{itemize}
\end{theorem}

Classes of structures with bounded signatures (see Section~\ref{sec:prelims} for precise definitions) correspond to \CSP instances with a bounded number of constraint types; e.g. maximum graph homomorphism.
For example, any class of dense graphs as in Theorem~\ref{thm:denseIsNice} is in fact $\size$-pliable.
An example of a class with unbounded signatures is any class of the form $\AOver{\Gg}{(r)}$ (we do not consider infinite signatures, but there are arbitrarily many symbols in those signatures).
Theorem~\ref{thm:niceWrt} allows us to give concrete and general examples of classes that are \emph{not} $\tw$-pliable:
the class of orientations of graphs in $\Gg$, where $\Gg$ is \emph{any} class of
unbounded average degree (Lemma~\ref{lem:niceIsMad} in Section~\ref{sec:non-examples}), or any
class of 3-regular graphs with unbounded girth (Lemma~\ref{lem:avg2} in
Section~\ref{sec:non-examples}).

\medskip
Finally, as a side result, we connect \emph{hyperfiniteness} to fragility.
A class of graphs $\Gg$ is called \emph{hyperfinite} if
for every $\eps>0$ there is a $k=k(\eps)$ such that
in every $G\in\Gg$ one can remove an at-most-$\eps$ fraction of edges to obtain a graph with connected components of size at most $k$.
For a monotone class of graphs (closed under taking subgraphs), hyperfiniteness easily implies bounded degree.
It is an important notion in property testing: many results in sparse graphs were generalised by the statement that every property of hyperfinite graphs is testable~\cite{NewmanS13}.
The idea, originating in the work of Benjamini, Schramm, and Shapira~\cite{BenjaminiSS08} and Hassidim, Kelner, Nguyen, and Onak~\cite{HassidimKNO09}, is that following the approach of Lipton and Tarjan, graphs with sufficiently sublinear separators, such as planar or excluded-minor graphs~\cite{AlonST90}, can be recursively partitioned into bounded-size components, which for bounded-degree graphs gives hyperfiniteness (see e.g.~\cite[Cor. 3.2]{CzumajSS09} for a slightly stronger property, cf.~\cite{MoshkovitzS18}).
This allows, analogously as in the dense case, to give a constant-size approximate description of such graphs by sampling constant-radius balls in them~\cite{NewmanS13}.
See \cite{goldreich2017} for a book on property testing and \cite{KumarSS19} for a recent improvement for excluded-minor graphs.

We show that a monotone class $\Gg$ is hyperfinite if and only if it is fractionally-$\tw$-fragile and has bounded degree.
In fact, replacing the parameter treewidth by the maximum size of a connected component in a graph, we have:

\begin{theorem}\label{thm:hyperfinite}
	Let $\Gg$ be a monotone class of graphs. The following are equivalent:
	\begin{itemize}
		\item $\Gg$ is hyperfinite;
		\item $\Gg$ is fractionally-$\tw$-fragile and has bounded degree;
		\item $\Gg$ is fractionally-$\cc$-fragile;		
		\item $\AOver{\Gg}{(r)}$ is $\cc$-pliable (for any $r\geq 2$).
	\end{itemize}
\end{theorem}

\noindent
The equivalence of the second and third bullet points was shown by Dvo\v{r}\'ak~\cite[Observation 15, Corollary 20]{Dvorak16},
while for the third and fourth the proof is established by a generalisation of
Lemma~\ref{lem:fragileIffNice}, cf. Lemma~\ref{lem:fragileIffNiceGen}.

Hyperfiniteness originates from the study of amenable groups and graphs limits,
with motivations in geometry and mathematical physics~\cite{Elek08:jfa}.
The unexpected connection with fractional fragility already found an application
in that area: Elek~\cite{Elek:pams}
showed that our Theorem~\ref{thm:hyperfinite} gives that last missing
implication in proving the equivalence of some properties of infinite,
bounded-degree graphs (in particular ``uniform local amenability'' and
``Property A''), which answers a question of Brodzki, Niblo, Spakula, Willett
and Wright~\cite{Brodzki13:jng}.

\subsection{Related work}

While this paper focuses on \rCSPs, Baker's technique and the regularity lemma apply to many more problems.
In fact Khanna and Motwani~\cite{Khanna96:stoc} argued that most known PTAS algorithms can be derived from three canonical optimisation problems on planar graphs, the first being \CSP and the latter two being so-called Max-Ones and Min-Ones CSPs (also solvable with Baker's technique).
One of the very few results that did not fit their framework was the PTAS for dense Max-Cut.
A follow-up work by Mezei, Wrochna, and \v{Z}ivn\'y~\cite{Mezei23:talg} on the
extended abstract of this work~\cite{rwz21:soda} extended some of the results of
the present paper to Min- and Max-CSPs with crisp constraints, which include the
Max-Ones and Min-Ones CSPs mentioned above.

Generic frameworks extending Baker's technique include
the bidimensionality theory of Demaine, Fomin, Hajiaghayi, and Thilikos~\cite{Demaine05:jacm} and its application in the design of PTASes by Demaine and Hajiaghayi~\cite{Demain05:soda} (which is however limited to minor-closed graph classes); 
monotone FO problems on minor-closed graph classes by Dawar, Grohe, Kreutzer,
and Schweikardt~\cite{Dawar06:lics};
 and the idea of Baker games, introduced by Dvo\v{r}\'ak~\cite{Dvorak20} (see also~\cite{Dvorak18}).
The latter gives conditions stronger than fractional-$\tw$-fragility,
but useful for problems beyond \CSPs,
and achievable for all examples known to be fractionally fragile.
The work of Dvo\v{r}\'ak and Lahiri~\cite{DvorakL21}, which appeared after the
present paper, gives a PTAS 
on fractionally-$\tw$-fragile classes of graphs
for problems incomparable with \CSPs, namely monotone maximisation
problems expressible in terms of distances.

De la Vega and Karpinski~\cite{VegaK02,VegaK06} extended the dense approach to subdense cases ($\Omega(\frac{n^2}{\log n})$ edges) for specific problems such as MaxCut and Max-2-SAT.
In contrast, they show that Max-Cut on graphs with $\Omega(n^{2-\delta})$ edges is hard to approximate, for any $\delta>0$.

The best known approximation algorithm for general \bCSPs is due to Charikar,
Hajiaghayi, and Karloff~\cite{Charikar11:algorithmica} and achieves an
approximation factor of $\Oh((nq)^{1/3})$, where $n$ is the number of variables
and $q$ is the alphabet size. On the hardness side,
Dinur, Fischer, Kindler, Raz, and Safra~\cite{DinurFKRS11}
showed that $\Oh(2^{\log^{1-\delta}(nq)})$-approximation of \bCSPs is NP-hard.
Manurangsi and Moshkovitz~\cite{Manurangsi15:approx} gave approximation
algorithms for \emph{dense} \bCSPs with large alphabet size (but not PTASes).
Manurangsi and Raghavendra~\cite{Manurangsi17:icalp} establish a tight trade-off between
running time and approximation ratio for dense \rCSPs for $r>2$.

CSPs have also been extensively studied for fixed constraint types, i.e., $\HOM(-,\structB)$ problems for fixed $\structB$.
Raghavendra showed that the best approximation ratio is always achieved by the basic SDP
relaxation~\cite{Raghavendra08:everycsp},
assuming Khot's unique games conjecture~\cite{Khot02stoc}.
The exactly solvable cases were characterised by Thapper and \v{Z}ivn\'y~\cite{tz16:jacm}. 
The approximation factor of \gCSPs was studied by Langberg,
Rabani, and Swamy~\cite{Langberg06:approx}.

\subsection{Overview}

In Section~\ref{sec:prelims}, we give formal definitions and present our basic tool:
two structures $\structA, \structB$ have similar values of $\opt{-,\structC}$ if and only if there is a certain fractional cover, which we call an \emph{overcast}, from $\structA$ to $\structB$ and from $\structB$ to $\structA$.
Section~\ref{sec:prelims} then relates our notion of pliability with two notions of
distances, and gives examples and non-examples of pliable classes of structures.

To prove that treewidth-pliability leads to a PTAS (Theorem~\ref{thm:main1}) the main idea is that an overcast allows to show that the values of $\opt{-,\structC}$ are still similar when we look at linear programming relaxations.
The details, as well as the definition of the Sherali-Adams linear programming
relaxation, are given in Section~\ref{sec:tract}.
In Section~\ref{sec:frag}, we introduce equivalent definitions of fractional
fragility and study their properties. This will allow us to 
prove Theorem~\ref{thm:fragileIsNice} by showing how the definition implies suitable overcasts.
This also allows us to establish Lemma~\ref{lem:fragileIffNice}.
Theorem~\ref{thm:niceWrt} is proved in Section~\ref{sec:sizeccproof}.
Theorem~\ref{thm:hyperfinite} on hyperfiniteness is proved in Section~\ref{sec:hyperfinite}.
Section~\ref{sec:dense} gives a proof of Theorem~\ref{thm:denseIsNice} on dense graphs.

We conclude with open questions in Section~\ref{sec:open}.

\section{Preliminaries}
\label{sec:prelims}
\subsection{Structures}
A \emph{signature} is a finite set $\sigma$ of (function) symbols $f$, each with
a specified finite arity $\ar(f)$. We denote by $|\sigma|$ the number of symbols in the
signature~$\sigma$. 
A \emph{structure} $\structA$ over a signature $\sigma$ (or
$\sigma$-structure $\structA$, for short) is a finite domain $A$ together with
a function $f^{\structA}: A^{\ar(f)} \to \qplus$ for each symbol $f \in
\sigma$.
We say that a class of structures has \emph{bounded} signatures if for the signatures $\sigma$ of structures in the class, $|\sigma|$ is bounded by a constant (so unbounded means arbitrarily many symbols; we do not consider infinite signatures).
Note that a class of $\sigma$-structures (that is, structures over a fixed
signature $\sigma$) has bounded signatures and bounded arities (the maximum
arity occurring in $\sigma$ is a finite constant).

We denote by $A,B,C,\dots$ the domains of structures $\structA,\structB,\structC,\dots$.
For sets $A$ and $B$, we denote by $B^A$ the set of all mappings from $A$ to $B$.
We define $\tup{\structA}$ to be the set of all pairs $(f,\tuple{x})$ such that
$f \in \sigma$ and $\tuple{x} \in A^{\ar(f)}$, and by $\tup{\structA}_{>0}$ the
pairs $(f,\tuple{x})\in\tup{\structA}$ with $f^{\structA}(\tuple{x})>0$.\looseness=-1

We denote $\|\structA\|_\infty := \max_{(f,\tuple{x}) \in
	\tup{\structA}} f^\structA(x)$ and  $\|\structA\|_1 := \sum_{(f,\tuple{x}) \in
	\tup{\structA}} f^\structA(x)$. 
For $\lambda\geq 0$ we write
$\lambda\structA$ for the \emph{rescaled} $\sigma$-structure with domain $A$ and
$f^{\lambda\structA}(\tuple{x})\defeq \lambda f^{\structA}(\tuple{x})$, for
$(f,\tuple{x}) \in \tup{\structA}$.

Given a $\sigma$-structure $\structA$, the \emph{Gaifman graph} (or
\emph{primal graph}), denoted by $\Gaifman(\structA)$, is the graph whose vertex set is
the domain $A$, and whose edges are the pairs $\{u,v\}$ for
which there is a tuple $\tuple{x}$ and a symbol $f\in \sigma$ such that $u,v$
appear in $\tuple{x}$ and $f^{\structA}(\tuple{x})>0$.

For $r\geq 2$ and a class of graphs $\Gg$, we denote by $\AOver{\Gg}{(r)}$ the class of $\sigma$-structures $\structA$ with $\Gaifman(\structA)\in\Gg$ and $\ar(f)\leq r$ for every
$f\in\sigma$. (Note that $\AOver{\Gg}{(r)}$ contains structures over distinct
signatures.)\looseness=-1

The \emph{maximum homomorphism problem} (\HOM) is the following computational problem.
An instance of \HOM consists of two structures $\structA$ and $\structB$ over the 
same signature. For a mapping $h:A \to B$, we define 
$\val(h) = \sum_{(f,\tuple{x})\in \tup{\structA}}f^{\structA}(\tuple{x})f^{\structB}(h(\tuple{x}))$.
The goal is to find the maximum value over all possible mappings $h:A\to B$.%
\footnote{While called maximum
	homomorphism, we note that the maximisation is over all possible maps, not
  only homomorphisms, i.e., those that map non-zero tuples into non-zero tuples.}
We denote this value by $\opt{\structA,\structB}$. 
Note that when seen as a Max-CSP instance, the domain of the left-hand side structure $\structA$ is the variable set, while the domain of the right-hand side structure $\structB$ is the \emph{alphabet}.

\begin{example} 
  Let $\sigma=\{f\}$ be a signature consisting of a single symbol $f$ of arity
  $\ar(f)=2$. Let $\structB$ be a $\sigma$-structure with the domain $B=\{0,1\}$
  and let $f^\structB:B^2\to\qplus$ be defined by $f^\structB(x,y)=1$ if $x\neq y$ and
  $f^\structB(x,y)=0$ if $x=y$. Given an undirected graph $G=(V,E)$, we can
  encode it as a $\sigma$-structure $\structA$ with the domain $A=V$ and with
  $f^\structA(x,y)=1$ if $\{x,y\}\in E$ and $f^\structA(x,y)=0$ otherwise. Now,
  the instance $(\structA,\structB)$ of \HOM is the same as the Max-Cut problem
  in $G$.
  The Max-DiCut problem (in a directed graph $(V,E)$) would be cast as \HOM very
  similarly. The only differences would be in the definition of $f^\structA$ and
  $f^\structB$: 
  $f^\structA(x,y)=1$ if $(x,y)\in E$ and $0$ otherwise,
  $f^\structB(x,y)=1$ if $x=0$ and $y=1$ and $0$ otherwise. 
\end{example}

\begin{example} An example of a problem that is not a \HOM is the Maximum
Independent Set problem. Intuitively, the ``no edges'' constraints imposed on
  an independent set are strict. This problem can be, however, cast as an instance of a
  \HOM with both rational and (negative) infinite costs, cf.~\cite{Mezei23:talg}
  for follow-up work.
\end{example}

Given a class $\A$ of structures, $\bHOM{\A}{-}$ is the problem restricted to instances
$(\structA,\structB)$ of \HOM with $\structA\in\A$ (it is a promise problem: algorithms are allowed to do anything when $\structA\not\in\A$).
Recall that for a class of graphs $\Gg$, the problem $\rCSP(\Gg)$ is equivalent to $\HOM(\AOver{\Gg}{(r)}, -)$.%
\footnote{Note that $\HOM(\AOver{\Gg}{(r)}, -)$ is different from the maximum graph homomorphism problem $\HOM(\Gg, -)$.
Indeed, graphs are also structures over the signature $\{e\}$ with one symbol of arity 2
(where $e^G(u,v) = [uv\text{ is an edge of }G]$, if the graph is not weighted).
To avoid confusion, we use $\Gg$ for a class of Gaifman graphs of some structures
and $\A$ for a class of graphs that are themselves used as input structures.}

\subsection{Overcasts}
Before we define pliability formally, it is useful to consider the following relation.
The starting point of all our results is the equivalence of this relation to a more combinatorial notion: the existence of a certain fractional cover, which we shall call an \emph{overcast}.

\begin{definition}
	Let $\structA$ and $\structB$ be $\sigma$-structures.
	We say that $\structA$ \emph{overcasts} $\structB$, denoted $\structA\overcasts\structB$ if, for all $\sigma$-structures $\structC$,
	we have that $\opt{\structA,\structC}\ \geq\ \opt{\structB,\structC}$.
\end{definition}

A distribution over a finite set $U$ is a function $\pi \colon U \to \QQ_{\geq 0}$ such that $\sum_{x \in U} \pi(x) = 1$. 
We write $\EX_{x\sim\pi} f(x)$ for $\sum_{x \in U} \pi(x) \cdot f(x)$ and $\Pr_{x\sim\pi}[\phi(x)]$ for $\EX_{x\sim\pi}[\phi(x)]$, where $[\phi(x)]$ is 1 if $x$ satisfies the predicate $\phi$ and 0 otherwise.

Given a map $g:A\to B$ and a tuple $\tuple{x}=(x_1,\ldots,x_m)\in A^m$, we write
$g(\tuple{x})$ for $(g(x_1),\ldots,g(x_m))$; i.e., we apply $g$ componentwise on
$\tuple{x}$. Hence, $g^{-1}(\tuple{y})=\{\tuple{x}\mid
g(\tuple{x})=\tuple{y}\}$.

\begin{definition}
	Let $\structA$ and $\structB$ be $\sigma$-structures. 
	An \emph{overcast} from $\structA$ to $\structB$ is a distribution $\omega$ over $B^A$ 
	such that for each $(f,\tuple{x})\in \tup{\structB}$ we have that
	\[\EX_{g\sim\omega} f^\structA(g^{-1}(\tuple{x}))\ \geq\ f^{\structB}(\tuple{x}).\]
	Here $f^{\structA}(g^{-1}(\tuple{x}))$ denotes the sum of $f^{\structA}(\tuple{y})$ over $\tuple{y}\in g^{-1}(\tuple{x}) \subseteq A^{\ar(f)}$. 
\end{definition}

Intuitively, an overcast from $\structA$ to $\structB$ is a random function from $A$ to $B$ such that for each edge (tuple) in $\structB$, its preimage has larger expected weight (value). In other words, each edge must be covered by at least its own weight, in expectation.
The following is a consequence of Farkas' lemma, as shown in Appendix~\ref{sec:Farkas}.%
\footnote{
The definitions of the $\overcasts$ relation and of an overcast are analogous to
the ``improvement'' relation and ``inverse fractional homomorphisms''
from~\cite{crz22:sicomp}. Here, however, $\opt$ is maximising, not
minimising, so inequalities in definitions are swapped. This has consequences
such as the fact that mappings in the support of an overcast are in general
not homomorphisms (mapping non-zero tuples to non-zero tuples), unlike for
inverse fractional homomorphisms. The proof of Proposition~\ref{prop:overcast}
nevertheless is identical to the proof of~\cite[Proposition~3.6]{crz22:sicomp}.\looseness=-1
}

\begin{proposition} \label{prop:overcast} 
	$\structA \overcasts \structB$ if and only if there is an overcast from $\structA$ to $\structB$. 
\end{proposition}

\subsection{Pliability and graph parameters}\label{sec:Dist}
Our definition of pliability involves a notion of distance which may be of independent interest.
It quantifies the \emph{relative} difference between two structures (as measured from the right by weighted multicut densities, in the language of Lov\'{a}sz's book on graph limits~\cite[Ch. 12]{LovaszBook}).
\begin{definition}\label{def:pliable}
	The \emph{opt-distance} between two structures with the same signature is defined as:\looseness=-1
	\[
	\dopt(\structA,\structB) := \textstyle\sup_\structC \left|\ln \opt{\structA,\structC}
	- \ln \opt{\structB,\structC}\right|.\]
	Here $\ln 0 = -\infty$ and $\left|\ln 0 - \ln 0\right|=0$.
	Equivalently, we can compare rescaled structures; by definition of~$\succeq$
  and the fact that $\opt{\lambda \structA, \structC} = \lambda\opt{\structA, \structC}$, we have:
	\[ \dopt(\structA,\structB) = \inf\big\{ \eps  \bigm| 		
	\structA \,\overcasts\, e^{-\eps}\,\structB\ \text{ and }\ \structB \,\overcasts\,
	e^{-\eps}\,\structA\, \big\}.\]
\end{definition}
\noindent
One may think of $e^{\pm\eps}$ as close to $1\pm\eps$. Formally
$\textstyle 1-\eps\ \leq\ e^{-\eps}\ \leq\ \frac{1}{1+\eps}\ =\
1-\eps+\Oh(\eps^2)$ for $\eps \geq 0$. \looseness=-1
Note that the first definition readily implies that opt-distance satisfies the triangle inequality (it defines a pseudometric).

\medskip
Finally, a class is treewidth-pliable if it is uniformly close to structures of bounded treewidth:

\begin{definition}
	A class of structures $\A$ is $\param$-\emph{pliable} with respect to a graph parameter $\param$ if 
	for every $\varepsilon>0$, there is $k=k(\varepsilon)$ such that
	for every $\sigma$-structure	$\structA\in\A$ there is a $\sigma$-structure $\structB$ with $\param(\structB) \leq k$ and $\dopt(\structA,\structB) \leq \eps$.
\end{definition}

Thus to show $\tw$-pliability of various classes, we will construct overcasts from structures $\structA$ in the class to $(1-\eps)\structB$, for some $\structB$ of bounded treewidth, and from $\structB$ back to $(1-\eps) \structA$.

\bigskip
\noindent
Given a graph $G$, we will consider pliability for the following graph parameters:
\begin{itemize}
  \item $\size(G) = |V(G)|$ -- the number of vertices of $G$,
    \item $\cc(G)$ -- the maximum size of a connected component of $G$,
    \item treedepth $\td(G)$, which is a parameter due to Ne\v{s}et\v{r}il and Ossona de Mendez~\cite{NesetrilM06}, whose definition we recall below,
    \item treewidth $\tw(G)$,\footnote{We refer
to~\cite{Diestel10:graph} for the standard definitions of treewidth, pathwidth and minors.}
    \item and finally the Hadwiger number $\Hadwiger(G)$, which is the maximum $k$ such that
$K_k$ is a minor of $G$.
    We use the Hadwiger number as an example of a broader ``sparsity'' parameter:
    for example, all planar graphs have Hadwiger number at most 4 (yet unbounded treewidth).
\end{itemize}

\begin{definition}
	The \emph{treedepth} $\td(G)$ of a graph $G$ is defined recursively as:
	\begin{itemize}
		\item $\ \max\limits_i\ \ \td(G_i)$, if $G$ is disconnected with components $G_i$;
		\item $\min\limits_{v \in V(G)} \td(G-v) +1$, if $G$ is connected and has more than one vertex;
		\item $1$, if $G$ has one vertex.
	\end{itemize}
\end{definition}

An equivalent definition is as follows:
a \emph{treedepth decomposition} of a graph $G$ is a rooted forest $T$ (a disjoint union of rooted trees) with $V(T)=V(G)$ such that for each $uv \in E(G)$, $u$ is an ancestor or descendant of $v$ in $T$.
In other words, $G$ is a subgraph of the transitive closure of a forest $T$ directed towards roots.
The treedepth of $G$ is equal to the minimum depth among all such decompositions of $G$.
Treedepth is a rather strict parameter: for example, stars have treedepth 2, but paths already have unbounded treedepth.
In fact, a short proof shows that a class of graphs has bounded treedepth if and only if the length of the longest path is bounded~\cite{NesetrilM06}.

\medskip

Bounded size implies bounded $\cc{}$ implies bounded $\td{}$ implies bounded pathwidth ($\pw$) implies bounded $\tw$ implies bounded Hadwiger number, more precisely:
\[\Hadwiger(G) - 1 \leq \tw(G) \leq \pw(G) \leq \td(G) - 1\text{ and }\td(G)\leq\cc(G)\leq\size(G).\]
Moreover, we also have the following inequality (less useful, because of the dependency on $G$):
\[\td(G) \leq (\tw(G) + 1) \cdot \log_2 |G|.\]

All these parameters are \emph{monotone}, that is, $\param(H)\leq \param(G)$ for a subgraph $H$ of a graph $G$.
Their boundedness implies bounded average degree $\frac{2|E(G)|}{|V(G)|}$.
More precisely, $\frac{2|E(G)|}{|V(G)|} \leq 2\tw(G)$ (because a graph of treewidth $k$ has a vertex of degree at most $k$);
Kostochka~\cite{Kostochka84} proved $\frac{2|E(G)|}{|V(G)|} \leq \Oh(h \sqrt h)$ where $h=\Hadwiger(G)$.

\medskip

The size of a structure $\structA$ is the number of vertices of its Gaifman
graph: $\size(\structA) = |V(\Gaifman(\structA))|$. The other graph
parameters are also defined in terms of the same parameter
on the Gaifman graph of the structure; e.g., the treewidth of a structure $\structA$ is the treewidth of its Gaifman
graph: $\tw(\structA)=\tw(\Gaifman(\structA))$. 
In particular, since the edges of $\Gaifman(\structA)$ come only from
tuples in $\structA$ of non-zero weight, a rescaled structure $\lambda\structA$, for
$\lambda>0$, has $\Gaifman(\lambda\structA)=\Gaifman(\structA)$ so the
parameters we consider do not change by rescaling; for $\lambda=0$,
$\Gaifman(\lambda\structA)$ has no edges.

We will often prove the easy directions in various characterisations via the
following observation that follows from the definition of pliability: 

\begin{observation}\label{obs:trivial}
If two graph parameters $\param$ and $\param'$ satisfy $\param\leq\param'$ then
$\param'$-pliability implies $\param$-pliability.
\end{observation}

\subsection{Opt-distance zero and edit-distance}
\label{sec:distance}

In this section we define our relative version of edit distance 
and prove it upper-bounds opt-distance.
We define the \emph{edit distance} $\dedit(\structA,\structB)$ between two valued $\sigma$-structures $\structA,\structB$ to be
\[ \dedit(\structA,\structB) :=
\min_{\text{bij. }\phi \colon A \to B} \sum_{f \in \sigma}
\frac{\sum_{\tuple{x} \in A^{\ar(f)}}  \left|f^\structA(\tuple{x}) - f^\structB(\phi(\tuple{x}))\right|}
{\min(\|\structA_f\|_1,\|\structB_f\|_1)}.\]
Here $\structA_f$ denotes the structure $\structA$ limited to the signature $\{f\}$,
so $\|\structA_f\|_1$ denotes $\sum_{\tuple{x} \in A^{\ar(f)}} f^\structA(\tuple{x})$.
The following generalises the notion of ``looplessness'' in graphs.

\begin{definition}\label{def:loopless}
	A $\sigma$-structure $\structA$ is \emph{loopless} if no tuple has a repetition.
	That is, for $(f,\tuple{x}) \in \tup{\structA}$ with $f^\structA(\tuple{x}) > 0$,
  $\tuple{x}$ consists of $\ar(f)$ different elements of
  $A$.\footnote{Equivalently, the Gaifman graph of $\structA$ is loopless.}
\end{definition}

\begin{lemma}\label{lem:editToOpt}
	For loopless structures, the opt-distance is bounded linearly by the edit
  distance:
	\[\dopt \leq C_\sigma \cdot \dedit, \]
	where $C_\sigma = \max_{f \in \sigma} \ar(f)^{\ar(f)}$.
\end{lemma}
\begin{proof}
	Let $\dedit = \dedit(\structA,\structB)$ and
	let $\phi \colon A \to B$ be a bijection minimising the expression in its definition.
	We will show that $e^{C_\sigma \cdot \dedit} \structA \succeq (1 + C_\sigma \cdot \dedit) \structA \succeq \structB$.
	Symmetrically, $e^{C_\sigma \cdot \dedit} \structB \succeq \structA$,
	hence $\dopt \leq C_\sigma \cdot \dedit$, which will conclude our claim.

	Observe that for $f \in \sigma$ \[\sum_{\tuple{x} \in A^{\ar(f)}}
	\left|f^\structA(\tuple{x}) - f^\structB(\phi     (\tuple{x}))\right| = 
	\sum_{\tuple{x} \in B^{\ar(f)}}
	\left|f^\structB(\tuple{x}) - f^\structA(\phi^{-1}(\tuple{x}))\right|  .\]

	Let $\delta := \frac{C_\sigma \cdot \dedit}{1 + C_\sigma \cdot \dedit}$, so $1-\delta = \frac{1}{1 + C_\sigma \cdot \dedit}$.
	To show $(1 + C_\sigma \cdot \dedit) \structA \succeq \structB$,
	we construct an overcast $\omega$ from $\structA$ to $(1-\delta)\structB$ as follows.
	With probability $(1-\delta)$ we map $\structA$ to $\structB$ with $\phi$;
	with probability $\delta$ we choose a tuple $(f,\tuple{x}) \in \tup{\structB}$ at random
	with probability proportional to its contribution in $d_1$, that is,
	$\frac{|f^\structB(\tuple{x})-f^\structA(\phi^{-1}(\tuple{x}))|}{\min(\|\structA_f\|_1,\|\structB_f\|_1)} \cdot \frac{1}{d_1}$,
	and we map all of $\structA$ uniformly at random to the elements of this tuple. That is, after choosing $(f,\tuple{x}) \in \tup{\structB}$, each tuple of $\structA_f$ gets mapped into $\tuple{x}$ with probability $\frac{1}{\ar(f)^{\ar(f)}}$ (assuming $\structA$ is loopless).
	Therefore, for each $(f,\tuple{x}) \in \tup{\structB}$:
	\[ \sum_{g \in B^A} \omega(g) \cdot f^\structA(g^{-1}(\tuple{x})) \geq
	(1-\delta) \cdot f^\structA(\phi^{-1}(\tuple{x}))
	+ 
	\delta \cdot 	\frac{|f^\structB(\tuple{x})-f^\structA(\phi^{-1}(\tuple{x}))|}{\min(\|\structA_f\|_1,\|\structB_f\|_1)} \cdot \frac{1}{d_1}
	\cdot \frac{\|\structA_f\|_1}{\ar(f)^{\ar(f)}}\]
	\[ \geq (1-\delta) f^\structA(\phi^{-1}(\tuple{x})) + \frac{\delta}{C_\sigma
		\cdot d_1} \cdot 	|f^\structB(\tuple{x})-f^\structA(\phi^{-1}(\tuple{x}))|
	\geq (1-\delta) f^\structB(\tuple{x}),\] 
	where the last inequality follows from $\frac{\delta}{C_\sigma \dedit} = \frac{1}{1+C_\sigma \dedit} = 1- \delta$.
	This shows that $\omega$ is indeed an overcast that certifies $\structA \succeq (1-\delta) \structB$.
\end{proof}

\begin{observation}\label{obs:closeToPliable}
	Let $\A$ be a $\tw$-pliable class.
	Let $\B$ be a class of structures such that
	for every $\structB \in \B$ there is an $\structA \in \A$
	with $\dopt(\structB,\structA) \leq f(\tw(\structB))$,
	for some function $f(n) \xrightarrow[n\to\infty]{} 0$.
	Then $\B$ is $\tw$-pliable.
\end{observation}
\begin{proof}
	Since $\A$ is $\tw$-pliable, for every $\eps>0$ there is a $k=k(\eps)$
	such that every structure $\structA$ in $\A$ is $\eps$-close
	to some structure of $\tw\leq k(\eps)$.
	To show that $\B$ is $\tw$-pliable, consider any $\eps>0$.
	Let $n_\eps$ be large enough so that $f(n) \leq \frac{\eps}{2}$ for $n \geq n_\eps$.
	Then for $\structB \in \B$,
	either $\tw(\structB) \leq n_\eps$
	or $\structB$ is $f(\tw(\structB)) \leq \frac{\eps}{2}$-close
	to some structure $\structA \in A$,
	which in turn is $\frac{\eps}{2}$-close
	to some structure of treewidth at most $k(\frac{\eps}{2})$.
	In either case $\structB$ is $\eps$-close to a structure of treewidth
	at most $\max(n_\eps, k(\frac{\eps}{2}))$.
\end{proof}

\subsection{Pliable examples}\label{subsec:pliable}

We give simple observations and examples: classes that are sufficiently close to
pliable classes (in edit or opt-distance) are themselves pliable.
We first consider consider simple examples with a fixed signature: graphs.

\begin{lemma}\label{lem:closeToPliable}
	Let $\A$ be a $\tw$-pliable class of graphs.
	Let $\B$ be a class of graphs such that
	every $H \in \B$ can be obtained from some $G \in \A$
	by adding or removing $f(|E(H)|)$ edges,
	for some function $f(m) \in o(m)$.
	Then $\B$ is $\tw$-pliable.
\end{lemma}
\begin{proof}
	By Lemma~\ref{lem:editToOpt},
	$\dopt(H,G) \leq 4 \dedit(H,G) \leq \frac{f(|E(H)|)}{|E(H)| - f(|E(H)|)} = f'(|E(H)|)$
	for some function $f'(n) \xrightarrow[n\to\infty]{} 0$.
	This function can be upper-bounded by a monotonic function $f''$ decreasing to~0,
	say $f''(x) \defeq \sup_{n \geq x} f'(n)$.
	Since $|E(H)|\geq \tw(H)$, we conclude $\dopt(H,G) \leq f''(|E(H)|) \leq f''(\tw(H))$.
	The claim follows by Observation~\ref{obs:closeToPliable}.
\end{proof}

Other simple examples arise from considering structures at opt-distance zero. This is a valued analogue of being homomorphically equivalent (see also \emph{valued cores} in~\cite{crz22:sicomp}).

\begin{example}\label{ex:bipartite}
	For every non-empty bipartite graph $G$, $\dopt(G,\lambda K_2) = 0$, for $\lambda = |E(G)|$.
	Therefore, since $\{\lambda K_2 \colon \lambda \in \QQ_{\geq 0}\}$ is trivially $\tw$-pliable,
	every class of bipartite graphs is $\tw$-pliable.
\end{example}
\begin{proof}
	A bipartite graph $G$ admits a homomorphism $h$ to $K_2$.
	This gives an overcast showing $G \overcasts \lambda K_2$: always map everything according to $h$.
	Conversely, mapping $\lambda K_2$ uniformly at random to edges of $G$ gives an overcast showing $\lambda K_2 \overcasts G$.
\end{proof}

\begin{example}\label{ex:triangle}
	Let $G$ be a 3-colourable graph such that every edge of $G$ occurs in exactly one triangle.
	Then $\dopt(G,\lambda K_3) = 0$ for $\lambda = |E(G)|/3$. Hence the class of all such graphs is $\tw$-pliable.
\end{example}
\begin{proof}
	A 3-colouring of $G$ corresponds to a homomorphism $h$ to $K_3$.
	Composing $h$ with a random rotation of $K_3$ gives an overcast from $G$ to $\lambda K_3$.
	Conversely, mapping $\lambda K_3$ to a uniformly random triangle in $G$
	covers each edge with probability $\frac{1}{\lambda}$,
	giving an overcast from $\lambda K_3$ to $G$.
\end{proof}

The above idea also implies that our results cannot be extended to \emph{finding} solutions.
This is analogous to the hardness of finding a 3-colouring of a graph that is homomorphically equivalent to $K_3$.

\begin{example}\label{ex:noFind}
	There is a class of weighted graphs $\A$ that is $\tw$-pliable, yet for some $\eps > 0$, there is no poly-time algorithm that finds a map $h \colon V(G) \to V(K_3)$ with $\val(h) \geq (1-\eps)\opt{G, K_3}$ for $G \in \A$, unless $\text{P}=\text{NP}$.	
\end{example}
\begin{proof}
	Let $\A$ be the class of weighted graphs $G$ satisfying $\dopt(G, \lambda K_3) = 0$ for $\lambda = \frac{\|G\|_1}{3}$.
	Suppose there is an algorithm as above for each $\eps > 0$.
	There are constants $\eps_0,d$ such that it is NP-hard to distinguish 3-colourable graphs of maximum degree $d$ from graphs where any map $h \colon V(G) \to V(K_3)$  miscolours more than an $\eps_0$ fraction of edges~\cite{GuruswamiK04}.
	We use our algorithm to solve the problem.
	Given an instance~$G$, let $G'$ be the weighted graph obtained by gluing a new triangle to every edge, and then assigning to every edge $e \in E(G')$ a weight $w(e)$ equal to the number of triangles it occurs in.
	Note $1 \leq w(e) \leq d$.
	Observe that if $G$ was 3-colourable, then $G'$ would be as well, hence $\opt{G', K_3} = \|G'\|_1$.
	Moreover, $G'$ would be in $\A$ by an argument similar as in Example~\ref{ex:triangle}.
	Hence running the algorithm on $G'$, we would find a 3-colouring $h$ which miscolours at most $\eps \opt{G', K_3} = \eps \|G'\|_1$ of the total weight $\|G'\|_1$.
	Since $w\geq 1$, it miscolours at most $\eps \|G'\|_1$ edges.
	Since $w \leq d$, this is at most $\eps d |E(G')| \leq 3\eps d |E(G)|$.
	Hence running the algorithm for $\eps = \frac{\eps_0}{3d}$ would find a colouring of the original graph $G$ that miscolours at most $\eps_0 |E(G)|$ edges.
	Therefore if we run this procedure for any $G$ (regardless of its 3-colourability), then it either outputs a colouring as above, or we can conclude that $G$ is not 3-colourable.
\end{proof}

\subsection{Non-pliable examples}
\label{sec:non-examples}

In this section we give examples of non-pliable classes.
In the process we show equivalent definitions of pliability (Lemmas~\ref{lem:niceToSelf} and \ref{lem:niceToSelfDual}).

We will use the following bound (in this section and in the proof of
Lemma~\ref{lemma:oriented-clique} in Section~\ref{sec:hardness}):
\[
  \sum_{i=0}^{\lambda m}\binom{m}{i} \leq 2^{H(\lambda)m},
\]
where $H(\lambda)$ is a function which satisfies $\lim_{\lambda\to 0} H(\lambda) = 0$; specifically, the binary entropy function $H(\lambda)=\lambda\log_2(\frac{1}{\lambda})-(1-\lambda)\log_2(1-\lambda)$.

Recall that by Lemma~\ref{lem:fragileIffNice}, for a class of Gaifman graphs $\Gg$,
the class $\AOver{\Gg}{(2)}$ of all structures over $\Gg$ is $\tw$-pliable if and only if
$\Gg$ is fractionally-$\tw$-fragile.
So the simplest examples of non-pliable classes are $\AOver{\Gg}{(2)}$ for some non-fractionally-fragile $\Gg$.
Fractional-$\tw$-fragility implies bounded expansion (a notion from the theory of sparse graphs introduced by Ne\v{s}et\v{r}il and Ossona de Mendez~\cite{NesetrilM08a}) and sublinear separators, e.g., 3-regular expander graphs are not fractionally-$\tw$-fragile, see~\cite{Dvorak16}.
Hence for $\Gg$ the class of all 3-regular graphs, $\AOver{\Gg}{(2)}$ is not $\tw$-pliable.

A somewhat more direct proof is to consider any class of 3-regular graphs of high girth.
Thomassen~\cite{Thomassen83} showed that such graphs behave much like graphs of high average degree.
We use essentially the same proof below:

\begin{lemma}\label{lem:girthAvgDegree}
	For $\delta>0$ and $g \in \NN$,
  every graph with average degree $\geq 2+\delta$ and girth $\geq 3g$
	has a minor with average degree $\geq g \delta + 2$. 
\end{lemma}
\begin{proof}
	Let $G=(V,E)$.
    Without loss of generality assume that $G$ is connected (otherwise take the component with the largest average degree).
	Let $A_1,\dots,A_m$ be a partition of $V$ into parts of size $|A_i| \geq g$ that induce connected subgraphs, 
  with $m$ maximum among such partitions (clearly one exists with $m=1$).
	
	We claim that each set $A_i$ induces a tree.
	Indeed, consider any spanning tree $T$ of $G[A_i]$ and let $e$ be an edge of $G[A_i]$ outside of $T$.
	Then $T+e$ contains a unique cycle, which must have length $\geq 3g \geq 2g$.
	Hence one can remove $e$ and some other edge from this cycle to split it into two intervals with $\geq g$ vertices.
	Removing these two edges from $T+e$ splits it into two components spanning $A_i$ with $\geq g$ vertices each.
	Hence $A_i$ could be replaced with the vertex sets of these two components, contradicting the choice of $m$.
	
	Similarly, we claim that every two sets $A_i, A_j$ are connected by at most one edge.
	Otherwise two such edges together with spanning trees of $A_i$ and $A_j$ would form a
	unicyclic graph, which could be split as above into three connected parts with $\geq g$ vertices each.
	
	Let $G'=(E',V')$ be the graph resulting by contracting the sets $A_i$.
	Since we contract sets of $\geq g$ vertices, $|V'| = m \leq \frac{|V|}{g}$.
	Since no two edges get identified and no loop gets created/removed in the process,
	the number of contractions is equal to $|E|-|E'|$ and to $|V|-|V'|$.
	Hence
	$|E'| = |E| - |V| + |V'| \geq (\frac{2+\delta}{2} - 1)|V| + |V'| \geq (g\frac{\delta}{2} + 1) |V'|$,
	so $G'$ has average degree $\geq g \delta + 2$.
	(We note that each $G[A_i]$ had diameter $< 2g-1$, as otherwise it could be split into two parts;
	hence the minor we obtain is relatively \emph{shallow}).
\end{proof}

\begin{proposition}\label{prop:delta}
	Let $\delta>0$ and let $\Gg$ be a class of graphs with unbounded girth
	and average degree $\geq 2+\delta$.
	Then $\Gg$ is not fractionally-$\tw$-fragile.
\end{proposition}
\begin{proof}
  Suppose that $\Gg$ is fractionally-$\tw$-fragile. Then, by Lemma~\ref{lem:fragile}, for $\eps = \frac{\delta}{2(2+\delta)}$ there is a $k=k(\eps)$ such that every graph in $\Gg$ has a subset $F\subseteq E(G)$ with $|F|\leq \eps |E(G)|$ such that $\tw(G-F)\leq k$.
	Let $G \in \Gg$ be a graph with girth $\geq \frac{12k}{\delta}$.
	Let $F$ be as above.
	Then $\tw(G-F) \leq k$ and $2|E(G-F)| \geq (1-\eps) \cdot 2|E(G)| \geq (1-\eps)(2+\delta)|V(G)| = (2+\frac{\delta}{2}) |V(G-F)|$.
	Therefore, $G-F$ has average degree $\geq (2+\frac{\delta}{2})$ and girth $\geq 3 \cdot \frac{4k}{\delta}$, 
	so by Lemma~\ref{lem:girthAvgDegree} it has a minor with average degree $\geq 2+ \frac{4k}{\delta} \frac{\delta}{2} > 2k$.
	But a minor of $G-F$ must have treewidth at most $\tw(G-F)\leq k$, so average degree $\leq 2k$, a contradiction.
\end{proof}

We now turn to classes of structures with a fixed signature $\sigma$.
We will show that the class of tournaments (orientations of complete graphs) is not $\tw$-pliable (or equivalently, $\size$-pliable, by Theorem~\ref{thm:niceWrt}), in contrast to cliques and dense graphs (Example~\ref{ex:clique} and Theorem~\ref{thm:denseIsNice}).

\begin{lemma}\label{lem:niceIsMad}
	Let $\Gg$ be a class of graphs of unbounded average degree.
	Let $\A$ be the class of (unweighted) orientations of graphs in $\Gg$.
	Then $\A$ is not size-pliable.
\end{lemma}

In order to prove Lemma~\ref{lem:niceIsMad}, we will need some definitions (only for this subsection) and alternative characterisations of size-pliability.

\begin{definition}
  For a graph parameter $\param$,
  let $\bar{\param}$ be the parameter defined as $\bar{\param}(G) := \max_i(\param(G_i))$,
  where $G_i$ are the connected components of $G$.
\end{definition}

For example, if $\param$ is $\size{}$, then $\bar{\param}$ is $\cc$ (max component size).
All the other parameters we consider ($\cc$, $\td$, $\tw$, $\Hadwiger$) satisfy $\bar{\param}=\param$.

\begin{definition}
  A parameter $\param$ is \emph{good} if 
  $\bar{p}$-pliability is the same as size-pliability on classes of structures with bounded signatures.
\end{definition}

We use this definition to state the next few lemmas in full generality.
Theorem~\ref{thm:niceWrt} shows that the parameters $\size$, $\cc$, $\td$, $\tw$, $\Hadwiger$ are good.

For two $\sigma$-structures $\structA,\structB$ and a function $g\colon A \to B$,
we define $\Img(g)$ to be the $\sigma$-structure on $B$
with $f^{\Img(g)}(\tuple{x}) := \min\left(f^{\structA}(g^{-1}(\tuple{x})),f^\structB(\tuple{x})\right)$.
Note that $\Img(g) \subseteq \structB$ (meaning each tuple has value in $\Img(g)$ less than or equal its value in $\structB$).

\begin{lemma}\label{lem:niceToSelf}
	Let $\param$ be a good parameter.
	Then a class of $\sigma$-structures $\A$ is size-pliable if and only if $\forall_{\eps>0} \exists_k \forall \structA \in \A$ there is an overcast $\omega$ from $\structA$ to $(1-\eps) \structA$ such that every $g\colon A \to A$ in its support has $\param(\Img(g)) \leq k$.
\end{lemma}
\begin{proof}
	For one direction, suppose that for every $\eps > 0$ there is an integer $k$ such that all $\structA \in \A$ have an overcast $\omega$ from $\structA$ to $(1-\eps)\structA$ such that every $g: A \to A$ in its support has $\param(\Img(g))\leq k$.
	Then for these $\eps,k,\structA$ we can take $\structB$ to be the disjoint union of rescaled structures $\structB_g := \omega(g) \Img(g)$.
	We have $\bar{\param}(\structB) \leq k$.
	The overcast $\omega$ naturally induces overcasts showing $\structA \overcasts \structB \overcasts (1-\eps)\structA$.
	Namely, we can define an overcast $\omega'$ from $\structA$ to $\structB$ by letting $\omega'(g')=\omega(g)$ for $g'$ mapping $A$ to $B_g \subseteq B$ just as $g$ maps $A$ to $\Img(g) \subseteq \structA$.
	We can also define an overcast $\omega''$ from $\structB$ to $(1-\eps)\structA$ by letting $\omega''(g'')=1$ for one function $g''$ mapping each $\structB_g \subseteq \structB$ to $\Img(g) \subseteq \structA$.
	Hence $\dopt(\structA,\structB) \leq \eps + \Oh(\eps^2)$ (recall $1\pm \eps$ is close to $e^{\pm\eps}$), which concludes the proof that $\A$ is $\bar{\param}$-pliable.
    Since we assume that $\param$ is a good parameter, $\A$ is size-pliable.
	
	In the other direction, suppose $\A$ is size-pliable, meaning for every $\eps > 0$ there is an integer $k$ such that all $\structA$ have a $\structB$ with $\dopt(\structA,\structB) \leq \eps$ and $|B| \leq k$.
	This means there are overcast $\omega$ and $\omega'$ showing $\structA \overcasts e^{-\eps}\structB$ and $\structB \overcasts e^{-\eps}\structA$, respectively.
	Then composing $\omega$ with $\omega'$ gives an overcast from $\structA$ to $(1-2\eps) \structA$ (since $e^{-2\eps} \geq 1-2\eps$), with the property that all images of functions $g$ in the support are of size at most $|B|\leq k$, which implies $\param(\Img(g))$ is bounded by some function of~$k$ (namely $\max \param(H)$ over all $k$-vertex graphs $H$).
\end{proof}

We can now use Farkas' lemma to deduce another equivalent formulation: 

\begin{lemma}\label{lem:niceToSelfDual}
	Let $\param$ be a good parameter.
	Then a class of $\sigma$-structures $\A$ is \underline{not} size-pliable if
  and only if $\exists_{\eps>0} \forall_{k \in \NN}$ there is a pair of $\sigma$-structures $\structA \in \A$ and $\structC$ with $C=A$, such that for every $g \colon A \to C$ with $\param(\Img(g))\leq k$, $\val(g) < (1-\eps)\val(\id)$.\quad(Here $\id$ is the identity map from $A$ to $C=A$).
\end{lemma}
\begin{proof}
	By Lemma~\ref{lem:niceToSelf}, $\structA$ is not size-pliable if and only if $\exists_{\eps>0} \forall_{k \in \NN}$ 
	the following LP over variables $\{\omega(g) \colon g \in V\}$, where $V :=
  \{g \in A^A \colon \param(\Img(g))\leq k\}$, has no non-negative rational solution:

	\begin{align*}
		&\sum_{g \in V}\omega(g) f^{\structA}(g^{-1}(\tuple{x})) \geq (1-\eps)f^{\structA}(\tuple{x}) &\forall (f,\tuple{x})\in \tup{\structA}\\
		&\sum_{g \in V}\omega(g) = 1 &
	\end{align*}
	
	By applying Lemma~\ref{lem:farkas1}, this is equivalent to the existence of a non-negative vector $(y(f,\tuple{x}))_{(f,\tuple{x})\in \tup{\structA}}$ such that
	$$\sum_{(f,\tuple{x})\in \tup{\structA}}y(f,\tuple{x})f^{\structA}(g^{-1}(\tuple{x})) < (1-\eps)\sum_{(f,\tuple{x})\in \tup{\structA}}y(f,\tuple{x})f^{\structA}(\tuple{x}) \quad\quad\quad \forall g \in V$$
	Let $\structC$ be the $\sigma$-structure on $C=A$ with $f^\structC(\tuple{x}) := y(f,\tuple{x})$.
	Then the above inequality is restated as follows (interpreting $g \in V$ and $\id$ as maps from $\structA$ to $\structC$):
	\[\val(g) < (1-\eps)\val(\id) \quad\quad\quad \forall {g \in V}.\qedhere\]
\end{proof}

\begin{remark}\label{remark:niceToSelfDual}
	The structures $\structA,\structC$ obtained above can be assumed to satisfy $\Gaifman(\structA) = \Gaifman(\structC)$ without loss of generality, because for any $(f,\tuple{x})\in\tup{\structA}$ such that one of $f^\structA(\tuple{x})$ or $f^\structC(\tuple{x})$ is zero, decreasing the other to zero will not change $\val(\id)$ and can only decrease $\val(g)$.
\end{remark}

\begin{proof}[Proof of Lemma~\ref{lem:niceIsMad}]
	Let $\eps$ be a constant to be chosen later ($\frac{1}{10}$ will do).
	Given any $k$, let $G \in \Gg$ be a graph with $m\geq 20 \cdot \binom{k}{2}$ edges, $n$ vertices, and average degree $\frac{2m}{n}\geq 100\log_2 k$.
	Let $\structA$ be a random orientation of $G$ (each edge is independently oriented in either direction with probability $\frac{1}{2}$).
	We claim that with positive probability $\structA$ admits no map $g \colon A \to A$ to itself with image of size at most $k$ such that $\val(g) \geq (1-\eps)\val(\id)$.
	This will prove that $\A$ satisfies the conditions of Lemma~\ref{lem:niceToSelfDual} and hence is not size-pliable.
	
	If a map as above existed, it would imply the existence of an oriented graph $D$ (with at most one arc between every two vertices) on at most $k$ vertices and a function $g \colon A \to V(D)$
	with $\val(g) \geq (1-\eps) m$.
	Observe that $\val(g)$ is the number of arcs of $\structA$ that are correctly
  mapped by $g$ (i.e., to an arc of $D$ with the same orientation).
	Hence there would be a set $F$ of at most $\eps m$ arcs of $\structA$ such that
	$g$ maps all arcs of $\structA-F$ correctly.
	Let us bound the probability that there exists such $D,F,g$.
	The number of possible $D$ is $\leq 3^{\binom{k}{2}}$;
	the number of possible $F$ is $\leq \sum_{i=0}^{\eps m}\binom{m}{i} \leq
  2^{H(\eps)m}$;
	the number of possible $g$ is $\leq k^{n}$.
	Note that $\frac{2m}{n} \geq 100\log_2 k$ and $3^{\binom{k}{2}} \leq 2^{m/10}$ by our choice of $G$.
	For fixed $D,F,g$, the probability that $g$ maps all arcs of $\structA-F$ correctly to $D$
	is at most $(\frac{1}{2})^{(1-\eps)m}$.
	Hence in total the probability that some such $D,F,g$ exist is at most 
	\[ 3^{\binom{k}{2}} \cdot 2^{H(\eps)m} \cdot k^n \cdot 2^{-(1-\eps)m} \leq
	2^{n\log_2 k - (1-\eps-H(\eps)-\frac{1}{10})m} \leq 2^{-(1-\eps-H(\eps)-\frac{1}{10}-\frac{1}{50}) \cdot m}.\]
	This is less than 1 for $\eps$ small enough so that $1-\eps-H(\eps)-\frac{1}{10} - \frac{1}{50} > 0$.	
\end{proof}

Finally, not all classes of bounded degree give pliable classes, even with a fixed signature.

\begin{lemma}\label{lem:avg2}
	Let $\Gg$ be a class of graphs with unbounded girth and average degree $\geq 2+\delta$ ($\delta > 0$).
	Let $\A$ be the class of (unweighted) orientations of graphs in $\Gg$.
	Then $\A$ is not size-pliable.\looseness=-1
\end{lemma}
\begin{proof}
	We show there exists an $\eps$ such that for all $k$, there is an orientation $\structA \in \A$ of a graph in $\Gg$ such that every function $g \colon \structA \to \structA$ with $\cc(\Img(g))\leq k$ has $\val(g) < (1-\eps)\val(\id)$.
	We choose $\eps$ later depending on $\delta$ only.
	
	For any given $k$, let $G \in \Gg$ be a graph of girth $>k$.
	Let $m=|E(G)|$.
	Let $\A$ be a random orientation of $G$:
	every edge is independently oriented in one direction or the other.
	We claim that the probability that there exists a $g \colon \structA \to \structA$ with $\cc(\Img(g))\leq k$ and $\val(g) \geq (1-\eps)\val(\id)$ is strictly less than one
	(so there exists an orientation that satisfies our goal).
	Note that $\val(\id) = m$ and $\val(g)$ is the number of arcs in $\structA$ that are mapped correctly (to an arc in $\structA$ with the same orientation);
	moreover, since the graph underlying $\structA$ has girth $>k$ and $\cc(\Img(g)) \leq k$,
	$g$ must map into an oriented forest (disjoint union of trees).
	So the event is equivalent to the following: there exists a set $F\subseteq E(G)$
	with $|F|\leq \eps m$ and a function $g \colon \structA \to \structA$ which maps all arcs of $\structA-F$ correctly into an oriented forest in $\structA$.
	
	The probability of this event can be union-bounded by the sum over $F\subseteq E(G)$ with $|F|\leq \eps m$ of the probability that all of $\structA-F$ can be mapped correctly into a subdigraph.
	The number of such $F$ is $\sum_{i=0}^{\eps m}\binom{m}{i} \leq 2^{H(\eps)
  \cdot m}$;
	It remains to bound, for a fixed $F$, the probability that $\structA - F$ can be mapped correctly.
	
	Consider a fixed $F\subseteq E(G)$ with $|F|\leq \eps m$. If $\structA - F$ can be mapped correctly into an oriented forest in $\structA$, then in particular it admits a homomorphism to $C_3$, the directed cycle digraph with three arcs. Let $T$ be a spanning forest of $\structA-F$ (a union of spanning trees of each connected component of $G-F$).
	There exists exactly one homomorphism from the edges of $T$ in $\structA$ to $C_3$ (up to rotations in $C_3$ of each component);
	every remaining edge in $\structA-F-E(T)$ closes an oriented cycle, so it has at most one orientation which allows to extend this unique homomorphism to it.
	Hence the probability that  $\structA - F$ admits a homomorphism to $C_3$ is at most 
	$(\frac{1}{2})^{m'}$ where $m'=m-|F|-|E(T)| \geq m-\eps m - |V(G)| \geq (1-\eps-\frac{2}{2+\delta}) m = (\frac{\delta}{2+\delta}-\eps) m$.
	
	All in all, the probability of our original event is at most $2^{H(\eps) \cdot m} \cdot (\frac{1}{2})^{m'} \leq 2^{-(\frac{\delta}{2+\delta}-\eps-H(\eps))m}$.
	Hence it suffices to choose $\eps$ small enough so that $\eps+H(\eps) < \frac{\delta}{2+\delta}$.
\end{proof}

\section{Pliable structures admit a PTAS: proof of Theorem~\ref{thm:main1}}
\label{sec:tract}

We first define the Sherali-Adams LP hierarchy~\cite{Sherali1990} for \HOM.
Let $(\structA,\structB)$ be an instance of \HOM over a signature $\sigma$ and let $k\geq\max_{f\in\sigma}\ar(f)$. 
For a tuple $\tuple{x}$, we denote by $\toset{\tuple{x}}$ the set of elements appearing in $\tuple{x}$. 
We write $\binom{A}{\leq k}$ for the set of subsets of $A$ with at most $k$ elements.
The \emph{Sherali-Adams relaxation of level $k$}~\cite{Sherali1990} of $(\structA,\structB)$ is the linear program given in
Figure~\ref{fig:sa}, which has
one variable $\lambda(X,s)$ for each $X\in\binom{A}{\leq k}$ and each $s\colon X \to B$.

We denote by $\optfrac{k}{\structA,\structB}$ the optimum value of this linear program. 

\begin{figure}[h!]
\fbox{\parbox{0.98\textwidth}{
\begin{align*}
\max \sum_{(f,\tuple{x}) \in \tup{\structA}, \; s\colon \toset{\tuple{x}} \to B}& \lambda(\toset{\tuple{x}},s) f^{\structA}(\tuple{x}) f^{\structB}(s(\tuple{x})) \\
\lambda(X,s) &= \sum_{r\colon Y\to B,\, r|_{X}=s} \lambda(Y,r) & \text{for $X \subseteq Y\in\textstyle\binom{A}{\leq k}$ and $s\colon X\to B$}  \\
\sum_{s\colon X \to B} \lambda(X,s) &= 1   &\text{for $X\in\textstyle\binom{A}{\leq k}$} \\
\lambda(X,s) &\geq 0   &\text{for $X\in\textstyle\binom{A}{\leq k}$ and $s\colon X \to B$}
\end{align*}
}}
\caption{The Sherali-Adams relaxation of level $k$ of $\HOM$ instance $(\structA,\structB)$.} 
\label{fig:sa}
\end{figure}

\begin{observation}
\label{obs:scaling}
Let $\structA$ be a $\sigma$-structure, $\lambda\geq 0$ and $k\geq\max_{f\in\sigma}\ar(f)$. Then for all $\sigma$-structures $\structC$, we have 
$\opt{\lambda\structA,\structC}=\lambda\opt{\structA,\structC}$ and $\optfrac{k}{\lambda\structA,\structC}=\lambda\optfrac{k}{\structA,\structC}$.  
\end{observation}  

\begin{definition}
Let $\structA$ and $\structB$ be $\sigma$-structures and $k\geq\max_{f\in\sigma}\ar(f)$. We write $\structA\overcasts_k \structB$ 
if, for all $\sigma$-structures $\structC$, we have $\optfrac{k}{\structA,\structC}\geq \optfrac{k}{\structB,\structC}$. 
\end{definition}

The proof of the following is analogous to the proof
of~\cite[Proposition~5.2]{crz22:sicomp}. For completeness, it is given in
Appendix~\ref{sec:overcast-sa}.

\begin{proposition} \label{prop:overcast-sa}
	Let $\structA$ and $\structB$ be $\sigma$-structures and $k\geq\max_{f\in\sigma}\ar(f)$. 
	If there is an overcast from $\structA$ to $\structB$ then $\structA\overcasts_k \structB$. 
\end{proposition}

Using Observation~\ref{obs:scaling} and Proposition~\ref{prop:overcast-sa}, we are ready to prove the following.

\begin{proposition}\label{prop:tract}
	Let $\structA$ be a $\sigma$-structure.
	Let $\eps \geq 0$ and $k\geq\max_{f\in\sigma}\ar(f)$.
	Suppose that there exists a $\sigma$-structure $\structB$ 
	such that $\dopt(\structA, \structB) \leq \eps$ and $\tw(\structB)\leq k$. 
	Then, for every $\sigma$-structure $\structC$, we have that 
	\[
    	\opt{\structA,\structC}\ \leq\ 
	    \optfrac{k}{\structA,\structC}\ \leq\ 
    	(1+\Oh(\eps)) \opt{\structA,\structC}.
	\] 
\end{proposition}
\begin{proof}
	The left-hand side inequality is from the definition of Sherali-Adams.
	For the right-hand side inequality, observe first that, by definition of $\dopt$, 
	$\structA \overcasts e^{-\eps} \structB$ and $\structB \overcasts e^{-\eps} \structA$.
	By Proposition~\ref{prop:overcast}, there is an overcast from $\structB$ to $e^{-\eps}\structA$,
	so by Proposition \ref{prop:overcast-sa}, it follows that $\structB\overcasts_{k} e^{-\eps}\structA$.
	By Observation \ref{obs:scaling}, we have that
	$\optfrac{k}{\structB,\structC} \geq e^{-\eps} \optfrac{k}{\structA,\structC}$. 
	Since $\tw(\structB)\leq k$, we have
	$\optfrac{k}{\structB,\structC}=\opt{\structB,\structC}$
  -- this follows, for example, from~\cite[Theorem~5.8]{crz22:sicomp}.%
  \footnote{Our definition of the LP slightly differs from~\cite{crz22:sicomp}, 
	where there are additional variables $\lambda(f,\tuple{x},s)$ associated with tuples $(f,\tuple{x})$ with $f^\structA(\tuple{x})>0$. 
	However, since we are assuming without loss of generality that $k\geq\max_{f\in\sigma}\ar(f)$, the two definitions are equivalent. 
	}
	Since moreover 	$\structA \overcasts e^{-\eps} \structB$, by Observation \ref{obs:scaling}, 
	it follows that $\opt{\structA,\structC} \geq e^{-\eps} \opt{\structB,\structC}$.
	Together, $\opt{\structA,\structC} \geq e^{-\eps} \opt{\structB,\structC} = e^{-\eps} \optfrac{k}{\structB,\structC} \geq e^{-2\eps} \optfrac{k}{\structA,\structC}$.
	Hence $\optfrac{k}{\structA,\structC} \leq e^{2\eps} \opt{\structA,\structC}$. 
\end{proof}

Since $\optfrac{k}{\structA,\structC}$ can be computed in time $(|A| \cdot |C|)^{\Oh(k)}$,
this concludes the proof of Theorem~\ref{thm:main1}.

\section{Fractional fragility}
\label{sec:frag}
To give Dvo\v{r}\'ak's definition of fractional fragility~\cite{Dvorak16}
we first define $\eps$-thin distributions.

\begin{definition}
	Let $\mathcal{F}$ be a family of subsets of a set $V$ and $\eps > 0$.
	We say that a distribution $\pi$ over $\mathcal{F}$ is \emph{$\eps$-thin} if $\Pr_{X \sim \pi} [v \in X] \leq \eps$ for all $v \in V$.
\end{definition}

\begin{definition}\label{def:fragile}
  For a graph parameter $\param$ and a number $k$, we define a
  \emph{$(\param\leq k)$-modulator} of a graph $G$ to be a set $X \subseteq
  V(G)$ such that $\param(G-X) \leq k$. A \emph{fractional $(\param\leq
  k)$-modulator} is a distribution $\pi$ of such modulators $X$.
  We say that a class of
  graphs $\Gg$ is \emph{fractionally-$\param$-fragile} if for every $\eps>0$
  there is a $k$ such that every $G \in \Gg$ has an $\eps$-thin fractional
  $(\param\leq k)$-modulator. 
  We can analogously define $(\param\leq k)$-edge-modulators $F \subseteq E(G)$
  and \emph{fractionally-$\param$-edge-fragility}.
\end{definition}

One crucial property of fractional fragility is that it allows a dual definition by a variant of Farkas' lemma (cf.~Appendix~\ref{sec:farkas} for details); this is already implicit in~\cite[Lemma~6]{DvorakS20}.

\begin{lemma}\label{lem:thin}
	Let $\mathcal{F}$ be a family of subsets of a set $V$. The following are equivalent:
	\begin{itemize}
		\item there is an $\eps$-thin distribution $\pi$ of sets in $\mathcal{F}$;
		\item for all non-negative weights $\left(w(v)\right)_{v\in V}$, there is an $X \in \mathcal{F}$ such that $w(X) \defeq \sum_{x \in X} w(x)\leq \eps \cdot w(V)$.
	\end{itemize}
\end{lemma}

Thus a class of graphs $\Gg$ is fractionally-$\tw$-fragile if and only if
for every $\eps>0$ there is a $k$ such that for every graph $G \in \Gg$
and every vertex-weight function $w$,
one can remove a set of vertices of weight at most $\eps \cdot w(V)$ to obtain a graph with $\tw\leq k$.

Another useful property of fractional fragility is that the edge version is equivalent to the vertex version,
for most parameters of interest.
Recall that each parameter we consider ($\Hadwiger{}$, $\tw{}$, $\tw{}$, $\td{}$, $\cc{}$, $\size{}$) is \emph{monotone}, meaning $\param(H)\leq \param(G)$ for $H$ a subgraph of $G$; and that the average degree $\frac{2|E(G)|}{|V(G)|}$ is bounded by a function of $\param$.

\begin{lemma}\label{lem:fragile}
	Let $\param$ be a monotone graph parameter such that the average degree $\frac{2|E(G)|}{|V(G)|}$ of a graph is bounded by a function of $\param(G)$.
	Let $\Gg$ be a class of graphs.
	Then the following are equivalent:
	\begin{itemize}
		\item $\Gg$ is fractionally-$\param$-fragile;
		\item $\Gg$ is fractionally-$\param$-edge-fragile;		
		\item $\forall_{\eps>0} \exists_k\forall_{G\in\Gg} \forall_{w\colon V(G) \to \QQ_{\geq0}} \exists_{X\subseteq V(G)}\ w(X)\leq \eps w(V(G))\ \text{and}\ \param(G-X)\leq k$;
		\item $\forall_{\eps>0} \exists_k\forall_{G\in\Gg} \forall_{w\colon E(G) \to \QQ_{\geq0}} \exists_{F\subseteq E(G)}\ w(F)\leq \eps w(E(G))\ \text{and}\ \param(G-F)\leq k$.
	\end{itemize}
\end{lemma}

\begin{proof}
	$(i)$ is equivalent to $(iii)$ and $(ii)$ is equivalent to $(iv)$ by Lemma~\ref{lem:thin}.

	It is easy to see that $(i)$ implies $(iv)$: 
	suppose for every $\eps>0$ there is a $k$ such that every $G\in\Gg$ has an $\eps$-thin fractional $(\param\leq k)$-modulator $\pi$.
	Let $w \colon E(G) \to \QQ_{\geq 0}$ be any edge-weight function.
	If we take a set $X$ from the distribution $\pi$ and remove the set $F$ of all edges incident to $X$, this yields a graph with $\param(G-F)\leq k$.
	Every vertex is in $X$ with probability $\leq \eps$, so every edge is in $F$ with probability $\leq 2\eps$.
	Hence the expected weight of $F$ is $\leq 2\eps w(E(G))$.
	So there exists a set $F \subseteq E(G)$ such that $w(F) \leq 2\eps w(E(G))$ and $\param(G-F) \leq k$.
	
	It remains to show that $(iv)$ implies $(iii)$.
	Let $f\colon \NN\to\NN$ be such that $\frac{2|E(G)|}{|V(G)|} \leq f(\param(G))$ for all graphs $G$.
	
	We first show that  $(iv)$ implies that $\Gg$ has bounded \emph{maximum average degree} $\mad(G) := \max_{H \subseteq G} \frac{2|E(H)|}{|V(H)|}$.
	Indeed, let $k:=k(\eps)$ be a number satisfying $(iv)$ for $\eps=\frac{1}{2}$.
	Then for any $G \in \Gg$ and any $H \subseteq G$, let $w\colon E(G) \to \QQ_{\geq 0}$ assign 1 to edges in $H$ and 0 to edges not in $H$.
	By assumption there is a set $F \subseteq E(G)$ such that $w(F)\leq \eps w(E(G))$ and $\param(G-F)\leq k$.
	Let $F' := F \cap E(H)$; then $|F'|=w(F)\leq \eps w(E(G)) = \eps|E(H)|$ and $\param(H-F')\leq\param(G-F)\leq k$.
	Hence $(1-\eps)|E(H)| \leq |E(H-F')| \leq \frac{f(k)}{2} \cdot |V(H-F')|$, which means $\frac{2|E(H)|}{|V(H)|}\leq \frac{f(k(\eps))}{1-\eps} = 2 f(k(\frac{1}{2}))$.
	That is, every subgraph $H$ of every graph $G$ in $\Gg$ has average degree at most $D:=2 f(k(\frac{1}{2}))$.
	
	This implies that every subgraph has some vertex of degree at most $D$ (this
  is called the \emph{degeneracy} of the graph: it is upper bounded by $\mad$).
	Hence every graph $G$ in $\Gg$ has an orientation $\vec{G}$ with maximum in-degree at most $D$
	(obtained by iteratively finding a vertex of degree at most $D$, orienting all remaining edges towards it, and removing the vertex).

	To show $(iii)$, let $\eps>0$, $k':=k(\frac{\eps}{D})$, $G\in\Gg$.
	Choose an orientation $\vec{G}$ of $G$ with maximum in-degree at most $D$.
	Given $w\colon V(G)\to\QQ_{\geq0}$, we can define $w'\colon E(G) \to \QQ_{\geq 0}$ as $w'(uv) := w(v)$ if $uv$ is directed towards $v$.
	By assumption, there is a set of edges $F$ such that $\param(G-F)\leq k'$ and
	\[ w'(F) \leq \frac{\eps}{D} w'(E(G)). \]
	Let $X := \{ v \colon \exists\,uv \in F\text{ directed towards }v\}$; then $G-X \subseteq G-F$, so $\param(G-X) \leq k'$.
	Note that
	\[w'(E(G)) = \sum_{\vec{uv} \in E(\vec{G})} w(v) = \sum_{v \in V(G)} \indeg(v) \cdot w(v) \leq D \cdot w(V(G))\]
	and
	\[w'(F) = \sum_{\vec{uv} \in F} w(v) \geq \sum_{v \in X} w(v) = w(X)\]
	Hence 
	\[w(X) \leq w'(F) \leq \frac{\eps}{D} w'(E(G)) \leq \eps \cdot w(V(G)).\]
	This concludes the proof that $(iv)$ implies $(iii)$.
\end{proof}

Dvo\v{r}\'ak and Sereni~\cite[Theorem~31]{DvorakS20} showed that graphs of bounded treewidth are fractionally-$\td$-fragile.
It follows from a result of DeVos et al.~\cite[Theorem 1.2]{DeVos+04} that for
$H$-minor-free graphs (for any graph $H$, that is classes of graphs that exclude some minor, or equivalently, classes of bounded Hadwiger number) 
are fractionally-$\tw$-fragile.\footnote{In fact, as shown by
Dvo\v{r}\'{a}k~\cite{Dvorak20}, a proof of van~den~Heuvel et~al.~\cite[Lemma
4.1]{HeuvelMQRS17} can be adapted to show this without the Graph Minors
Structure Theorem.}
These two facts establish the following equivalence, by~\cite[Lemma 12]{Dvorak16}.

\begin{theorem}[\cite{DvorakS20,DeVos+04,Dvorak16}]\label{thm:fragileWrt}
	The following are equivalent for a class of graphs $\Gg$:\looseness=-1
	\begin{itemize}
			\item $\Gg$ is fractionally-$\td$-fragile;
			\item $\Gg$ is fractionally-$\tw$-fragile;
			\item $\Gg$ is fractionally-$\Hadwiger$-fragile.
	\end{itemize}\vspace*{-1ex}
\end{theorem}

\subsection{Fragility implies pliability: proof of Theorem~\ref{thm:fragileIsNice}}
We denote by $G\uplus H$ the disjoint union of graphs $G$ and $H$.
All graph parameters $\param$ we consider satisfy 
$\param(G \uplus H) = \max(\param(G),\param(H))$ for all $G,H$ (that is: $\cc$, $\td$, $\tw$, $\Hadwiger$, excluding only size; we never consider fractional-size-fragility, as it is equivalent to just bounded size).

\begin{lemma}\label{lem:fragileToNice}
	Let $\param$ be a graph parameter such that $\param(G \uplus H) = \max(\param(G),\param(H))$ for all $G,H$.
	Let $\A$ be a class of structures of bounded arity $r$ such that the class $\Gg$ of their Gaifman graphs is fractionally-$\param$-fragile.
	Then $\A$ is $\param$-pliable.
\end{lemma}
\begin{proof}
	For $\eps>0$, let $\eps' := \frac{\eps}{1+\eps} \cdot \frac{1}{r}$.
    By definition of fractional-$\param$-fragility,
	$\exists_{k(\eps)} \forall_{G\in\Gg}$ $G$ has an $\eps'$-thin
  fractional $(\param \leq k)$-modulator.
	Let $\structA \in \A$ be a structure with Gaifman graph $G \in \Gg$.
	By assumption, $G$ has a fractional $(\param\leq k)$-modulator~$\pi$
	such that for every $v \in V(G)$, $\Pr_{X\sim \pi}[v \in X] \leq \eps'$.
	For $X \subseteq V(G) = A$ in the support of $\pi$, let $\structB_X$ be the rescaling of $\structA - X$ by a factor of $\pi(X)$;
	let $\structB$ be the disjoint union of all $\structB_X$.
	Since each $X$ in the support of $\pi$ is a $(\param\leq k)$-modulator and $\param$ is closed under disjoint union, $\param(\Gaifman(\structB)) \leq k$.
	
	We define overcasts $\omega\colon \structA \to \structB$ and $\omega' \colon \structB \to (1-r\eps')\structA$.
	The first, $\omega$, maps $\structA$ identically to each component
  $\structB_X$ of $\structB$ with probability $\pi(X)$ (vertices of $\structA$ in $X$ are mapped arbitrarily in the same component).
	The second, $\omega'$, deterministically maps each component $\structB_X$ of $\structB$ identically to $\structA$.
	To check that $\omega'$ is indeed an overcast, consider a tuple $(f,\tuple{x}) \in \tup{\structA}$.
	The tuple is covered by its copies in $\structB_X$ with weight $\pi(X) \cdot f^\structA(X)$ for all $X$ which do not intersect $\tuple{x}$.
	In total, the fraction of $f^\structA(\tuple{x})$ lost is hence exactly $\Pr_{X \sim \pi} [X \cap \tuple{x} \neq \emptyset]$, which is (by union bound and by the assumption $|\tuple{x}|\leq r$) at most $\eps' r$.
	Since $1-\eps'r = \frac{1}{1+\eps} \geq e^{-\eps}$,
	we have $\structA \overcasts \structB \overcasts (1-\eps' r) \structA \overcasts e^{-\eps} \structA$,
	which means $\structB$ is a structure at opt-distance $\leq \eps$ from $\structA$.
\end{proof}

\noindent
This concludes Theorem~\ref{thm:fragileIsNice}: structures on fractionally-$\tw$-fragile graphs are $\tw$-pliable.

\subsection{Pliability vs fragility: proof of Lemma~\ref{lem:fragileIffNice}}
\label{sec:fragproof}

For Lemma~\ref{lem:fragileIffNice}, we need the other direction than in
Theorem~\ref{thm:fragileIsNice}: that if all structures on Gaifman graphs in $\Gg$ are $\tw$-pliable, then $\Gg$ is fractionally-$\tw$-fragile.
To do this, intuitively, we consider, for a graph $G \in \Gg$, a structure $\structA$ where each edge is used by a different symbol of the signature.
If we have a structure $\structB$ (of bounded treewidth) close to $\structA$ in opt-distance, this implies overcasts from $\structA$ to $e^{-\eps}\structB$ and from $\structB$ to $e^{-\eps} \structA$; 
composing the two gives an overcast from $e^{+\eps}\structA$ to $e^{-\eps}\structA$ 
in which (since each edge is used by a different symbol) an edge can only be covered by itself.
This shows that the overcasts are mostly injective and that $\structB$, sandwiched between $e^{+\eps}\structA$ and  $e^{-\eps}\structA$, must be close in edit distance.
The bounded treewidth of $\structB$ then implies that the graph $G$ underlying $\structA$ is in fact fractionally-$\tw$-edge-fragile, which by Lemma~\ref{lem:fragile} concludes the proof.

The formal proof of Lemma~\ref{lem:fragileIffNice} follows. In fact, we prove
prove the statement for any reasonable parameter, including $\cc$, $\td$, $\tw$, and $\Hadwiger$;
the conclusion is the edge variant of fractional fragility, but the two are equivalent by Lemma~\ref{lem:fragile}.

\begin{lemma}[Lemma~\ref{lem:fragileIffNice} more generally] 
  \label{lem:fragileIffNiceGen}
	Let $\param$ be a monotone graph parameter such that $\param(G \uplus H) = \max(\param(G),\param(H))$.
	For every integer $r \geq 2$, 
	a class of graphs $\Gg$ is fractionally-$\param$-edge-fragile if and only if
	  $\AOver{\Gg}{(r)}$ is $\param$-pliable.
\end{lemma}
\begin{proof}
	If $\Gg$ is fractionally-$\param$-fragile,
	then by Lemma~\ref{lem:fragileToNice} $\AOver{\Gg}{(r)}$ is $\param$-pliable.
	
	For the other direction, suppose $\AOver{\Gg}{(r)}$ is $\param$-pliable:
  \[\forall_{\eps>0}\exists_{k(\eps)} \forall_{\structA \in \AOver{\Gg}{(r)}} \exists_{\structB \vphantom{ \AOver{\Gg}{(r)}}}\  \param(\structB) \leq k \text{ and } \dopt(\structA,\structB) \leq \eps.\]
	Let $\eps > 0$ and let $k:=k(\frac{\eps}{2})$.
	For a graph $G \in \Gg$, let $\sigma$ be the signature with a different binary symbol $f_e$ for each $e \in E(G)$.
	Let $\structA$ be the $\sigma$-structure with domain $V(G)$ and values $f_e^{\structA}(u,v) = 1$ if $\{u,v\}=e$, 0 otherwise.
	(The arity can be increased to exactly $r$ by adding dummy or repeated variables).
	By assumption, there is a $\sigma$-structure $\structB$ such that $\param(\structB) \leq k$ and $\dopt(\structA,\structB)\leq\frac{\eps}{2}$.
	Let $\omega,\omega'$ be overcasts from $\structA$ to $\exp(-\frac{\eps}{2}) \cdot \structB$ and from $\structB$ to $\exp(-\frac{\eps}{2})\cdot  \structA$, respectively.
	
  For $g$ with $\omega(g)>0$ and $g'$ with $\omega'(g')>0$, let $F_{gg'} \subseteq E(G)$ be the subset of edges $e$ such that $g'(g(e)) \neq e$ or $f^\structB_e(g(e)) = 0$.
	Since $g' \circ g$ is the identity on $E(G) - F_{gg'}$, the functions $g'$ and $g$ are bijections between this set and a subset of edges of $\Gaifman(\structB)$.
	Hence $G - F_{gg'}$ is isomorphic to a subgraph of $\Gaifman(\structB)$, which implies $\param(G-F_{gg'}) \leq k$.
	
	Let $e \in E(G)$. We claim that $\Pr\limits_{\substack{g\sim\omega\\g'\sim\omega'}} [e \in F_{gg'}] \leq \eps$.
	This holds essentially because the composition of $\omega$ and $\omega'$ is an overcast from $\structA$ to $\exp(-\eps) \cdot \structA$ and because the only edge with non-zero value of $f^\structA_e$ is $e$ itself.
	Formally, since $\omega$ is an overcast, we have:
  \[\text{for each }e_B\in E(\Gaifman(\structB)) \quad \quad \EX_{g\sim\omega} f^{\structA}_e(g^{-1}(e_B)) \geq \exp(-\textstyle\frac{\eps}{2}) f_e^\structB(e_B).\]
	Note that by construction of $\structA$, $f^{\structA}_e(g^{-1}(e_B)) = [g(e)=e_B]$.
	Hence for each $e_B\in E(\Gaifman(\structB))$,
  \[\EX_{g\sim\omega} [g(e)=e_B] \geq \exp(-\textstyle\frac{\eps}{2}) f_e^\structB(e_B).\]

	\noindent	
	Moreover, since $\omega'$ is an overcast, we have:
  \[\EX_{g'\sim \omega'} f^{\structB}_e(g'^{-1}(e)) \geq \exp(-\textstyle\frac{\eps}{2}) f_e^\structA(e).\]
  That is: \[\EX_{g'\sim \omega'}  \sum_{\substack{e_B \in E(\Gaifman(\structB))\\g'(e_B) = e}} f^{\structB}_e(e_B) \geq \exp(-\textstyle\frac{\eps}{2}).\]
	
	\noindent Putting the two together:
	\begin{align*}
	\Pr_{\substack{g\sim\omega\\g'\sim\omega'}} [e \not\in F_{gg'}] &=
 	  \Pr_{\substack{g\sim\omega\\g'\sim\omega'}} [g'(g(e)) = e \text{ and } g(e) \in E(\Gaifman(\structB))] \\
 	&=\EX_{\substack{g\sim\omega\\g'\sim\omega'}} \sum_{\substack{e_B \in E(\Gaifman(\structB))\\g'(e_B) = e}} [g(e) = e_B] \\
 	&\geq \EX_{g'\sim\omega'} \sum_{\substack{e_B \in E(\Gaifman(\structB))\\g'(e_B) = e}}  \exp(-\textstyle\frac{\eps}{2}) f_e^\structB(e_B)\\
 	&\geq \exp(-\eps) \geq 1-\eps. 
 	\end{align*}
 	
 	Therefore, we obtained a distribution of edge sets $F_{gg'} \subseteq E(G)$ such that $\param(G-F_{gg'}) \leq k$ 
 	satisfying $\Pr\limits_{\substack{g\sim\omega\\g'\sim\omega'}} [e \in F_{gg'}] \leq \eps$.
 	This is an $\eps$-thin fractional $(\param\leq k)$-edge-modulator.
\end{proof}

\section{From Hadwiger- to size-pliability: proof of Theorem~\ref{thm:niceWrt}}
\label{sec:sizeccproof}

In this section, we prove Theorem~\ref{thm:niceWrt}: pliability with respect to different parameters yields equivalent definitions.
The first half of Theorem~\ref{thm:niceWrt} follows easily
from already established results and a simple observation, cf.
Section~\ref{subsec:1.5easy}.
The second half of Theorem~\ref{thm:niceWrt} reduces to showing that structures of bounded treedepth with a bounded signature are $\size$-pliable.
The strategy for the proof is similar to a proof of Ne\v{s}et\v{r}il and Ossona de~Mendez~\cite[Corollary 3.3]{NesetrilM06} that relational structures of bounded treedepth have bounded cores.
However the argument is much more intricate due to the fact that we consider valued structures:
the statement that there are only finitely many structures of size at most $C$, for every $C$, is not true anymore.
The main difficulty is proving an approximate version of it.

\subsection{Treewidth-, treedepth-, and Hadwiger-pliability}
\label{subsec:1.5easy}

The first half of Theorem~\ref{thm:niceWrt}, that is, the equivalence of
$\param$-pliability for $\param \in \{\tw,\td,\Hadwiger\}$, will follow (as
detailed in Corollary~\ref{cor:amen-arity} below) from the equivalence of
fractional-$\param$-fragility for these parameters
(Theorem~\ref{thm:fragileWrt}), the fact that fragility implies pliability
(Lemma~\ref{lem:fragileToNice}), and transitivity of pliability, in the
following sense.

\begin{observation}[Transitivity of pliability]\label{obs:niceIsTransitive}
	Let $\mathcal{A}$ be a class of structures with signatures from a set $\Sigma$.
	Suppose $\mathcal{A}$ is $\param$-pliable
	and for each $k$, $\{\structA \colon \param(\structA)\leq k\}$ is $\param'$-pliable,
	where $\structA$ runs over all structures with signatures in $\Sigma$.
	Then $\mathcal{A}$ is $\param'$-pliable.
\end{observation}
\begin{proof}
    Intuitively, this hold because $\dopt$ is a pseudometric.
	Formally, suppose a class $\A$ is $\param$-pliable.
	Then every $\structA \in \A$ is $\frac{\eps}{2}$-close (in $\dopt$ distance) to some $\structB$ with $\param(\structB) \leq k$ (for some $k$ depending on $\frac{\eps}{2}$).
	By assumption, every  $\structB$ with $\param(\structB) \leq k$ is $\frac{\eps}{2}$-close to some $\structC$ with $\param'(\structC) \leq k'$ (for some $k'$ depending on $\frac{\eps}{2}$ and $k$).
	Hence $\structA$ is $\eps$-close to some structure $\structC$ with $\param'(\structC) \leq k'(\frac{\eps}{2},k(\frac{\eps}{2}))$, which only depends on $\eps$.
\end{proof}

\begin{corollary}\label{cor:amen-arity}
	Let $\mathcal{A}$ be any class of structures. The following are equivalent:
	\begin{itemize}
		\item $\mathcal{A}$ is $\td$-pliable;
		\item $\mathcal{A}$ is $\tw$-pliable;
		\item $\mathcal{A}$ is $\Hadwiger$-pliable.
	\end{itemize}	
\end{corollary}
\begin{proof}
	Since $\td(G) \geq \tw(G) + 1 \geq \Hadwiger(G)$ for any graph $G$, each
  bullet point implies the next by Observation~\ref{obs:trivial}.
	It suffices to show that $\Hadwiger$-pliability implies $\td$-pliability.
	By Observation~\ref{obs:niceIsTransitive} it suffices to show that for every $k$,
	the class $\A$ of all structures with Hadwiger number at most $k$ (and arbitrary signatures) is $\td$-pliable.
	These are structures whose Gaifman graphs exclude the clique $K_{k+1}$ as a minor.
	Their Gaifman graphs are thus fractionally-$\td$-fragile by Theorem~\ref{thm:fragileWrt}.
	Since their Gaifman graphs do not include cliques $K_{k+1}$ the arity of symbols 
	with non-zero tuples is bounded by $k$.
	By Lemma~\ref{lem:fragileToNice}, this implies that $\A$ is $\td$-pliable
	(high-arity symbols with no non-zero tuples can be ignored).
\end{proof}

\subsection{From cc-pliability to size-pliability}
To show the second half of Theorem~\ref{thm:niceWrt},
i.e., the equivalence of $\td$-pliability, $\cc$-pliability, and $\size$-pliability
(for structures with bounded signatures),
it will be easier to first focus on the latter two.
Since there are only finitely many distinct signatures of bounded size and arity,
we can focus on a single fixed signature (as finite unions of pliable classes are pliable).

Since $\cc\leq\size$,
by Observation~\ref{obs:trivial} we have that $\size$-pliability implies
$\cc$-pliability. The rest of this section is devoted to proving that
$\cc$-pliability implies $\size$-pliability (for a single fixed signature).
This (and in fact equivalence of the two)
would be trivial if there were only a bounded number of distinct values of tuples,
since then there can be only a bounded number of components up to isomorphism,
and isomorphic components can be merged.

\begin{observation}\label{obs:merge}
	For any structure $\structA$ and numbers $\lambda_1,\dots,\lambda_n \in \Qn$,
	the disjoint union $\lambda_1 \structA \uplus \dots \uplus \lambda_n \structA$ is
	equivalent (i.e., at $\dopt$-distance zero) to $\lambda \structA$, where $\lambda=\lambda_1+\dots+\lambda_n$.
\end{observation}
\begin{proof}
	An overcast in one direction deterministically maps each component $\lambda_i \structA$ to $\lambda \structA$ identically.
	An overcast in the other direction maps $\lambda \structA$ to the component $\lambda_i \structA$ with probability $\lambda_i/\lambda$.
\end{proof}

For a structure $\structA$ with components of bounded size and $\qplus$-values,
we can try to change the values slightly to find a structure $\structB$ at small
edit distance which uses a bounded number of different values (and then proceed
as above). This works if the ratio of the maximum value to the minimum non-zero
value is bounded. If this ratio is large, we could try to change the extremely
small values to zero, hoping the edit distance is small (relative to the
extremely large values). However, this does not always work: consider structures
$\structA$ with few large values and many small values (for example a structure
having $2^i$ tuples of value $2^{n-i}$, for $i=0\dots n$). So the general case
cannot be reduced to the case of finitely many distinct values just by finding a
structure close in edit distance. Nevertheless, instead of requiring the
modified structure $\structB$ to have a bounded number of components up to
isomorphism, it suffices to require a bounded number of components up to
rescaling (two structures $\structB_1,\structB_2$ being the same up to rescaling
if $\structB_1=\lambda \structB_2$ for some $\lambda>0$). This minor
weakening turns out to be sufficient to fix our problem. We formalise this first
as a statement on sequences of vectors of bounded dimension (which will encode a
sequence of components of bounded size).

\begin{lemma}\label{lem:vectors}
	Let $d \in \NN, \eps>0$. There is a $k$ such that for every sequence of vectors $\vvec{1}{},\dots,\vvec{n}{} \in \QQ_{\geq0}^d$,
	there is a sequence $\wvec{1}{},\dots,\wvec{n}{}\in \QQ_{\geq0}^d$ such that for each coordinate $i = 1,\dots,d$, ${\sum_{j=1\dots n} |\vvec{j}{i} - \wvec{j}{i}|} \leq \eps{\sum_{j=1\dots n} \vvec{j}{i}}$, and such that up to rescaling, there are only $k$ distinct vectors in $\wvec{1}{},\dots,\wvec{n}{}$.
\end{lemma}

\begin{proof}
	\newcommand\mass[1]{\textrm{mass}_{#1}}
	The proof is by induction on $d$.
	Let $d\in\NN, \eps>0$ and consider a sequence $v_1,\dots,v_n \in \QQ_{\geq 0}^d$.
	For $d=1$, the sequence already has only one vector up to rescaling
	(or two, if it contains the zero vector), so let $d\geq 2$.
	
	Let $J = \{1,\dots,n\}$.
	For a subset $X\subseteq J$, denote $\mass{i}(X) := \sum_{j \in X} \vvec{j}{i}$.
	We focus on the first two coordinates and in particular the ratio of the second to the first.
	For $c \in \RR$, let $J_{<c} := \{ j \in J \colon \vvec{j}{2} < c \cdot \vvec{j}{1} \}$.
	Define $J_{\leq c}, J_{>c}, J_{\geq c}$ analogously.
	
	Let $c$ be maximum such that $\mass{2}(J_{<c}) \leq \frac{\eps}{3} \cdot \mass{2}(J)$.
	For $j \in J_{<c}$, let $w_j$ be the vector obtained from $v_j$ by zeroing the 2nd coordinate.
	The resulting difference is
	\[\sum_{j \in J_{<c}} |\vvec{j}{2}-\wvec{j}{2}| \leq \frac{\eps}{3} \sum_{j \in J} \vvec{j}{2}.\]
	
	By maximality of $c$ we have
  \[\mass{2}(J_{\leq c}) > \frac{\eps}{3} \cdot \mass{2}(J).\]
	Observe that the left hand side can be bounded as follows:
  \[\mass{2}(J_{\leq c}) \leq c \cdot \mass{1}(J_{\leq c}) \leq c \cdot \mass{1}(J)\]
	and similarly the right hand side can be bounded as follows, for $c' := c \cdot \frac{3d}{\eps^2}$:
  \[\frac{\eps}{3} \cdot \mass{2}(J) \geq \frac{\eps}{3} \cdot  \mass{2}(J_{\geq c'}) \geq c \cdot \frac{d}{\eps} \cdot \mass{1}(J_{\geq c'}).\]
	Altogether, this implies
  \[ c \cdot \mass{1}(J)  >  c \cdot \frac{d}{\eps} \cdot \mass{1}(J_{\geq c'}),\]
	which after rearranging gives
	\[  \mass{1}(J_{\geq c'}) < \frac{\eps}{d}\cdot \mass{1}(J).\]
	For $j \in J_{\geq c'}$, let $w_j$ be the vector obtained from $v_j$ by zeroing the 1st coordinate.
	The resulting difference is $\sum_{j \in J_{\geq c'}} |\vvec{j}{1}-\wvec{j}{1}| \leq \frac{\eps}{d} \sum_{j \in J} \vvec{j}{1}$.
	
	The only remaining vectors, in $J_{mid} := J \setminus (J_{< c}  \cup J_{\geq
  c'})$, satisfy $c \cdot \vvec{j}{1} \leq \vvec{j}{2} < c' \cdot \vvec{j}{1}$.
  We can round down their 2nd coordinate to $c \cdot \vvec{j}{1}$ times an integer power of $e^{\eps/3}$.
	That~is, for $j \in J_{mid}$, let $w_j$ be the vector obtained from $v_j$ by decreasing the 2nd coordinate to $\wvec{j}{2} := c \cdot e^{a\eps/3} \cdot \vvec{j}{1}$ with $a \in \NN$ maximum such that $\wvec{j}{2} \leq \vvec{j}{2}$.
	Observe that $a \geq 0$ and since $c \cdot e^{a\eps/3} \cdot \vvec{j}{1} \leq \vvec{j}{2} \leq c' \cdot \vvec{j}{1}$,
	we have $e^{a \eps/3} \leq \frac{c'}{c} = \frac{3d}{\eps^2}$ and thus $a \leq \frac{3}{\eps} \cdot \ln(\frac{3d}{\eps^2})$.
	Note also that $1 \geq \frac{\wvec{j}{2}}{\vvec{j}{2}} > e^{-\eps/3}$, hence $\frac{|\vvec{j}{2}-\wvec{j}{2}|}{\vvec{j}{2}} \leq 1-e^{-\eps/3} < \frac{\eps}{3}$, so the resulting difference is $\sum_{j\in J_{mid}} |\vvec{j}{2}-\wvec{j}{2}| \leq \frac{\eps}{3} \sum_{j \in J} \vvec{j}{2}$
	
	To summarise, all vectors $w_j$ satisfy $w_j \leq v_j$ (coordinate-wise) and when limited to their first two coordinates as $\left(\begin{smallmatrix}\wvec{j}{1}\\\wvec{j}{2}\end{smallmatrix}\right)$, are either multiples of $(\begin{smallmatrix}1\\0\end{smallmatrix})$ (if $j \in J_{< c}$), or multiples of $(\begin{smallmatrix}0\\1\end{smallmatrix})$ (if $j \in J_{\geq c'}$),
	or multiples of $(\begin{smallmatrix}1\\c \cdot e^{a\eps/3}\end{smallmatrix})$, for some $a \in \{0,1,\dots,K\}$ for $K := \lfloor \frac{3}{\eps} \cdot \ln(\frac{3d}{\eps^2}) \rfloor$.
	The resulting differences in the first and second coordinate, respectively, are bounded by as
	\[\sum_{j \in J_{\geq c'}} |\vvec{j}{1}-\wvec{j}{1}| \leq \frac{\eps}{d} \sum_{j \in J} \vvec{j}{1}\] 
	\[\sum_{j \in J_{<c}} |\vvec{j}{2}-\wvec{j}{2}| + \sum_{j \in J_{mid}} |\vvec{j}{2}-\wvec{j}{2}| \leq 
	(\frac{\eps}{3} + \frac{\eps}{3}) \sum_{j\in J} \vvec{j}{2}.\] 
	
	We replace the sequence $\vvec{j}{}$ with the sequence $\wvec{j}{}$ and repeat the same process for the 1st and $i$-th coordinate, for $i=3,\dots,d$.
	Since each step only zeroes the 1st coordinate of some vectors and decreases the other coordinates,
	the final resulting sequence $\wvec{j}{}$, when compared to the initial sequence $\vvec{j}{}$ satisfies:
  \[\sum_{j \in J} |\vvec{j}{1}-\wvec{j}{1}| \leq (d-1) \cdot \frac{\eps}{d} \sum_{j \in J} \vvec{j}{1}\]
	\[\sum_{j \in J} |\vvec{j}{i}-\wvec{j}{i}| \leq 
	(\frac{\eps}{2} + \frac{\eps}{2}) \sum_{j\in J} \vvec{j}{i} \quad\mbox{for }i=2,\dots,d.\]
	Each vector $\wvec{j}{}$ either has its 1st coordinate zeroed, or all its other coordinates are determined as $\wvec{j}{1}$ times one of $2+K$ possible ratios.
	Among vectors with $\wvec{j}{1} \neq 0$ there are thus at most $(2+K)^{d-1}$ different vectors, up to rescaling.
	The vectors with $\wvec{j}{1} = 0$ can be inductively reduced as $(d-1)$-dimensional vectors to ${\wvec{j}{}}'$ containing $k(d-1,\frac{\eps}{3})$ distinct vectors up to rescaling (where $k(d-1,\frac{\eps}{3})$ is the constant $k$ obtained by inductive assumption for $d-1$~and~$\frac{\eps}{3}$) and satisfying\looseness=-1
  \[\sum_{j\in J} |\wvec{j}{i} - {\wvec{j}{i}}'| \leq \frac{\eps}{3} \sum_{j\in J}\wvec{j}{i} \leq \frac{\eps}{3} \sum_{j\in J} \vvec{j}{i}.\]
	Altogether, the difference is bounded by $\frac{2\eps}{3}+\frac{\eps}{3}=\eps$ and the number of distinct vectors up to rescaling is bounded by $(2+K)^{d-1}+k(d-1,\frac{\eps}{3})$.
\end{proof}

\begin{lemma}\label{lem:niceComponentToNiceSize}
	For a fixed signature $\sigma$ and $d \in \NN$,
	the class of $\sigma$-structures with maximum connected component size at most
  $d$ is size-pliable.
\end{lemma}
\begin{proof}
	We simply present each component $\structA_i$ of $\structA$ as a vector encoding the value of all tuples $(f,\tuple{x})\in \tup{\structA_i}$.
	The dimension of such a vector, for a component of size $d$, is $d':=\sum_{f \in \sigma} d^{\ar(f)}$.
	Smaller components can be treated as components of size $d$ by adding dummy vertices and tuples.
	
	For any $\eps>0$, let $\eps' := \frac{\eps/C_\sigma}{1+\eps/C_\sigma}$, where $C_\sigma = \max_{f\in \sigma}\ar(f)^{\ar(f)}$ .
	The previous lemma guarantees the existence of a number $k=k(\eps',d')$ such that for every $\sigma$-structure $\structA$ with $n$ components of size at most $d$, the corresponding vectors $\vvec{1}{},\dots,\vvec{n}{}$ are approximated by vectors $\wvec{1}{},\dots,\wvec{n}{}$ such that there are at most $k$ distinct vectors up to rescaling and such that, for $i=1\dots d'$, 
	 \[\ {\sum_{j=1\dots n} |\vvec{j}{i} - \wvec{j}{i}|} \leq \eps'{\sum_{j=1\dots n} \vvec{j}{i}}.\]
	 These vectors encode a $\sigma$-structure $\structB$ with only at most $k$ distinct components up to rescaling, all of size at most $d$,
	 which is hence (by Observation~\ref{obs:merge}) equivalent to a $\sigma$-structure $\structB'$ bounded in size by $k \cdot d$.
	 Moreover, the guarantee on $\eps'$ allows us to bound edit distance as follows:
	 \[ \sum_{j=1\dots n} |\vvec{j}{i} - \wvec{j}{i}| \leq \eps' \sum_{j=1\dots n} \vvec{j}{i} \leq  \eps' \sum_{j=1\dots n} \big( \min(\vvec{j}{i},\wvec{j}{i}) + |\vvec{j}{i} - \wvec{j}{i}|\big) \]
	 (for $i=1\dots d'$), hence the edit distance (as defined in Section~\ref{sec:distance}) is
	 \[ \dedit(\structA,\structB) \leq \max_{i=1\dots d'} \frac{\sum_{j=1\dots n} \left|\vvec{j}{i} - \wvec{j}{i}\right|}{\min( \sum_{j=1\dots n} \vvec{j}{i}, \sum_{j=1\dots n} \wvec{j}{i})} \leq \frac{\eps'}{1-\eps'} = \frac{\eps}{C_\sigma}.\]
	 By Lemma~\ref{lem:editToOpt}, $\dopt(\structA,\structB') \leq \dopt(\structA,\structB) + \dopt(\structB,\structB') = \dopt(\structA,\structB) \leq C_\sigma \cdot \dedit(\structA,\structB) \leq \eps$.
	 (While Lemma~\ref{lem:editToOpt} assumes the structures to be loopless, this can be ensured by replacing tuples with repeated elements like $(f,(x_1,x_1,x_2))$, say, with $(f',(x_1,x_2))$ for a new symbol $f'$).
\end{proof}

By Observation~\ref{obs:niceIsTransitive} (transitivity of pliability), we
conclude that for a fixed signature $\sigma$, if a class of $\sigma$-structures
$\A$ is cc-pliable then it is also size-pliable. 
Thus, we have shown equivalence of size-pliability and cc-pliability
(for structures of bounded signatures).

\subsection{From treedepth-pliability to size-pliability}\label{sec:treedepth}
In order to finish the proof of Theorem~\ref{thm:niceWrt},
it remains to show the equivalence of $\td$- and $\size$-pliability.
We do this by extending the above proof for $\cc$- and $\size$-pliability.

One of the main reasons for which treedepth is useful (and easier to work
with than, say, treewidth) is that the only way for a graph of small treedepth
to be large is to have many repeating parts, like in a large star graph (see
e.g.~\cite[Theorem 3.1]{NesetrilM06}). This implies that in a class of graphs of
bounded treedepth, homomorphic cores have bounded size. This does not extend
to weighted graphs or structures in general, but we can approximate the weights or values 
as before.

\begin{lemma}\label{lem:niceTdToNiceSize}
	For a fixed signature $\sigma$ and $d \in \NN$, the class of $\sigma$-structures $\{\structA \colon \td(\structA) \leq d\}$ is size-pliable.
\end{lemma}
\begin{proof}
	\def\pack{\mathrm{pack}}
	\def\unpack{\mathrm{unpack}}
	We prove by induction on $d$ that the statement holds for each signature $\sigma$.
	It suffices to prove the statement for connected $\sigma$-structures of treedepth at most $d$.	
	Indeed, this implies that disconnected $\sigma$-structures of treedepth at most $d$ are 
	  $\cc$-pliable, which we already know implies $\size$-pliability.
	
	For $d=1$, each component of the Gaifman graph is a single vertex and we are done.
	So let $d>1$ and assume that for each signature $\sigma$ and each $\eps>0$, there is a $k=k(d-1,\sigma,\eps)$ such that every $\sigma$-structure with treedepth $\leq d-1$ has a $\sigma$-structure of size $\leq k$ at opt-distance at most $\eps$.
	Let $\sigma$ be a signature and $\structA$ a $\sigma$-structure of treedepth $d$.
	Let $G$ be the Gaifman graph of $\structA$.
	Since it is connected, we can find a vertex $v \in V(G) = A$ such that $\td(G-v)=d-1$.
	
	We now define a new signature $\sigma'$ and a $\sigma'$-structure $\pack(\structA)$ whose Gaifman graph will be $G-v$, but will contain all the information about $\structA$.
	Let $\sigma' = \{(f,I) \colon f \in \sigma, I \subseteq \{1,\dots,\ar(f)\}\}$ and $\ar((f,I)) := \ar(f) - |I|$, for $(f,I)\in \sigma'$.
	For $\tuple{x} \in (A-v)^{\ar((f,I))}$, let $(f,I)^{\pack(\structA)}(\tuple{x}) := f^\structA(\tuple{x'})$ where $\tuple{x'} \in A^{\ar(f)}$ is the tuple obtained from $\tuple{x}$ by introducing $v$ at positions $I$.
	Note that $\sigma'$ is bounded: $|\sigma'| = \sum_{f \in \sigma} 2^{\ar(f)}$.
	
	The $\sigma'$-structure $\pack(\structA)$ has treedepth $d-1$, so by inductive assumption there is a $\sigma'$-structure $\structB$ at opt-distance at most $\eps$ with size at most $k=k(d-1,\sigma',\eps)$.
	We define $\unpack(\structB)$ to be the $\sigma$-structure with domain $B \cup \{v\}$
	and $f^{\unpack(\structB)}(\tuple{x}) := (f,I)^\structB(\tuple{x'})$, where $I$ is the set of positions in $\tuple{x}$ containing $v$ and $\tuple{x'}$ is the tuple obtained by removing these positions.
	It is straightforward to check that $\unpack(\pack(\structA))$ is equal to $\structA$ and that for any $\sigma'$-structures $\structA',\structB'$ we have $\dopt(\unpack(\structA'),\unpack(\structB')) \leq \dopt(\structA',\structB')$, hence 
	$$\dopt(\structA,\unpack(\structB)) = \dopt(\unpack(\pack(\structA)),\unpack(\structB)) \leq \dopt(\pack(\structA),\structB) \leq \eps.$$
	Hence $\unpack(\structB)$ is a $\sigma$-structure at opt-distance $\leq \eps$ from $\structA$ of size $\leq k+1$.
\end{proof}

By Observation~\ref{obs:niceIsTransitive} (transitivity of pliability), this
shows that, for a class of $\sigma$-structures, $\td$-pliability implies
$\size$-pliability. Since $\td\leq\size$, Observation~\ref{obs:trivial} shows
that $\size$-pliability implies $\td$-pliability, thus concluding the proof of Theorem~\ref{thm:niceWrt}:

\begin{corollary}
	A class of $\sigma$-structures is $\td$-pliable if and only if it is size-pliable.
\end{corollary}

\section{Hyperfinite classes are fractionally fragile: proof of Theorem~\ref{thm:hyperfinite}}
\label{sec:hyperfinite}

Recall that class of graphs is if \emph{hyperfinite}
if for every $\eps>0$ there is a $k \in \NN$ such that
every graph in the class can be turned into a graph with
connected components of size at most $k$ by removing
an at-most-$\eps$ fraction of all edges.%
\footnote{In other work, the definition of \emph{hyperfinite} often considers
	the number of removed edges divided by the total number of vertices.
  However, they only deal	with bounded degree graphs (in which the number of
  edges is linear in the number of vertices), which makes
  the two definitions equivalent.}
A class of graphs is \emph{monotone} if it is closed under taking subgraphs.
In this section, we prove the following result.

\begin{theorem*}[Theorem~\ref{thm:hyperfinite} restated]
	Let $\Gg$ be a monotone class of graphs. The following are equivalent:
	\begin{itemize}
		\item $\Gg$ is hyperfinite;
		\item $\Gg$ is fractionally-$\tw$-fragile and has bounded degree;
		\item $\Gg$ is fractionally-$\cc$-fragile;		
		\item $\AOver{\Gg}{(r)}$ is $\cc$-pliable for any $r\geq 2$.
	\end{itemize}
\end{theorem*}

\medskip
\noindent
The last two bullets are shown equivalent by Lemma~\ref{lem:fragileIffNiceGen};
the middle two bullets were shown equivalent by Dvo\v{r}\'ak~\cite[Observation 15, Corollary 20]{Dvorak16}.
It remains to prove their equivalence with the first bullet point.

\begin{lemma}\label{lem:hyperfiniteIsFragile}
	Let $\Gg$ be a monotone class of graphs.
	$\Gg$ is hyperfinite if and only if it is fractionally-$\cc$-fragile.
\end{lemma}
\begin{proof}
	Hyperfiniteness of a monotone class $\Gg$ is equivalent to hyperfiniteness of 0-1-edge weighted graphs in $\Gg$:
	\[\forall_{\eps>0} \exists_k \forall_{G\in\Gg} \forall_{w\colon E(G) \to \{0,1\}} \exists_{F\subseteq E(G)}\ w(F)\leq \eps w(E(G))\ \text{and}\ \cc(G-F)\leq k.\]	
	Hence it is trivially implied by the edge version of fractional-$\cc$-fragility (which allows arbitrary nonnegative weights) in Lemma~\ref{lem:fragile}.
	It remains to show the other direction.
	
	By definition of hyperfiniteness, for every $\eps>0$ there is a $k=k(\eps)$ such that
	for all graphs $G \in \Gg$, one can remove a set of edges $F$ with $|F|\leq \eps|E(G)|$ so that $\cc(G-F) \leq k$.
	Observe that graphs in $\Gg$ have degree bounded by $\Delta := 2k(\frac{1}{2})$;
	otherwise, a graph with degree $\geq 2k(\frac{1}{2})+1$ would contain a star with that many edges as a subgraph and removing half of these edges always leaves a component with at least $k(\frac{1}{2})+1$ edges and vertices.
	
	We aim to show that
	\[\forall_{\eps>0} \exists_k \forall_{G\in\Gg} \forall_{w\colon E(G) \to \QQ_{\geq0}} \exists_{F\subseteq E(G)}\ w(F)\leq \eps w(E(G))\ \text{and}\ \cc(G-F)\leq k.\]
	For $\eps>0$, let $\eps'$ be chosen later and let $k'=k(\eps')$.
	Let $G \in \Gg$ and $w \colon E(G) \to \QQ_{\geq 0}$.
	We want to find a set $F \subseteq E(G)$ such that  $w(F) \leq \eps w(E(G))$ and $\cc(G-F) \leq k'$.	
	Note that our task would be trivial if the weights of all edges were within a constant factor $\alpha$ of each other:
	just set $\eps' = \frac{\eps}{\alpha}$, find $F\subseteq E(G)$ such that $|F| \leq \eps' |E(G)|$ and $\cc(G-F) \leq k'$ and conclude that $w(F) \leq \alpha\eps' w(E(G)) = \eps w(E(G))$.
	
	In general, let us partition the edges of $G$ into buckets depending on their weight:
	for $i \in \ZZ$, let $B_i := \{ e \in E(G) \mid
  \left(\frac{\eps}{6\Delta}\right)^{i} \geq w(e) >
  \left(\frac{\eps}{6\Delta}\right)^{i+1}\}$ (edges with weight zero can be
  removed without loss of generality).\footnote{Note that all but a finite
  number of $B_i$'s will be empty.}
	For $L :=\lceil \frac{3}{\eps}\rceil$, we will remove every $L$-th bucket from $G$.
	That is, for $j \in \{0,\dots,L-1\}$, let $B'_j := \bigcup_{i \in \ZZ} B_{iL + j}$.
	Let $j^* \in \{0,\dots,L-1\}$ be such that $w(B'_{j^*})$ is minimum; 
	since $B'_0 \cup \dots \cup B'_{L-1}$ is a partition of $E(G)$, $w(B'_{j^*}) \leq \frac{1}{L} w(E(G)) \leq \frac{\eps}{3} w(E(G))$.
	We can thus remove the edges $B'_{j^*}$ from $G$.
	Since this removes every $L$-th bucket,
	the remaining edges are partitioned into blocks $C_i := B_{iL+j^*+1} \cup \dots \cup B_{iL+j^*+L-1}$ of $L-1$ buckets for $i \in \ZZ$.
	Each block contains weights within a constant factor of each other: 
	$\min\{w(e) \colon e \in C_{i}\} \geq \left(\frac{\eps}{6\Delta}\right)^{L-1} \cdot \max\{w(e) \colon e \in C_i\}$.
	Moreover, since there is a gap of one bucket in between one block and the next,
	$\max\{w(e) \colon e \in C_{i+1}\} < \frac{\eps}{6\Delta} \cdot \min\{w(e) \colon e \in C_i\}$.
	
	The latter property allows us to disconnect the blocks from each other.
	Indeed, 
	for each $C_i$ with increasing $i$ (starting from the smallest $i$ such that $C_i$ is non-empty),
	we shall remove all remaining edges on the boundary of $C_i$:
	\[F_i := \{ e \colon e \in C_j \text{ for some }j > i,\ e \text{ shares a vertex with some }e' \in C_i\}.\]
	Since $|F_i| \leq 2\Delta |C_i|$ and $\max\{w(e) \colon e \in F_i\} < \frac{\eps}{6\Delta} \cdot \min\{w(e) \colon e \in C_i\}$, we have $w(F_i) < \frac{\eps}{3} w(C_i)$.
	In total, for $F := \bigcup_{i \in \ZZ} F_i$, we have $w(F) < \frac{\eps}{3} w(E(G))$.
	For all $i\in \ZZ$, the edges of $C_i - F$ are disjoint from edges of $C_j - F$ for all $j > i$.
	Since the sets $C_i-F$ partition edges of $G-B'_{j^*}-F$, this means that every connected component of 
	$G-B'_{j^*}-F$ is contained in one of the edge sets $C_i-F$.
	
	Finally, since 	$\min\{w(e) \colon e \in C_i - F\} \geq \alpha \cdot \max\{w(e) \colon e \in C_i - F\}$ for $\alpha := \left(\frac{\eps}{3\Delta}\right)^{L-1}$,
	we have reduced our problem to the trivial case when weights are all within a constant factor of each other.
	That is, let $\eps' := \frac{\alpha\eps}{3}$.
	For each $i\in \ZZ$, let $G_i$ be the subgraph of $G-B'_{j^*}-F$ formed from connected components contained in $C_i-F$.
	By assumption, there is a set $F'_i \subseteq E(G_i)$  such that $|F'_i| \leq \eps' |E(G_i)|$ and $\cc(G_i-F'_i) \leq k(\eps') = k'$.
	Then $w(F'_i) \leq \max\{w(e) \colon e \in E(G_i)\} \cdot |F'_i| \leq \frac{1}{\alpha} \cdot \min\{w(e) \colon e \in E(G_i)\} \cdot \eps' \cdot |E(G_i)| \leq \frac{\eps'}{\alpha} w(E(G_i)) = \frac{\eps}{3} w(E(G_i))$.
	In total, for $F' := \bigcup_{i \in \ZZ} F'_i$ we have $w(B'_{j^*} \cup F \cup F') \leq (\frac{\eps}{3} + \frac{\eps}{3} + \frac{\eps}{3}) w(E(G))$
	and $\cc(G-B'_{j^*} - F - F') \leq k'$.
\end{proof}

\section{Dense graphs are pliable: proof of Theorem~\ref{thm:denseIsNice}}
\label{sec:dense}

Our goal is to prove Theorem~\ref{thm:denseIsNice}.

\begin{theorem*}[Theorem~\ref{thm:denseIsNice} restated]
	Let $c>0$ and let $\A$ be a class of graphs with at least $cn^2$ edges.
	Then $\A$ is $\tw$-pliable.
	Consequently, $\HOM(\A,-)$ admits a PTAS.
\end{theorem*}

In order to do so, we prove the following result.

\begin{theorem}\label{thm:dense-size}
	Let $c>0$. The class of (unweighted, undirected) graphs with at least $cn^2$ edges is size-pliable.
\end{theorem}

Theorem~\ref{thm:dense-size} implies Theorem~\ref{thm:denseIsNice}. Indeed, since
$\tw\leq\size$, by Observation~\ref{obs:trivial} we have $\size$-pliability implies
$\tw$-pliability. By Theorem~\ref{thm:main1}, $\tw$-pliability implies a PTAS.

We start with simple examples of dense graphs. Observe that large cliques
can be arbitrarily well approximated by cliques of constant size $\lceil\frac{2}{\eps}\rceil$ (up to
normalising total edge weights).

\begin{example}\label{ex:clique}
	Let $0<\eps<1$ and let $n,k \geq \frac{2}{\eps}$.
	Then $\dopt(K_n,\lambda K_k) \leq \eps$, for $\lambda = \binom{n}{2}/\binom{k}{2}$.
\end{example}

\begin{proof}
	For $n,k\geq 2$, define an overcast $\omega$ by taking a random function $V(K_n) \to V(K_k)$ (each vertex is placed independently uniformly at random).
	Then for each $e \in E(K_k)$,
	\[\EX_{g \sim \omega} |g^{-1}(e)| = \sum_{e' \in E(K_n)} \EX_{g \sim \omega} [g(e')=e] =\binom{n}{2} \frac{2}{k^2} = \lambda \cdot \binom{k}{2} \cdot \frac{2}{k^2} = (1-\frac{1}{k}) \lambda . \]
	Therefore $K_n \overcasts (1-\frac{1}{k}) \lambda K_k$.
	Symmetrically $\lambda K_k \overcasts (1-\frac{1}{n}) K_n$.	
	Since $1-x \geq e^{-2x}$ for $0\leq x\leq \frac{1}{2}$,
	this means $\dopt(K_n, \lambda K_k) \leq \frac{2}{\min(n,k)}$.
	Consequently if $n , k \geq \frac{2}{\eps}$, then $\dopt(K_n, \lambda K_k) \leq \eps$.
\end{proof}

In particular, this means the class $\A$ consisting of all clique graphs is
size-pliable. This corresponds to an easy PTAS for graph $\HOM(\A,-)$: the maximum
graph homomorphism from $K_n$ to $G$ is well approximated by finding the maximum
graph homomorphism from a constant size $K_k$ to $G$ and mapping $K_n$ randomly to the resulting $\leq k$ vertices in $G$.
The situation is very different for Densest Subgraph problems, because they
disallow choosing two equal vertices in $G$ (see Observation~\ref{obs:densest-hard}).

As another important example, consider Erd\H{o}s-R\'{e}nyi random graphs $G(n,p)$ (for constant $p\in(0,1)$; each pair in $\binom{n}{2}$ becomes an edge independently with probability $p$).
Any two such graphs are similar to each other (and in fact to $p K_n$, as well
as to $\lambda K_k$ for constant $k$ and suitable $\lambda$); more precisely, we
have:

\begin{example}\label{ex:random}
	Let $p,\eps>0$ be constants.
	Let $G_1,G_2$ be independent Erd\H{o}s-R\'{e}nyi random graphs $G(n,p)$.
	Then $\Pr[\dopt(G_1,G_2) < \eps] \to 1$ as $n\to\infty$.
\end{example}
\begin{proof}[Proof sketch]
	Let $k$ be a sufficiently large constant depending on $\eps$ only.
	It is sufficient to prove that $\Pr[\dopt(G_1,\lambda K_k) < \frac{\eps}{2}] \to 1$ as $n\to \infty$.
	The rescaling factor here is $\lambda := p \binom{n}{2}/\binom{k}{2}$.
	The number of edges of $G(n,p)$ is concentrated around $p \binom{n}{2}$,
	so just as before a random function gives $G(n,p) \overcasts (1-\frac{1}{k})\lambda K_k \overcasts e^{-\eps/2} \lambda K_k$ with high probability (tending to 1 as $n\to\infty$).
	
	For the other direction, we use the fact that the number of $k$-cliques in $G(n,p)$ is concentrated around the mean $\binom{n}{k} p^{\binom{k}{2}}$
	and, more strongly, the number of $k$-cliques containing any given edge of $G(n,p)$ (conditioned on it being an edge) is concentrated around the mean $\binom{n-2}{k-2} p^{\binom{k}{2}-1}$.
	The concentration is good enough that with high probability, \emph{every} edge of $G(n,p)$ is contained in 
	$(1\pm\frac{\eps}{4}) \binom{n-2}{k-2} p^{\binom{k}{2}-1}$ $k$-cliques (see e.g.~\cite{Spencer99}).
	Thus if we take $\omega$ by mapping $\lambda K_k$ injectively to a random $k$-clique in $G(n,p)$, then w.h.p. for each edge $e$ of $G(n,p)$ we have
	\[ \EX_{g \sim \omega} |g^{-1}(e)| \geq (1-\frac{\eps}{4}) \binom{n-2}{k-2} p^{\binom{k}{2}-1} / \binom{n}{k} p^{\binom{k}{2}} = (1-\frac{\eps}{4}) \frac{k(k-1)}{n(n-1)} p^{-1} = (1-\frac{\eps}{4}) \lambda^{-1}.\]
	Thus $\lambda K_k \overcasts e^{-\eps/2} G_1$ and consequently $\dopt(G_1,\lambda K_k) \leq \frac{\eps}{2}$ w.h.p.
\end{proof}

To show Theorem~\ref{thm:denseIsNice}, we extend the above informal proof to any class of dense graphs.
This is possible because of \emph{Szemer\'edi's regularity lemma}~\cite{Szemeredi78:reglar},
which, very roughly speaking, guarantees that all such graphs are random-like.
This allows to provide similar bounds on the number of $k$-cliques containing any given edge,
a fact known as the \emph{extension lemma}, though we prove a variant that is somewhat tighter than usual.

\begin{remark}\label{rem:tournaments}
Note that the above proof sketch does not work for random tournaments (orientations of cliques): if we try to approximate them by the small graph $\frac{1}{2} \overset\leftrightarrow{K_k}$ (each arc taken with weight $\frac{1}{2}$), then every overcast from it to a tournament will always lose at least half of the total weight.
If instead we tried to take a small random tournament, no overcast to it from the big random tournament would work.
Indeed, Lemma~\ref{lem:niceIsMad} in Section~\ref{sec:non-examples} shows the class of tournaments is not pliable
(neither are ``random tournaments'', i.e., the proof can be adapted to show that
any class which contains a random tournament with constant probability cannot be
pliable)
and in fact the problem $\HOM(\A,-)$ for the class of tournaments $\A$ is hard to approximate, as we
show in Lemma~\ref{lemma:hard-fixed-sign-all} in Section~\ref{sec:hardness}.
This is why, even though variants of the regularity lemma exist for directed graphs and even more general structures, we limit our discussion to undirected graphs (the proofs do extend to $[0,1]$-weighted undirected graphs, however).
\end{remark}

In the rest of this section, we will prove Theorem~\ref{thm:denseIsNice}.
While we only prove this for unweighted graphs, it will be notationally convenient to treat them as $\{0,1\}$-weighted graphs, with $w_G(u,v) := [uv \in E(G)]$.
For sets $U,V \subseteq V(G)$, we denote by $w_G(U,V) := \sum_{u \in U} \sum_{v \in V} w_G(u,v)$ the number of edges between $U$ and $V$ (or their total weight).
The regularity lemma states that every graph can be partitioned into a bounded number of parts so that the bipartite graph between every two parts is random-like in the following strong sense:

\begin{definition}
	A bipartite graph $G=(V_1,V_2,E)$ of density $d := \frac{w_G(V_1,V_2)}{|V_1| |V_2|}$
	is \emph{$\eps$-homogeneous}
	if for all $W_1\subseteq V_1$, $W_2\subseteq V_2$,
	\[w_G(W_1,W_2) = d |W_1| |W_2| \, \pm \, \eps |V_1| |V_2|.\]
	
	For an $n$-vertex graph $G$ and an integer $k$, an \emph{$\eps$-regular $k$-partition} of $G$ is a partition
	$V_1, \dots, V_k$ of $V(G)$ such that
	$|V_i| \in \{\lfloor\frac{n}{k}\rfloor,\lceil\frac{n}{k}\rceil\}$ for $i \in [k]$ and
	the bipartite graph $(V_i, V_j, E(G) \cap V_i \times V_j)$ is $\eps$-homogeneous\footnote{The usual statement of the regularity lemma replaces $\eps$-homogeneity (with additive error) by a notion called \emph{$\eps$-regularity} (with relative error, but holding only for $|W_i|\geq \eps |V_i|$).
		The two are however easily shown to be equivalent, up to a change from $\eps$ to $\eps^{1/3}$.},
	for all $ij \in \binom{[k]}{2}$.
\end{definition}

We use the following strong version of Szemer\'{e}di's regularity lemma (see Theorem 2.2 in~\cite{RodlShacht07a}, Lemma 5.2. in~\cite{LovaszSzegedy07}, or Chapter~9 in~\cite{LovaszBook} for a detailed discussion).

\begin{theorem}[Regularity Lemma]\label{thm:regularity}
	For every $\eps_1 > 0$ and every non-decreasing $f \colon \NN \to \NN$, there is an integer $k$ such that
	for every sufficiently large graph $G$, one can add/remove $\eps_1 |V(G)|^2$ edges to obtain a graph
	which admits a $\frac{1}{f(k')}$-regular $k'$-partition for some $\frac{1}{\eps_1} \leq k'\leq k$.
\end{theorem}

Another way to view this is to define, for a partition $\Pp = (V_1,\dots,V_k)$ of a graph $G$,
the \emph{quotient graph $G_{/\Pp}$} as the weighted graph with vertex set $[k]$ and
weights $w_{G_{/\Pp}}(i,j) := w_G(V_i,V_j)$ for $(i,j) \in [k]^2$.
The quotient graph for an $\eps$-regular partition is then a graph of bounded
size that is \emph{close} to the original graph: the notion of closeness arising
from the definition of $\eps$-homogeneity is known as \emph{cut distance} (see
Chapter~8 in~\cite{LovaszBook}), but later we show the same holds for opt-distance:

\begin{theorem}\label{thm:regularityOptDistance}
	Let $G$ be a graph with density $c := \frac{|E(G)|}{n^2}$.
	For $0<\eps_0<1$, suppose $G$ has an $\eps$-regular $k$-partition $\Pp = (V_1,\dots,V_k)$
	with $\frac{1}{k} \leq \frac{c}{10} \frac{\eps_0}{1+\eps_0}$ and $\eps \leq \left(\frac{1}{k}\right)^{8k^2}$.
	Then $\dopt(G,G_{/\Pp}) \leq \eps_0$.
\end{theorem}

With this view it is easy to see that classes of dense graphs are pliable. Formally:
\begin{proof}[Proof of Theorem~\ref{thm:denseIsNice} assuming Theorem~\ref{thm:regularityOptDistance}]
	Let $\A$ be a class of graphs with $\geq cn^2$ edges.
	We want to show that for every $\eps_0 > 0$ there is a $k$ such that every $G \in \A$ has a weighted graph $H$ of size at most $k$ with $\dopt(G,H) \leq \eps_0$.
	
	For $\eps_0>0$, let $\eps_1 := \frac{c}{10} \cdot \frac{\eps_0/2}{1+\eps_0/2}$.
	Note that we can assume that all sufficiently large graphs $G \in \A$ have no loops: if $|V(G)| \geq \frac{1}{\eps_1}$, then the number of loops is at most $|V(G)| \leq \frac{1}{c|V(G)|} |E(G)| \leq \frac{\eps_1}{c} |E(G)|$.
	Hence by removing them we obtain a graph $G'$ such that $\dedit(G,G') \leq \frac{|E(G)|-|E(G')|}{\min(|E(G)|,|E(G')|)} \leq \frac{\eps_1/c}{1-\eps_1/c} \leq \eps_0/20$.
	By Lemma~\ref{lem:editToOpt}, $\dopt(G,G') \leq \eps_0/5$ (the direction $G \overcasts G'$ is trivial, while the other direction only requires $G'$ to have no loops).
	
	Let $f(k) := k^{8k^2}$.
	By the Regularity Lemma (Theorem~\ref{thm:regularity}), there is an integer $k \geq \frac{1}{\eps_1}$ such that for every graph $G$ of size $>k$, one can 
	add/remove $\eps_1 |V(G)|^2$ edges to obtain a graph $H$ which admits an $\frac{1}{f(k')}$-regular $k'$-partition $\Pp$ for some $\frac{1}{\eps_1} \leq k' \leq k$.
	If $G \in \A$, then $G$ has at least $cn^2$ edges, so $\dedit(G,H) \leq \frac{\eps_1}{c-\eps_1} \leq \eps_0/20$.
	Since we can assume that $G$ is loopless, by Lemma~\ref{lem:editToOpt}, $\dopt(G,H) \leq \eps_0/5$.
	By Theorem~\ref{thm:regularityOptDistance}, $\dopt(H,H_{/\Pp}) \leq \eps_0/2$.
	Hence $H_{/\Pp}$ is the graph of size at most $k$ we seek, at opt-distance at most $\eps_0$ from $G$.
\end{proof}

The strategy of the proof of Theorem~\ref{thm:regularityOptDistance} is very similar to Example~\ref{ex:random}.
One direction is trivial: an overcast from $G$ to $G_{/\Pp}$ is given simply by deterministically mapping all of $V_i$ to $i$, for $i \in [k]$.
For the other direction, we will take a subgraph $F$ of $G_{/\Pp}$ obtained by removing edges of small weight (keeping $F$ close to $G_{/\Pp}$) and removing weights, and then map $G_{/\Pp}$ to a random copy of $F$ in $G$.
We need to estimate the number of such copies (this is known as the \emph{counting lemma}) and, more generally, the number of such copies containing any given edge of $G$ (the \emph{extension lemma}).
Both are standard lemmas in the theory of dense graph limits, in particular our proof of the counting lemma mimics Lemma~10.22 in~\cite{LovaszBook}.
However, we will prove a version of the extension lemma with somewhat tighter bounds than usual (depending on all $\binom{k}{2}$ edge densities between parts of the regularity partition).

For a graph $F$ on vertex set $[k] \defeq \{1,\dots,k\}$, we will treat $F$ as a subset of $\binom{[k]}{2}$.
For a partition $\Pp = (V_1,\dots,V_k)$ of a graph $G$, a \emph{$\Pp$-map} is a function
$g \colon [k] \to V(G)$ such that $g(i) \in V_i$ for all $i \in [k]$.
We denote $\hom_g(F,G) := \prod_{ij \in F} w_G(g(i), g(j))$; for $\{0,1\}$-weighted graphs, this is equal to 1 if $g$ is a homomorphism from $F$ to $G$ and 0 otherwise.

Let us first observe two consequences of $\eps$-homogeneity.
First, the notion can be extended from subsets $W_1 \subseteq V_1$ to any
function $f \colon V_1 \to [0,1]$ simply by considering subsets where the function takes at least a given value (here $\|f\|_1 \defeq \sum_x f(x)$):
\begin{observation}\label{obs:homogenousFunction}
	Let $G=(V_1,V_2,E)$ be $\eps$-homogeneous of density $d$.
	Then for every $f \colon V_1 \to [0,1]$ and $g \colon V_2 \to [0,1]$,
	\[\sum\limits_{x_1 \in V_1} \sum\limits_{x_2 \in V_2} f(x_1) g(x_2) w_G(x_1,x_2) =
		d \|f\|_1 \|g\|_1  \  \pm \  \eps |V_1| |V_2|.\]
\end{observation}
\begin{proof}\allowdisplaybreaks[1]
	For $y \in [0,1]$, let $V_1^y := \{ x \in V_1 \colon f(x) \geq y \}$ and define $V_2^y$ analogously for $g$.
	Notice that $f(x) = \int_0^1 [y \leq f(x)] \,dy = \int_0^1 [x \in V_1^y] \,dy$.
	Hence
	\begin{gather*}
		\sum_{x_1 \in V_1} \sum_{x_2 \in V_2} f(x_1) g(x_2) w_G(x_1,x_2) = \\
		\int_0^1 \int_0^1 \sum_{x_1 \in V_1} \sum_{x_2 \in V_2} [x_1 \in V_1^{y_1}]  [x_2 \in V_2^{y_2}] w_G(x_1,x_2) \,dy_1 \,dy_2 = \\
		\int_0^1 \int_0^1 w_G(V_1^{y_1}, V_2^{y_2}) \,dy_1 \,dy_2 =\\
		\int_0^1 \int_0^1 \big( d |V_1^{y_1}| |V_2^{y_2}| \pm \eps |V_1| |V_2|\big) \,dy_1 \,dy_2 =\\
		d \cdot \Big(\int_0^1 |V_1^{y_1}| \,dy_1\Big) \Big(\int_0^1 |V_2^{y_2}| \,dy_2\Big) \pm \eps |V_1| |V_2|.
	\end{gather*}
	Since $\displaystyle\int_0^1 |V_1^y| dy = \int_0^1 \sum_{x\in V_1} [y \leq
  f(x)] dy = \sum_{x\in V_1} f(x) = \|f\|_1 $ and analogously for $g$, the claim follows.
\end{proof}

Second, while we cannot say much about any one fixed vertex, we can make similar approximations for most vertices:

\begin{observation}\label{obs:homogeneousPointed}
	Let $G=(V_1,V_2,E)$ be $\eps$-homogeneous with density $d$.
	For every $g \colon V_2 \to [0,1]$,
	there are at least $(1-2\sqrt{\eps}) |V_1|$ vertices $x_1$ in $V_1$ such that 
	$\sum_{x_2} g(x_2) w_G(x_1, x_2) = d \|g\|_1 \  \pm \  \sqrt{\eps} |V_2|$.
\end{observation}
\begin{proof}
	Let $W_1^{-}$ be the set of those $x_1$ in $V_1$ for which the sum is too small:
	\[\textstyle\sum_{x_2} g(x_2) w_G(x_1, x_2) <  d  \|g\|_1  - \sqrt{\eps} |V_2|.\]
	Let $f \colon V_1 \to [0,1]$ be the characteristic function of $W_1^{-}$.
	Then
	\[\textstyle\sum_{x_1} \sum_{x_2} f(x_1) g(x_2) w_G(x_1, x_2) < \|f\|_1 \cdot \left(d \|g\|_1  - \sqrt{\eps} |V_2|\right).\]
	By Observation~\ref{obs:homogenousFunction}, this implies $\|f\|_1 \cdot \sqrt{\eps} |V_2| < \eps |V_1| |V_2|$.
	Hence $|W_1^{-}| = \|f\|_1 < \sqrt{\eps} |V_1|$.
	We can define and bound $W_1^{+}$ analogously.
	Then $V_1 \setminus (W_1^{-} \cup W_1^{+})$ is a set of size at least $(1-2\sqrt{\eps}) |V_1|$ as claimed.
\end{proof}

The counting lemma says that the number of $\Pp$-maps that are homomorphisms from $F$ to $G$ is close to what one would expect in a purely random graph with the same densities.
Note that the number of all $\Pp$-maps $g \colon [k] \to V(G)$ is exactly $\prod_{i \in [k]} |V_i|$.

\begin{lemma}[Counting Lemma]\label{lem:counting}
	Let $\Pp=(V_1,\dots,V_k)$ be an $\eps$-regular $k$-partition of an $n$-vertex graph $G$.
	Let $d_{ij} := \frac{w_G(V_i,V_j)}{|V_i| |V_j|}$.
	For each $F\subseteq \binom{[k]}{2}$,
	\[  \sum_g \hom_g(F,G) =
		\Big(\prod_{i \in [k]} |V_i|\Big)
		\Big(\prod_{ij \in F} d_{ij} \  \pm \  \eps |F|\Big),\]
	where the sum is over all $\Pp$-maps $g \colon [k] \to V(G)$.
\end{lemma}
\begin{proof}
	Let us write $\sum_{(x_i)_{i \in [k]}}$ as a shorthand for $\sum_{x_1 \in V_1} \dots \sum_{x_k \in V_k}$.
	We wish to approximate 
	\[\sum_{(x_i)_{i \in [k]}} \prod_{ij \in F} w_G(x_i,x_j).\]
	We do so by replacing each factor $w_G(x_i,x_j)$ by its average value $d_{ij}$, one by one.
	That is, we prove for all subsets $F' \subseteq F$ by induction that 
	\begin{equation}\label{eq:counting}\tag{*}
	\sum_{(x_i)_{i \in [k]}} \prod_{ij \in F} w_G(x_i,x_j) \quad = \quad
		\sum_{(x_i)_{i \in [k]}} \prod_{ij \in F-F'} w_G(x_i,x_j) \prod_{ij \in F'} d_{ij} \quad \pm \quad |F'| \cdot \eps  \prod_{i \in [k]} |V_i|.
	\end{equation}
	
	Clearly this is true initially for $F' = \emptyset$ and eventually by reaching $F'=F$ we will have proved that 
	\[\begin{gathered}
		\sum_{(x_i)_{i \in [k]}} \prod_{ij \in F} w_G(x_i,x_j) \quad = 	\quad 
		\sum_{(x_i)_{i \in [k]}} \prod_{ij \in F} d_{ij} \quad \pm \quad |F| \cdot \eps  \prod_{i \in [k]} |V_i|;
\end{gathered}\]
	which proves the claim, as $\sum_{(x_i)_{i \in [k]}} 1 = \prod_{i \in [k]} |V_i|$. 
	
	To prove the induction step, suppose \eqref{eq:counting} holds for some $F' \subset F$ and let $ab \in F-F'$.
	Let $w'_G(x_i,x_j)$ denote $w_G(x_i,x_j)$ if $ij \not\in F'$ and $d_{ij}$ otherwise.
	Then the 
  left-hand-side
  in \eqref{eq:counting} is
	\[\begin{gathered}
		\sum_{(x_i)_{i \in [k]}} \prod_{ij \in F} w'_G(x_i,x_j) = \\
		\sum_{(x_i)_{i \in [k]-a-b}} h \sum_{x_a} \sum_{x_b} f(x_a) g(x_b) w'_G(x_a,x_b),
	\end{gathered}\]
	where for any fixed choice of $(x_i)_{i \in [k]-a-b}$, we let
	\[\begin{gathered}
			h := \prod_{\substack{ij \in F-F'\\i \neq a, j \neq b}} w'_G(x_i,x_j), \quad\quad
			f(x_a) := \prod_{\substack{ij \in F-F'\\i = a, j \neq b}} w'_G(x_a,x_j), \quad\quad
			g(x_b) := \prod_{\substack{ij \in F-F'\\i \neq a, j = b}} w'_G(x_i,x_b).	
	\end{gathered}\]
	Since $h,f,g \leq 1$, the claim then follows from Observation~\ref{obs:homogenousFunction}:
	replacing $w_G(x_a,x_b)$ with $d_{ab}$ adds an error of at most
	$\sum_{(x_i)_{i \in [k]-a-b}} h  \cdot \eps |V_a| |V_b| \leq \eps |V_a| |V_b| \cdot \sum_{(x_i)_{i \in [k]-a-b}} 1 = \eps \prod_{i\in [k]} |V_i|$.
\end{proof}

\begin{lemma}[Extension Lemma]\label{lem:extension}
	Let $\Pp=(V_1,\dots,V_k)$ be an $\eps$-regular $k$-partition of an $n$-vertex graph $G$.
	Let $d_{ij} := \frac{w_G(V_i,V_j)}{|V_i| |V_j|}$.
	For each $F\subseteq \binom{[k]}{2}$ and each $ab \in F$,
	all but $2 k \sqrt{\eps} |V_a| |V_b|$ edges $x_a x_b \in V_a \times V_b$ satisfy
	\[
		\sum_{g} \hom_g(F,G) = 
		\Big(\prod_{i \in [k]-a-b} |V_i|\Big) \cdot
		\Big( w_G(x_a,x_b) \cdot \prod_{ij \in F-ab} d_{ij}
			\quad \pm \quad
		\sqrt{\eps} |F|\Big)
	\]
	where the sum is over all $\Pp$-maps $g \colon [k] \to V(G)$ such that $g(a)=x_a$ and $g(b)=x_b$.
\end{lemma}
\begin{proof}
	The argument is the same as in the counting lemma, except that edges incident to $a,b$ have to be handled differently.	
	First note that for every $c \in [k]-a-b$ and every fixed $x_b \in V_b$, by
  Observation~\ref{obs:homogeneousPointed} (with $g(x_c):=w_G(x_c,x_b)$), the
  following holds for all but at most $2\sqrt{\eps}|V_a|$ vertices $x_a$ in $V_a$:
	\begin{equation}\label{eq:triangle}\tag{**}
	\sum_{x_c} w_G(x_a,x_c) w_G(x_c, x_b) = d_{ac} \Big(\sum_{x_c} w_G(x_c,x_b)\Big) \pm \sqrt{\eps}|V_c|.
	\end{equation}
	For each $c \in [k]-a-b$  and each $x_b \in V_b$, we will ignore those edges going to $x_a \in V_a$ that fail \eqref{eq:triangle}.
	
	Similarly for each $c \in [k]-a-b$, by Observation~\ref{obs:homogeneousPointed} (with $g(x_c):=1$)
	the following holds for all but at most $2\sqrt{\eps} |V_b|$ vertices $x_b \in V_b$:
	\begin{equation}\label{eq:homogeneousDegree}\tag{***}
	\sum_{x_c} w_G(x_b, x_c) = d_{bc} |V_c| \pm \sqrt{\eps}|V_c|.
	\end{equation}
	We ignore all edges $x_ax_b \in E(V_a,V_b)$ incident to $x_b$ for which \eqref{eq:homogeneousDegree} fails.
	Thus for all but $\leq 2 \cdot k \cdot |V_b| \cdot 2\sqrt{\eps}|V_a|$ edges $x_ax_b \in E(V_a,V_b)$,  \eqref{eq:triangle} and \eqref{eq:homogeneousDegree} hold for all $c \in [k]-a-b$.	
	
	Fix any such $x_a \in V_a, x_b \in V_b$. We wish to approximate
	\[\begin{gathered}
		 \sum_{(x_i)_{i \in [k]-a-b}} \prod_{ij \in F} w_G(x_i,x_j) = \\	
		 w_G(x_a,x_b) \cdot \sum_{(x_i)_{i \in [k]-a-b}} \prod_{ij \in F-ab} w_G(x_i,x_j)=\dots
	 \end{gathered}\]
	Just as in the proof of the counting lemma, we replace factors $w_G(x_i,x_j)$ by $d_{ij}$ one by one.
	We first do this for pairs in $F_{0} := \{ij \in F \mid i,j \neq a,b\}$, since the argument works without change, incurring an error of $\pm \eps |F_{0}| \prod_{i\in[k]-a-b} |V_i|$ (we denote this by $\simeq$ for simplicity and sum up the errors at the end of the proof).
	Since $d_{ij}$ does not depend on the choice of $x_i \in V_i, x_j \in V_j$, we can rearrange:
	\[ \dots\simeq w_G(x_a,x_b) \Big(\prod_{ij \in F_{0}} d_{ij}\Big) \cdot \sum_{(x_i)_{i \in [k]-a-b}} \prod_{i \in [k]-a-b} w_G(x_a,x_i) w_G(x_i,x_b)= \dots\]
	
	Then, for each $c \in [k]-a-b$ we can replace $ac$ by isolating the factors that depend on $x_c$ and applying \eqref{eq:triangle} (as before $w'_G(x_a,x_i)$ denotes either $w_G(x_a,x_i)$ or $d_{ai}$ depending on whether we have already replaced $ai$):
	\begin{multline*}
		 \dots = w_G(x_a,x_b) \Big(\prod_{ij \in F_{0}} d_{ij}\Big) \cdot
		\sum_{x_c} w_G(x_a,x_c) w_G(x_c, x_b) \ \cdot \\
		\sum_{(x_i)_{i \in [k]-a-b-c}} \prod_{i \in [k]-a-b-c} w'_G(x_a,x_i) w_G(x_i,x_b) \simeq \dots
	\end{multline*}
	Having thus replaced all edges $ac$ for $c \in [k]-a-b$,
	the only remaining edges are of the form $ib$ for $i \in [k]-a-b$,
	so by denoting $F_1 := \{ ij \in F \mid i,j \neq b\}$ the expression becomes:
	\[\begin{gathered}
		\dots \simeq w_G(x_a,x_b) \Big(\prod_{ij \in F_{1}} d_{ij}\Big) \cdot 
		\sum_{(x_i)_{i \in [k]-a-b}} \prod_{i \in [k]-a-b} w_G(x_i,x_b) = \\
		w_G(x_a,x_b) \Big(\prod_{ij \in F_{1}} d_{ij}\Big) \cdot 
		\prod_{i \in [k]-a-b} \Big(\sum_{x_i \in V_i} w_G(x_i,x_b)\Big) \simeq \\
		w_G(x_a,x_b) \Big(\prod_{ij \in F} d_{ij}\Big) \cdot 
		\prod_{i \in [k]-a-b} |V_i|,
	\end{gathered}\]
	where the last approximation follows from \eqref{eq:homogeneousDegree}.
	For each of the $|F|$ approximations used, the incurred additive error on the whole expression was at most $\pm \sqrt{\eps} \prod_{i \in [k]-a-b} |V_i|$.
\end{proof}

We are now ready to prove Theorem~\ref{thm:regularityOptDistance}.
The proof strategy was outlined above: map $G_{/\Pp}$ to a random copy of $F$ in $G$, where $F$ marks heavy-enough edges of  $G_{/\Pp}$.
\begin{proof}[Proof of Theorem~\ref{thm:regularityOptDistance}]
	Let $G$ be a graph with density $c := \frac{E(G)}{n^2}$.
	Let $\eps_0<1$, $\frac{1}{k} \leq \frac{c}{10} \frac{\eps_0}{1+\eps_0}$ and $\eps \leq \left(\frac{1}{k}\right)^{8k^2}$,
	and suppose $G$ has an $\eps$-regular $k$-partition $\Pp = (V_1,\dots,V_k)$.
	We claim that $\dopt(G,G_{/\Pp}) \leq \eps_0$.
	As mentioned above, $G \overcasts G_{/\Pp}$ holds trivially, so it remains to show an overcast from $G_{/\Pp}$ to $e^{-\eps_0}G$.
	
	Let $d_{ij} := \frac{w_G(V_i,V_j)}{|V_i||V_j|}$ for $ij \in [k]^2$.
	Let $F \subseteq \binom{[k]}{2}$ be the set of edges $ij$ such that $i \neq j$ and $d_{ij} \geq \frac{1}{k}$.
	Note that $\prod_{ij \in F} d_{ij} \geq (\frac{1}{k})^{|F|} \geq (\frac{1}{k})^{k^2} \geq \eps^{1/8}$.	
	Let $G'$ be the subgraph of $G$ obtained by removing:
	\begin{itemize}
		\item $E(G[V_i])$, for $i \in [k]$ (the total weight removed in this step is $\leq k \left(\frac{n}{k}\right)^2$)
		\item $E_G(V_i,V_j)$, for $ij \not\in F$ (their total weight is $\leq k^2 \cdot \frac{1}{k} \cdot \left(\frac{n}{k}\right)^2$)
		\item edges of weight $<\eps^{1/8}$ (if $G$ is $[0,1]$-weighted; their total weight is $\leq \eps^{1/8} n^2$)
		\item edges $x_a x_b \in V_a \times V_b$ for which the Extension Lemma~\ref{lem:extension} does not hold, for each $ab \in F$ (their total weight is $\leq |F| \cdot 2 k\sqrt{\eps} \left(\frac{n}{k}\right)^2$).
	\end{itemize}
	The total weight of removed edges is
	\[\|G\|_1-\|G'\|_1 \leq n^2 (\frac{1}{k} + \frac{1}{k} + \eps^{1/8} + 2k \sqrt{\eps})
		\leq n^2 \cdot \frac{5}{k}.
	\]
	By our assumption on $k$, $\frac{5}{ck} \leq \frac{1}{2} \frac{\eps_0}{1+\eps_0} < 1$.
	Since $\|G\|_1 \geq cn^2$, $\dedit(G,G') \leq \frac{\|G\|_1-\|G'\|_1}{\min(\|G\|_1,\|G'\|_1)} \leq \frac{n^2 \cdot \frac{5}{k}}{n^2(c-\frac{5}{k})} = \frac{5/ck}{1-5/ck} \leq \eps_0/2.$ 
	Therefore, by Lemma~\ref{lem:editToOpt}, $G' \overcasts e^{-\eps_0/2}G$ (this
  direction requires only $G'$ to be loopless, which is true because we removed $E(G[	V_i])$ for all $i$).
	Thus it remains to show that $G_{/\Pp} \overcasts e^{-\eps_0/2}G'$.
	
	We define an overcast $\omega$ from $G_{/\Pp}$ to $(1-\frac{\eps_0}{1+\eps_0})G'$ as follows:
	every $\Pp$-map $g \colon [k] \to V(G)$ is taken with probability proportional to $\hom_g(F,G)$;
	that is, $\omega(g) := \hom_g(F,G) / N$ where by the Counting Lemma~\ref{lem:counting}, the normalisation factor is (using $\prod_{ij \in F} d_{ij} \geq \eps^{1/8}$):
	\[N := \sum_g \hom_g(F,G) =
		\Big(\prod_{i \in [k]} |V_i|\Big)
		\Big(\prod_{ij \in F} d_{ij} \  \pm \  \eps |F|\Big)
		\leq 
		\Big(\prod_{i \in [k]} |V_i|\Big)
		\Big(\prod_{ij \in F} d_{ij}\Big)
		(1 + \eps^{7/8} k^2).
	\]
	To verify that $\omega$ is indeed an overcast, we need to check that for each edge $uv$ of $G'$
	\[ \EX_{g \sim \omega} w_{G_{/\Pp}}(g^{-1}(uv)) \geq (1-\eps') w_G(u,v) .\]
	Let $a,b \in [k]$ be such that $u \in V_a$ and $v \in V_b$.
	By the Extension Lemma~\ref{lem:extension}, since we removed from $G'$ edges
  that do not satisfy it and edges with $w_G(u,v) < \eps^{1/8}$, we have:
	\[\begin{gathered}
		N \cdot \EX_{g \sim \omega} w_{G_{/\Pp}}(g^{-1}(uv)) =
		\sum_{g\colon g(a)=u,g(b)=v} \hom_g(F,G) w_{G_{/\Pp}}(a,b) = \\
		w_G(V_a,V_b) \sum_{g\colon g(a)=u,g(b)=v} \hom_g(F,G) = \\
		d_{ab} |V_a| |V_b| \Big(\prod_{i \in [k]-a-b} |V_i|\Big)
		\Big( w_G(u,v) \prod_{ij \in F-ab} d_{ij}
			\quad \pm \quad
		\sqrt{\eps} |F|\Big) \geq \\
	 	\Big(\prod_{i\in[k]} |V_i| \Big) \Big(w_G(u,v) \prod_{ij \in F} d_{ij} \  - \  \sqrt{\eps} k^2\Big)
	 	\geq \\
	 	\Big(\prod_{i\in[k]} |V_i| \Big) w_G(u,v) \Big(\prod_{ij \in F} d_{ij}\Big)
	 		(1 - \eps^{1/4} k^2).
	\end{gathered}\]
	(The last inequality follows from $w_G(u,v) \cdot \prod_{ij \in F} d_{ij} \geq \eps^{1/8} \cdot \eps^{1/8} = \eps^{1/4}$).
	Dividing by the upper bound on $N$, we conclude:
	\[\EX_{g \sim \omega} w_{G_{/\Pp}}(g^{-1}(uv))  \geq w_G(u,v) \frac{1 - \eps^{1/4} k^2}{1 + \eps^{7/8} k^2} \]
	The ratio here can be bounded quite brutally: 
	\[ \geq  \frac{1 - \eps^{1/4} k^2}{1 + \eps^{1/4} k^2} \geq 1-2\eps^{1/4}k^2 \geq 1-\frac{2}{k} \geq 1 - \frac{1}{2}\frac{\eps_0}{1+\eps_0} \geq \frac{1}{1+\eps_0/2} \geq e^{-\eps_0/2}.\]
	This concludes the proof that $G_{/\Pp} \overcasts e^{-\eps_0/2}G'$ and hence $\dopt(G_{/\Pp},G) \leq \eps_0$.
\end{proof}

\section{Hardness of approximation}
\label{sec:hardness}

We show that $\HOM(\A,-)$, where $\A$ is the class of all tournaments (orientations of cliques),
has no PTAS. 
This holds under the \emph{Gap Exponential Time Hypothesis} (Gap-ETH)~\cite{Manurangsi17:icalp,Dinur2016MildlyER}
which states that no $2^{o(n)}$-time algorithm can distinguish between a satisfiable 3SAT formula and one which is not even $(1-\varepsilon)$-satisfiable for some constant $\varepsilon>0$. 

In fact we only require the following weaker conjecture:
\begin{conjecture}\label{conj:pih}
	There exists an $\eps>0$ such that
	given a $\{0,1\}$-valued Max-2-CSP instance with $k$ variables and alphabet size $n$
	no $f(k) \cdot n^{\Oh(1)}$ time algorithm can distinguish between the following two cases:
	\begin{itemize}
		\item there is an assignment satisfying every constraint;
		\item no assignment satisfies more than $(1-\eps)$ constraints.
	\end{itemize}
\end{conjecture}

Gap-ETH implies Conjecture~\ref{conj:pih}: this follows from a proof by Chalermsook et al.~\cite{CCK20}, in fact with a much larger approximation gap, which was further improved by Dinur and Manurangsi~\cite{DinurPasin18}.
Direct proofs for the above simpler version can be found in~\cite{LokshtanovRSZ20} and~\cite[Appendix A]{Bhattacharyya18}.
Lokshtanov et al.~\cite{LokshtanovRSZ20} moreover propose the Parameterised Inapproximability Hypothesis, stating that the above promise problem is W[1]-hard.

The problem can be rephrased as a minor variation of Densest-$k$-Subgraph (sometimes known as Maximum Colored Subgraph Isomorphism):

\begin{observation}
	\label{obs:densest-hard}	
	Conjecture~\ref{conj:pih} is equivalent to the following.
	There is an $\eps>0$ such that no $f(k)\cdot n^{\Oh(1)}$ time algorithm can,
	given $k$, a graph $G$ on $n$ vertices, and a proper $k$-colouring $c$ of it,
	distinguish between the following two cases:
	\begin{itemize}
		\item $G$ contains a $k$-clique $v_1,\dots,v_k$ (without loss of generality $c(v_i)=i$);
		\item every $k$-tuple $v_1,\dots,v_k$ with $c(v_i)=i$ induces a subgraph on $<(1-\eps)\binom{k}{2}$ edges in $G$.
	\end{itemize}
\end{observation}

(Indeed, the $k$ variables in the Max-2-CSP correspond to $v_1,\dots,v_k$, the set of vertices coloured $i$ is the alphabet for variable $v_i$, and the edges between two colour sets define a constraint).
As a side note, we remark that an inspection of the proof of~\cite[Theorem~5.12]{CCK20} gives that Gap-ETH implies that the above is hard even if the soundness case is strengthened as follows, for \emph{any} constant $\delta>0$:
\begin{itemize}
	\item every $k$-tuple $v_1,\dots,v_k$ (regardless of colours) induces a subgraph on $<\delta\binom{k}{2}$ edges in $G$.
\end{itemize}

The problem in Observation~\ref{obs:densest-hard} is almost a maximum graph homomorphism problem on cliques,
except that, crucially, the mapping $i \mapsto v_i$ is forced to be injective.
To show that $\HOM(\A,-)$ is hard for the class $\A$ of tournaments,
intuitively, we use the fact that a map from a random tournament to itself must map most arcs to themselves (and is hence approximately injective).
This is formalised as follows.

\begin{lemma}
\label{lemma:oriented-clique}
For every $\delta > 0$, there exists constants $0 < \lambda < \delta$ and $N\geq 1$ such that the following holds.
For every $k\geq N$, there is an orientation $\structA$ of the clique of size $k$ such that every mapping $g:A\to A$ of $\structA$ to itself with $\val(g)\geq(1-\lambda)\binom{k}{2}$ must map at least $(1-\delta)\binom{k}{2}$ arcs to themselves.
\end{lemma}
\begin{proof}
For $\delta>0$, denote $m:=\binom{k}{2}$ and choose $N\geq 1$ such that $k\log_2 k\leq \frac{\delta}{3}m$, for all $k\geq N$. 
Let $\lambda>0$ be constant to be chosen later. 
Let $\structA$ be a random orientation of the clique of size $k$ with $k\geq N$ (each edge is independently oriented in either direction with probability $\frac{1}{2}$). 
We will show that with positive probability $\structA$ admits no map $g: A \to A$ to itself with the properties that $\val(g)\geq(1-\lambda)m$ but less than
  $(1-\delta)\binom{k}{2}$ arc of $\structA$ are mapped identically by $g$.

If a map as above existed, it would imply the existence of a set $F$ of arcs of $\structA$ with $|F|\leq \lambda m$ 
and a mapping $g:A\to A$ such that $g$ maps all the arcs of $\structA-F$ correctly,
and such that $g$ maps less than $(1-\delta)\binom{k}{2}$ vertex pairs identically.
Let us bound the probability that there exist such $F,g$.
The number of possible $F$ is $\leq \sum_{i=0}^{\lambda m}\binom{m}{i} \leq 2^{H(\lambda)m}$.
The number of possible $g$ is $\leq k^{k}$.
For fixed $F,g$, if $g$ maps less than $(1-\delta)\binom{k}{2}$ vertex pairs identically,
then the number of remaining arcs of $\structA-F$ is at least $(1-\lambda)m-(1-\delta)\binom{k}{2} = (\delta-\lambda)m$;
the probability that all these arcs are mapped correctly by $g$ 
is at most $\frac{1}{2}^{(\delta-\lambda)m/2}$
(each of these arcs is mapped correctly with probability $\frac{1}{2}$; since the function $g$ forms cycles on the set of arcs, the events for individual arcs are not independent, but if we ignore one arc from each cycle they are; since cycles have length at least 2, we ignore at most $\frac{1}{2}$ of these arcs).
Hence in total the probability that some such $F,g$ exist is at most 
	\[2^{H(\lambda)m}  \cdot k^k \cdot  2^{-(\delta-\lambda)m/2}=\\
	2^{-(\frac{\delta}{2}-\frac{\lambda}{2}-H(\lambda))m} \cdot 2^{k\log_2 k} \leq 2^{-(\frac{\delta}{2}-\frac{\lambda}{2}-H(\lambda)-\frac{\delta}{3})m}.\]
This is less than $1$ by taking $\lambda$ small enough so that $\frac{\delta}{6}-\frac{\lambda}{2}-H(\lambda) > 0$.	
\end{proof}

This allows us to make the reduction.

\begin{lemma}
	\label{lemma:hard-fixed-sign-all}
	For every $\delta > 0$, there exists constants $0 < \lambda < \delta$ and $N\geq 1$ such that the following holds.	
	Given $k\geq N$, a graph $G$ on $n$ vertices, and a proper $k$-colouring $c$ of $G$,
	we can compute in $f(k)\cdot n^{\Oh(1)}$ time an orientation $\structA$ of the clique of size $k$ and a directed graph $\structB$ such that  
	\begin{itemize}
		\item if $G$ contains a clique of size $k$, then $\opt{\structA,\structB}=\binom{k}{2}$, 
		\item if every $v_1,\dots,v_k$ in $G$ with $c(v_i)=i$ induce $<(1-2\delta)\binom{k}{2}$ edges,
		then $\opt{\structA,\structB}<(1-\lambda)\binom{k}{2}$.
	\end{itemize}
\end{lemma}
\begin{proof}
For $\delta>0$, let $\lambda$ and $N$ be as in Lemma~\ref{lemma:oriented-clique}. 
Given $k\geq N$, $G$ and a proper $k$-colouring $c:V(G)\to\{1,\dots,k\}$ of $G$, we start by computing an orientation $\structA$ of the clique on the set of colours $\{1,\dots,k\}$ as in Lemma~\ref{lemma:oriented-clique} (in time depending on $k$ only). 
The directed graph $\structB$ has vertex set $V(G)$ and $(u,v)$ is an arc in $\structB$ iff $\{u,v\}\in E(G)$ and $({c(u)}, {c(v)})$ is an arc in $\structA$. 
Suppose that $G$ contains a clique $\{v_1,\dots,v_k\}$ of size $k$.
Without loss of generality $c(v_i)=i$.
Then $\opt{\structA,\structB}=\binom{k}{2}$ via the mapping $h(i)\defeq v_i$.

Assume now that $\opt{\structA,\structB} \geq (1-\lambda)\binom{k}{2}$,
so there is a mapping $g:A\to B$ with $\val(g)\geq (1-\lambda)\binom{k}{2}$.
Note that $c:B\to A$ is a homomorphism from $\structB$ to $\structA$. 
It follows that the mapping $c\circ g:A\to A$ from $\structA$ to itself has $\val(c\circ g)\geq (1-\lambda)\binom{k}{2}$. 
By Lemma~\ref{lemma:oriented-clique}, we have 
that $c\circ g$ maps at least $(1-\delta)\binom{k}{2}$ arcs to themselves. 
Let $F$ be the set of arcs that are not mapped to themselves by $c \circ g$ (so $|F| \leq \delta\binom{k}{2}$).
Let $F'$ be the set of arcs of $\structA$ that are mapped incorrectly by $g$ (so $|F'|\leq \lambda\binom{k}{2}$).  
The remaining arcs, $\structA-F-F'$, satisfy the following:
their number is at least $(1-\delta-\lambda)\binom{k}{2} \geq (1-2\delta)\binom{k}{2}$;
they are mapped by $g$ to some arcs in $\structB$ and hence to some edge in~$G$;
and if $i \in\{1,\dots,k\}$ is an endpoint of any of these arcs, then $c(g(i))=i$.
We can hence take $v_i \defeq g(i)$ if $i$ is not isolated in $\structA-F-F'$
and take an arbitrary $v_i$ with $c(v_i)=i$ otherwise;
the resulting $k$-tuple induces at least $(1-2\delta)\binom{k}{2}$ edges in $G$ and satisfies $c(v_i)=i$.
\end{proof}

From Observation~\ref{obs:densest-hard} (with some $\eps>0$) and Lemma~\ref{lemma:hard-fixed-sign-all} (with $\delta=\frac{\eps}{2}$) we conclude:
\begin{corollary}\label{cor:tournamentsAreHard}
	Assuming Conjecture~\ref{conj:pih}, there is a constant $\lambda>0$ such that $\HOM(\A,-)$
	for the class of tournaments $\A$ 
	has no $(1-\lambda)$-approximation running in time $f(|\structA|)(|\structA|+|\structB|)^{\Oh(1)}$.
\end{corollary}

In particular, assuming Gap-ETH, there is no PTAS (and actually no \emph{FPT approximation scheme}) for $\HOM(\A,-)$.

\section{Open questions}
\label{sec:open}

\paragraph{Dichotomy}
Our results, in particular Lemma~\ref{lem:fragileIffNice}, lead us to believe
that perhaps the next question has a positive answer. 

\begin{question}\label{q:dichotomy}
  Does $\rCSP(\Gg)$ admit a PTAS for every $r$ if and only if $\Gg$ is
  fractionally-$\tw$-fragile?
\end{question}

\noindent
A concrete open question concerns cographs, i.e., graphs of clique-width two:
Are cographs pliable? More generally, are graphs of bounded clique-width
pliable? Note that while certain techniques and results transfer from
structurally sparse graphs to dense graphs (clique-width being an example), this
fails for pliability as pliability is not closed under complementation since all
classes of dense graphs are pliable.

Some example cases where it would be important to show hardness of approximation
(or at least integrality gaps for constant levels of the Sherali-Adams
hierarchy) in order to shed light on Question~\ref{q:dichotomy} are classes of unbounded average degree or classes of 3-regular graphs with unbounded girth. In fact, we do not know of any examples of non-pliable classes of structures $\A$ for which $\HOM(\A,-)$ admits a PTAS and thus it is non inconceivable that pliability captures all tractable cases.

\paragraph{Constant-factor approximation}
Instead of PTASes one can of course ask about the existence of \emph{some} constant-factor approximation.
For fixed signatures, $\HOM$ always admits a simple constant-factor approximation: essentially map everything randomly to the densest $r$-tuple, where $r$ is the maximum arity.
For the general $\rCSP(\Gg)$ problem the situation is more interesting:
in general (when $\Gg$ is the class of all graphs) a constant-factor approximation is impossible;
on the other hand for any monotone class of bounded average degree, there is again a simple solution:
because such classes have bounded degeneracy, the edge set can be partitioned into a constant number of trees, where the problem can be solved exactly.
The results of \cite{DinurPasin18} imply that if the average degree is too high, the problem is again hard.
Can a dichotomy be shown?

\paragraph{Weak hyperfiniteness}
As shown in Theorem~\ref{thm:hyperfinite},
monotone hyperfinite classes are fractionally-$\tw$-fragile and have bounded degree.
The vertex version of hyperfiniteness is called ``weakly hyperfinite'' in~\cite{NesetrilM12chapter}
or ``fragmentable'' in~\cite{EdwardsM94,EdwardsFarr12}.
Is is strictly weaker: stars satisfy it, despite having unbounded degree.
Ne\v{s}et\v{r}il and Ossona de Mendez~\cite{NesetrilM12chapter} proved that for a monotone class of graphs $\Gg$ of bounded average degree, $\Gg$ is weakly hyperfinite if and only if for every $d\in \NN$, $\{G \in \Gg \colon \max\deg(G)\leq d\}$ is hyperfinite.
This suggests a possible extension to graphs of unbounded degree:
are monotone weakly hyperfinite classes fractionally-$\tw$-fragile?
This would imply a conjecture of Dvo\v{r}\'{a}k~\cite{Dvorak16}, that all graph classes with strongly sublinear separators are fractionally-$\tw$-fragile.
However, it is not even known whether all monotone weakly hyperfinite classes have bounded average degree.

\paragraph{Efficient PTAS}
As mentioned in the introduction, both dense graphs and hyperfinite graphs can be approximated by constant-size descriptions, and in fact by constant-size random samples.
Since $\size$-pliability also approximates with constant-size descriptions,
this suggests there may be a general way to sample from such structures
to give constant-time approximations (for an appropriate input model).
In particular, can property-testing results for hyperfinite graphs be extended to fractionally-$\tw$-fragile graphs?
Perhaps this could be a way to obtain EPTASes\footnote{Efficient PTAS (EPTAS) is a PTAS whose running time is $\Oh(n^c)$ for a constant $c$ independent of $\varepsilon$.} for $\rCSPs$ with fixed alphabets.
Our methods seems unlikely to give an EPTAS directly.
The analogous question in the exact setting is as follows; we believe it to be open.

\begin{question}\label{q:noEPTAS}
	Is 3-colouring fixed-parameter tractable when parameterised by the treewidth of the core (the smallest homomorphically equivalent subgraph)?
	That is, given a graph $G$ which is promised to have a core of treewidth at most $k$, can we decide its 3-colourability in time $f(k) |G|^{\Oh(1)}$ for some function $f$?
\end{question}

In this question treewidth could also be replaced by size, in which case an algorithm with running time $\Oh(n)^k$ is trivial (test every $k$-subgraph for 3-colourability).
A similar algorithm for treewidth is due to Dalmau et al.~\cite{DalmauKV02}, see also~\cite[Theorem 3.1]{Grohe07:jacm}.

\section*{Acknowledgements}
The authors would like to thank the referees of the extended
abstract~\cite{rwz21:soda} and (in particular) of this full version for many
detailed and thoughtful comments.

{\small
\bibliographystyle{plainurl}
\bibliography{rwz23}
}

\appendix

\section{Farkas' lemma and proofs of Proposition~\ref{prop:overcast} and Lemma~\ref{lem:thin}} \label{sec:Farkas}
\label{sec:farkas}

Farkas' lemma is the fundamental duality for systems of linear inequalities.
\begin{lemma}[\protect{Farkas' lemma~\cite[Corollary~7.1d]{Schrijver86:ILP}}]\label{lem:farkasOriginal}
	Let $A$ be an $m \times n$ rational matrix and $\bar{b} \in \QQ^m$. Then, exactly one of the two holds:
	\begin{itemize}
		\item $A\bar{x} = \bar{b}$ for some $\bar{x} \in \QQ^n$ with $\bar{x}\geq 0$, or
		\item $A^T\bar{y} \geq 0$ and $\bar{b}^T\bar{y} < 0$ for some $\bar{y} \in \QQ^m$.
	\end{itemize}
\end{lemma}

\noindent
For the duality between the existence of overcasts and the overcast relation $\overcasts$, we use Farkas' lemma in the following form:

\begin{lemma}[Farkas' lemma, variant 1]\label{lem:farkas1}
	Let $A$ be an $m \times n$ rational matrix and $\bar{b} \in \QQ^m$.
	Exactly one of the following holds:
	\begin{itemize}
		\item there are $x_i\in\QQ_{\geq 0}$ ($i=1,\dots,n$) such that $\sum_i x_i = 1$ and $\sum_i A_{i,j} x_i \geq b_j$ for $j=1,\dots,m$;
		\item there are $y_j\in\QQ_{\geq 0}$ ($j=1,\dots,m$) such that $\sum_j A_{i,j} y_j < \sum_j b_j y_j$ for $i=1,\dots,n$.	
	\end{itemize}
\end{lemma}
\begin{proof}
	\newcommand{\slack}{s}
	The first condition is equivalent to the existence of a solution in variables $x_i \in \QQ_{\geq 0}$ ($i=1,\dots,n)$ and $\slack_j \in \QQ_{\geq 0}$ ($j=1,\dots,m$) of the following system:
	\begin{align*}
		\sum_i x_i &= 1\\
		\sum_i A_{i,j} x_i - \slack_j &= b_j \quad \quad \text{(for $j=1,\dots,m$)}.
	\end{align*}
	By Lemma~\ref{lem:farkasOriginal}, this system has a solution if and only if the following system has no solution	in variables $z \in \QQ$ and $y_j \in \QQ^m$ for $j=1,\dots,m$:
	\begin{align*}
		z + \sum_j A_{i,j} y_j &\geq 0 \quad \quad \text{(for $i=1,\dots,n$)}\\
		-y_j &\geq 0 \quad \quad \text{(for $j=1,\dots,m$)}\\
		z + \sum_j b_j y_j &< 0
	\end{align*}
	Equivalently, there are no $y_j' = -y_j \in \QQ_{\geq 0}^m$ such that
	$\sum_j A_{i,j} y_j' \leq z < \sum_j b_j y_j'$ (for $i=1,\dots,n$).
\end{proof}

\begin{proposition*}[Proposition~\ref{prop:overcast} restated]
	Let $\structA$ and $\structB$ be $\sigma$-structures.
	Then, $\structA \overcasts \structB$ if and only if there is an overcast from $\structA$ to $\structB$. 
\end{proposition*}

\begin{proof}
First, suppose that there exists an overcast $\omega$ from $\structA$ to $\structB$. 
Let $\structC$ be a $\sigma$-structure. Then, if $h$ is a maximum-value mapping from $\structB$ to $\structC$ we have
\begin{align*}
\opt{\structB,\structC} = \sum_{(f,\tuple{x})\in \tup{\structB}}f^{\structB}(\tuple{x})f^{\structC}(h(\tuple{x})) 
&\leq \sum_{(f,\tuple{x})\in \tup{\structB}} \left(\sum_{g \in B^A}\omega(g)f^{\structA}(g^{-1}(\tuple{x})) \right) f^{\structC}(h(\tuple{x}))\\
&= \sum_{g \in B^A}\omega(g) \left( \sum_{(f,\tuple{x})\in \tup{\structB}} f^{\structA}(g^{-1}(\tuple{x}))f^{\structC}(h(\tuple{x})) \right)\\
&= \sum_{g \in B^A}\omega(g) \left( \sum_{(f,\tuple{y})\in \tup{\structA}} f^{\structA}(\tuple{y}) f^{\structC}(h(g(\tuple{y}))) \right)
\end{align*}
and hence there exists $g \in B^A$ such that $\opt{\structB,\structC} \leq
  \sum_{(f,\tuple{y})\in \tup{\structA}} f^{\structA}(\tuple{y})
  f^{\structC}(h(g(\tuple{y})))=\val(h\circ g) \leq \opt{\structA,\structC}$.
  Therefore, $\structA \overcasts \structB$. For the converse implication, we
  shall use Lemma~\ref{lem:farkas1}.
If there is no overcast from $\structA$ to $\structB$,
this means there are no numbers $\omega(g) \in \QQ_{\geq 0}$ (for $g \in B^A$)
such that $\sum_g \omega(g)=1$ and $\sum_{g \in B^A} \omega(g) f^{\structA}(g^{-1}(\tuple{x})) \geq f^\structB(\tuple{x})$ for $(f,\tuple{x}) \in \tup{\structB}$.
By the lemma above, this is equivalent to the existence of $y(f,\tuple{x}) \in \QQ_{\geq 0}$ (for $(f,\tuple{x})\in \tup{\structB}$)
such that 
$$\sum_{(f,\tuple{x}) \in \tup{\structB}} f^{\structA}(g^{-1}(\tuple{x})) y(f,\tuple{x}) <
  \sum_{(f,\tuple{x}) \in \tup{\structB}}  f^\structB(\tuple{x}) y(f,\tuple{x})
  \quad\quad \text{for all }g \in B^A.$$

\noindent
Let $\structB_{\bar y}$ be the $\sigma$-structure with domain $B$ such that $f^{\structB_{\bar y}}(\tuple{x})=y(f,\tuple{x})$, for all $(f,\tuple{x})\in \tup{\structB}$ ($f^{\structB_{\bar y}}(\tuple{x})=0$ otherwise). 
By the above, $\opt{\structA,\structB_{\bar y}}<\opt{\structB,\structB_{\bar y}}$ and hence $\structA \mathbin{\not\overcasts} \structB$. 
\end{proof}

For the duality between $\eps$-thin distributions of modulators and weights
avoiding any $\eps$-small modulator, 
we use a variant of Farkas' lemma obtained by taking $-A$ and $-b$ in
Lemma~\ref{lem:farkas1}.

\begin{lemma}[Farkas' lemma, variant 2]\label{lem:farkas2}
	Let $A$ be an $m \times n$ rational matrix and let $\bar{b} \in \QQ^m$.
	Exactly one of the following holds:
	\begin{itemize}
		\item there are $x_i\in\QQ_{\geq 0}$ ($i=1,\dots,n$) such that $\sum_i x_i = 1$ and $\sum_i A_{i,j} x_i \leq b_j$ for $j=1,\dots,m$;
		\item there are $y_j\in\QQ_{\geq 0}$ ($j=1,\dots,m$) such that $\sum_j A_{i,j} y_j > \sum_j b_j y_j$ for $i=1,\dots,n$.	
	\end{itemize}
\end{lemma}

\begin{lemma*}[Lemma~\ref{lem:thin} restated]
	Let $\mathcal{F}$ be a family of subsets of a set $V$. The following are equivalent:
	\begin{itemize}
		\item there is an $\eps$-thin distribution $\pi$ of sets $X \in \mathcal{F}$ (i.e., for all $v \in V$, $\Pr_{X\sim \pi} [v \in X]  \leq \eps$);
		\item for all non-negative weights $\left(w(v)\right)_{v\in V}$, there is an $X \in \mathcal{F}$ such that $w(X) \leq \eps \cdot w(V)$.
	\end{itemize}
\end{lemma*}
\begin{proof}
	The first item is equivalent to the existence of numbers $\pi(X) \in \QQ_{\geq 0}$ for $X \in \mathcal{F}$ such that $\sum_X \pi(X) = 1$ and for all $v \in V$, $\sum_X [v \in X] \cdot \pi(X) \leq \eps$.
	By Lemma~\ref{lem:farkas2}, this holds if and only if there are no numbers $w(v) \in \QQ_{\geq 0}$ for $v \in V$ such that for all $X \in \mathcal{F}$, $\sum_v [v \in X] \cdot w(v) > \sum_v \eps \cdot w(v)$.
\end{proof}

\section{Proof of Proposition~\ref{prop:overcast-sa}}
\label{sec:overcast-sa}

Let SA$_k(\structA,\structC)$ denote the Sherali-Adams linear programming relaxation of $\HOM(\structA,\structC)$, given in Figure~\ref{fig:sa} in Section~\ref{sec:tract}.
Recall that $\optfrac{k}{\structA,\structC}$ denotes its optimum value and
we write $\structA \overcasts_k \structB$ if $\optfrac{k}{\structA,\structC} \geq \optfrac{k}{\structB,\structC}$ for all structures $\structC$ with the same signature as $\structA$ and $\structB$.

\begin{proposition*}[Proposition~\ref{prop:overcast-sa} restated]
	Let $\structA$ and $\structB$ be $\sigma$-structures and $k\geq\max_{f\in\sigma}\ar{f}$. 
	If there is an overcast from $\structA$ to $\structB$ then $\structA\overcasts_k \structB$. 
\end{proposition*}

\begin{proof}
Let $\structC$ be an arbitrary $\sigma$-structure, $\omega$ be an overcast from $\structA$ to $\structB$ and $\lambda$ be an optimal 
solution to SA$_k(\structB,\structC)$. (Recall that for a tuple $\tuple{x}$ we denote by $\toset{\tuple{x}}$ the set of elements appearing in $\tuple{x}$.)
We have that 
\begin{align*}
\optfrac{k}{\structB,\structC} &= \sum_{(f,\tuple{x}) \in \tup{\structB}, \, s:\toset{\tuple{x}} \to C} \lambda(\toset{\tuple{x}},s) f^{\structB}(\tuple{x}) f^{\structC}(s(\tuple{x}))\\
&\leq \sum_{(f,\tuple{x}) \in \tup{\structB},\, s:\toset{\tuple{x}} \to C} \left(\sum_{g \in B^A} \omega(g)f^{\structA}(g^{-1}(\tuple{x})) \right) \lambda(\toset{\tuple{x}},s) f^{\structC}(s(\tuple{x}))\\
&= \sum_{g \in B^A} \omega(g) \left( \sum_{(f,\tuple{x}) \in \tup{\structB}, \, s:\toset{\tuple{x}} \to C} \lambda(\toset{\tuple{x}},s) f^{\structA}(g^{-1}(\tuple{x})) f^{\structC}(s(\tuple{x})) \right)\\
&= \sum_{g \in B^A} \omega(g) \left( \sum_{(f,\tuple{y}) \in \tup{\structA}, \, s:g(\toset{\tuple{y}}) \to C} \lambda(g(\toset{\tuple{y}}),s) f^{\structA}(\tuple{y}) f^{\structC}(s(g(\tuple{y}))) \right)
\end{align*}
and hence there is $g:A\to B$ such that 
\begin{align}
\optfrac{k}{\structB,\structC} &\leq \sum_{(f,\tuple{y}) \in \tup{\structA}, \, s:g(\toset{\tuple{y}}) \to C} \lambda(g(\toset{\tuple{y}}),s) f^{\structA}(\tuple{y}) f^{\structC}(s(g(\tuple{y})))\label{eq:supp}
\end{align}

For $Y\in\binom{A}{\leq k}$ and $r: Y \to C$, we define
\[
\lambda'(Y,r)=
\begin{cases}
\lambda(g(Y),s) &\quad \text{if there exists $s:g(Y)\to C$ such that $s \circ g = r$}\\
0 &\quad \text{otherwise} 
\end{cases}
\]

Note that $\lambda'$ is a feasible solution of SA$_k(\structA,\structC)$. 
Indeed, for $Y\in\binom{A}{\leq k}$, we have
$$ \sum_{r: Y \to C}\lambda'(Y,r)=\sum_{s: g(Y) \to C}\lambda'(Y,s\circ g)=\sum_{s: g(Y) \to C}\lambda(g(Y),s)=1.$$
Moreover, let $Z\subseteq Y\in \binom{A}{\leq k}$, and $r: Z \to C$. 
If there is no $s:g(Z)\to C$ such that $s \circ g = r$, then 
$$ \lambda'(Z,r) = 0 = \sum_{t:Y \to C,\, t|_Z = r}\lambda'(Y,t).$$
If such a mapping $s$ exists, then 
\begin{align*}
\sum_{t:Y \to C,\, t|_Z = r}\lambda'(Y,t)
&= \sum_{s':g(Y) \to C,\, s'|_{g(Z)}=s} \lambda'(Y,s'\circ g)\\
&= \sum_{s':g(Y) \to C,\, s'|_{g(Z)}=s} \lambda(g(Y),s')\\
&= \lambda(g(Z),s)\\
&= \lambda'(Z,r). 
\end{align*}

Since $\lambda'$ is feasible and by (\ref{eq:supp}), we conclude that 
\begin{align*}
\optfrac{k}{\structA,\structC} &\geq \sum_{(f,\tuple{y}) \in \tup{\structA}, \, r:\toset{\tuple{y}} \to C} \lambda'(\toset{\tuple{y}},r) f^{\structA}(\tuple{y}) f^{\structC}(r(\tuple{y}))\\
&= \sum_{(f,\tuple{y}) \in \tup{\structA}, \, s:g(\toset{\tuple{y}}) \to C} \lambda'(\toset{\tuple{y}},s\circ g) f^{\structA}(\tuple{y}) f^{\structC}(s(g(\tuple{y})))\\
&= \sum_{(f,\tuple{y}) \in \tup{\structA}, \, s:g(\toset{\tuple{y}}) \to C} \lambda(g(\toset{\tuple{y}}),s) f^{\structA}(\tuple{y}) f^{\structC}(s(g(\tuple{y})))\\
&\geq \optfrac{k}{\structB,\structC}.
\end{align*}
\end{proof}

\end{document}